\documentclass{article}

\usepackage{microtype}
\usepackage{graphicx}
\usepackage{booktabs} 

\usepackage{amsmath,amsfonts,bm}

\def\floor#1{\lfloor #1 \rfloor}
\def\1{\bm{1}}

\DeclareMathAlphabet{\mathsfit}{\encodingdefault}{\sfdefault}{m}{sl}
\SetMathAlphabet{\mathsfit}{bold}{\encodingdefault}{\sfdefault}{bx}{n}

\newcommand{\E}{\mathbb{E}}

\DeclareMathOperator*{\argmax}{arg\,max}
\DeclareMathOperator*{\argmin}{arg\,min}

\usepackage{hyperref}
\usepackage{url}

\usepackage{multirow}

\usepackage{amsmath}
\usepackage{amsfonts}
\usepackage{amssymb}
\usepackage{amsthm}
\usepackage{enumerate}
\usepackage{newlfont}
\usepackage{color}
\usepackage{bbm}
\include{amsthm_sc}

\usepackage{stmaryrd}
\usepackage{mathrsfs}

\usepackage{fancyhdr}
\usepackage{graphicx}
\usepackage{fancybox}
\usepackage{setspace}
\usepackage[capitalize,noabbrev]{cleveref}
\usepackage{amssymb}

\usepackage{mathtools}

\usepackage{algorithm}
\usepackage{algorithmic}

\usepackage{makecell}
\usepackage{float}

\usepackage{subfig}
\usepackage{comment}
\usepackage{xr-hyper}
\usepackage{svg}

\usepackage[accepted]{icml2023}

\usepackage{amsmath}
\usepackage{amssymb}
\usepackage{mathtools}
\usepackage{amsthm}

\newtheorem{thm}{Theorem}
\newtheorem{cor}{Corollary}
\newtheorem{lem}{Lemma}
\newtheorem{prop}{Proposition}

\newtheorem{rem}{Remark}

\newcommand{\p} {\mathbb{P}}

\DeclareMathOperator{\var}{var}

\newcommand{\caE}{{\mathcal E}}

\newcommand{\caF}{{\mathcal F}}

\newcommand{\bsA}{{\boldsymbol A}}
\newcommand{\bsB}{{\boldsymbol B}}
\newcommand{\bsC}{{\boldsymbol C}}

\newcommand{\bsE}{{\boldsymbol E}}

\newcommand{\bsS}{{\boldsymbol S}}

\newcommand{\bsM}{{\boldsymbol M}}
\newcommand{\bsX}{{\boldsymbol X}}
\newcommand{\bsY}{{\boldsymbol Y}}

\newcommand{\bsU}{{\boldsymbol U}}
\newcommand{\bsV}{{\boldsymbol V}}

\newcommand{\bsZ}{{\boldsymbol Z}}

\newcommand{\bsv}{{\boldsymbol v}}

\newcommand{\bsx}{{\boldsymbol x}}

\newcommand{\bsp}{{\boldsymbol p}}
\newcommand{\bsq}{{\boldsymbol q}}
\newcommand{\bsr}{{\boldsymbol r}}
\newcommand{\bsu}{{\boldsymbol u}}
\newcommand{\bsg}{{\boldsymbol g}}
\newcommand{\bsh}{{\boldsymbol h}}

\newcommand{\beq}{ \begin{equation} }
\newcommand{\eeq}{ \end{equation} }

\renewcommand{\P}{\mathbb{P}}

\usepackage[textsize=tiny]{todonotes}
\renewcommand{\algorithmicrequire}{\textbf{Input:}}
\renewcommand{\algorithmicensure}{\textbf{Output:}}

\icmltitlerunning{Recovering Top-Two Answers and Confusion Probability in Multi-Choice Crowdsourcing}

\begin{document}

\twocolumn[
\icmltitle{Recovering Top-Two Answers and Confusion Probability\\ in Multi-Choice Crowdsourcing}



\icmlsetsymbol{equal}{*}

\begin{icmlauthorlist}
\icmlauthor{Hyeonsu Jeong}{yyy}
\icmlauthor{Hye Won Chung}{yyy}
\end{icmlauthorlist}

\icmlaffiliation{yyy}{School of Electrical Engineering, KAIST, Daejeon, Korea}

\icmlcorrespondingauthor{Hyeonsu Jeong}{hsjeong1121@kaist.ac.kr}
\icmlcorrespondingauthor{Hye Won Chung}{hwchung@kaist.ac.kr}

\icmlkeywords{Machine Learning, ICML}

\vskip 0.3in
]



\printAffiliationsAndNotice{}  

\begin{abstract}

Crowdsourcing has emerged as an effective platform for labeling large amounts of data in a cost- and time-efficient manner.
Most previous work has focused on designing an efficient algorithm to recover only the ground-truth labels of the data.
In this paper, we consider multi-choice crowdsourcing tasks with the goal of recovering not only the ground truth, but also the most confusing answer and the confusion probability. The most confusing answer provides useful information about the task by revealing the most plausible answer other than the ground truth and how plausible it is.
To theoretically analyze such scenarios, we propose a model in which there are the top two plausible answers for each task, distinguished from the rest of the choices. Task difficulty is quantified by the probability of confusion between the top two, and worker reliability is quantified by the probability of giving an answer among the top two.
Under this model, we propose a two-stage inference algorithm to infer both  the top two answers and the confusion probability. We show that our algorithm achieves the minimax optimal convergence rate. We conduct both synthetic and real data experiments and demonstrate that our algorithm outperforms other recent algorithms.  We also show the applicability of our algorithms in inferring the difficulty of tasks and in training neural networks with top-two soft labels. 


\end{abstract}

\section{Introduction}\label{sec:intro}

Crowdsourcing has been widely adopted to solve a large number of tasks in a time- and cost-efficient manner with the help of human workers. 
In this paper, we consider multiple-choice tasks, where a worker is asked to provide a single answer among multiple choices.
Some examples are as follows:
1) Using crowdsourcing platforms such as MTurk, we solve object counting or classification tasks on a large collection of images.
Answers can be noisy either due to the difficulty of the scene or due to unreliable workers making random guesses.
2) Scores are collected from referees for papers submitted to a conference. For certain papers, scores can vary widely among reviewers, either due to the inherent nature of the paper (clear pros and cons) or due to the reviewer's subjective interpretation of the scoring scale \citep{stelmakh2019peerreview4all,liu2022integrating}.

In the above scenarios, the answers provided by human workers may not be consistent among themselves, not only due to the presence of unreliable workers, but also due to the inherent difficulty of the tasks. 
In particular, for multiple choice tasks, there can exist plausible answers other than the ground truth, which we call \emph{confusing answers}.\footnote{This phenomenon is evident on public datasets: for `Web' dataset \citep{zhou2012learning}, which consists of five-choice tasks, the most dominant top-two answers of each task account for  80\% of the total answers, and the ratio between the top two is 2.4:1.  }
For tasks with confusing answers, even reliable workers may provide wrong answers due to confusion. Thus, we need to decompose the two different causes of wrong answers: low reliability of workers and confusion due to task difficulty. 



However, most previous models of multi-choice crowdsourcing do not adequately model the errors from confusion. For example, the single-coin Dawid-Skene model \citep{DS}, which is the most widely studied model in the literature, assumes that a worker is associated with a single skill parameter that is fixed across all tasks, which models the probability of giving a correct answer for every task. Under this model, any algorithm that infers the worker's skill would count a confused labeling as the worker's error and lower its accuracy estimate for the worker, resulting in a wrong estimate of the worker's true skill level.
 
To model the effect of confusion in multi-choice crowdsourcing problems, we propose a new model in which each task can have a confusing answer other than the ground truth, with a different confusion probability across tasks. Task difficulty is quantified by the confusion probability between the top two plausible answers, and worker skill is modeled by the probability of giving an answer among the top two, to distinguish reliable workers from pure spammers who just give random guesses among possible choices.
We justify the proposed top-two model with public datasets. Under this new model, we aim to recover both the ground truth and the most confusing answer with the confusion probability, which indicates how plausible the recovered ground truth is compared to the most confusing answer. 

We provide an efficient two-stage inference algorithm to recover the top-two plausible answers and the confusion probability. The first stage of our algorithm uses the spectral method to obtain an initial estimate for the top-two answers and the confusion probability, and the second stage uses this initial estimate to estimate the workers' reliabilities and to refine the estimates for the top-two answers. Our algorithm achieves the minimax optimal convergence rate. We then perform experiments comparing our method to recent crowdsourcing algorithms on both synthetic and real datasets, and show that our method outperforms other methods in recovering top-two answers. This result demonstrates that our model better explains the real-world datasets, including errors due to confusion. Our code is available at \url{https://github.com/Hyeonsu-Jeong/TopTwo}.
Our main contributions can be summarized as follows. 
\begin{itemize}
\item {\it Top-two model:} We propose a new model for multi-choice crowdsourcing tasks, where each task has top-two answers, and the difficulty of the task is quantified by the confusion probability between the top two plausible answers. We justify the proposed model by analyzing six public datasets and showing that the top-two structure explains the real datasets well. 
\item {\it Inference algorithm and its application:} We propose a two-stage algorithm that recovers the top-two answers and the confusion probability of each task at the minimax optimal convergence rate. We demonstrate the potential applications of our algorithm not only in crowdsourced labeling, but also in two more applications: (i) quantifying task difficulty, and (ii) training neural networks for classification with soft labels that include the top-two information and the task difficulty. 
\end{itemize}

\section{Related works}\label{sec:related}

\paragraph{Dawid-Skene(D\&S) model.}

In crowdsourcing \citep{welinder2010multidimensional,DBLP:conf/icml/LiuW12,demartini2012zencrowd,aydin2014crowdsourcing,demartini2012zencrowd}, one of the most widely studied models is the Dawid-Skene (D\&S) model \citep{DS}. In this model, each worker is associated with a single confusion matrix, fixed across tasks, that models the probability of giving a label $b\in[K]$ for the true label $a\in[K]$ for a $K$-ary classification task. In the single-coin D\&S model, the model is further simplified such that each worker has a fixed skill level regardless of the true label or task. 
Under the D\&S model, various methods have been proposed to estimate the confusion matrix or skill of workers by spectral methods \citep{SEM,EoR, GhostSVD, karger2013efficient},  iterative algorithms \citep{Karger, karger2011iterative,li2014error, liu2012variational,ok2016optimality}, or rank-1 matrix completion \citep{PGD,M-MSR,ibrahim2019crowdsourcing}. The estimated skill can be used to infer the ground truth answer by approximating maximum likelihood (ML)-type estimators \citep{gao2013minimax,gao2016exact,SEM,karger2013efficient,li2014error,raykar2010learning,smyth1994inferring,ipeirotis2010quality,berend2014consistency}.
Unlike the D\&S models, our model allows each worker to make errors over tasks with different probabilities due to confusion. 
Thus, our algorithm needs to estimate not only the worker's skill, but also the task difficulty. Since the number of tasks is often much larger than the number of workers in practice, estimating task difficulty is much more challenging than estimating worker skill. We provide a statistically-efficient algorithm to estimate the task difficulty and use this estimate to infer the top-two answers.

\paragraph{Task Dificculty.}

We also remark that there have been some recent attempts to model task difficulty in crowdsourcing \citep{khetan2016achieving, shah2020permutation, krivosheev2020detecting, shah2018reducing, bachrach2012grade, li2019exploiting, tian2015max}. However, these works are either restricted to binary tasks \citep{khetan2016achieving, shah2020permutation, shah2018reducing} or focus on grouping confusable classes \citep{krivosheev2020detecting, li2019exploiting, tian2015max}.
Our result, on the other hand, applies to any set of multi-choice tasks of which the choices are not necessarily restricted to a fixed set of classes/labels.
\paragraph{Modeling confusion.}

There is a growing interest in the machine learning community in utilizing soft labels to train deep neural networks.  Various methods have been proposed to generate soft labels of training data, e.g., by using mixup \citep{zhang2017mixup, sohn2022genlabel} or by using the output of trained models \cite{sabetpour2021truth}. Also, CIFAR-10H \cite{peterson2019human} dataset was generated by using all the human annotations  from the data collection step as soft labels on images. Our algorithm estimates the task difficulty and the top two answers, which can produce a new form of soft label that can be used in this line of work, as will be discussed in Sec. \ref{sec:NN_soft}.

\paragraph{Notation.} 
For a vector $\bsx$, $x_i$ represents the $i$-th component of $\bsx$.
For a matrix $\bsM$, $M_{ij}$ refers to the $(i,j)$th entry of $\bsM$. 
For any vector $\bsx$, its $\ell_2$ and $\ell_\infty$-norm are denoted by $\|\bsx\|_2$ and $\|\bsx\|_\infty$, respectively. 
We follow the standard definitions of asymptotic notations, $\Theta(\cdot)$, $O(\cdot)$, $o(\cdot)$, and $\Omega(\cdot)$.

\section{Model and Problem Setup}\label{sec:model}

We consider a crowdsourcing model to infer the top two most plausible answers among $K$ choices for each task. 
There are $n$ workers and $m$ tasks. 
For each task $j\in[m]:=\{1,\dots,m\}$, we denote the correct answer by $g_j\in[K]$ and the next plausible or the most confusing answer by $h_j\in[K]$.
We call the pair $(g_j,h_j)$ the top two answers for the task $j\in[m]$.
Let $\bsp\in[0, 1]^n$ and $\bsq\in(1/2,1]^m$ be parameters modeling the reliability of the workers and the difficulty of the tasks, respectively.
For each pair of $(i,j)$, the $j$-th task is assigned to the $i$-th worker independently with probability $s$.
We use a matrix $\bsA\in \mathbb{R}^{n\times m}$ to represent the responses of the workers, where $A_{ij}=0$ if the $j$-th task is not assigned to the $i$-th worker, and if it is assigned, $A_{ij}$ is equal to the received label. The distribution of $A_{ij}$ is determined by the worker reliability $p_i$ and the task difficulty $q_j$ as follows:
\beq\label{eqn:A_dist}
A_{ij}=\begin{cases}
g_j,& \text{w.p. } s\left(p_i q_j+\frac{1-p_i}{K}\right),\\
h_j,& \text{w.p. } s\left(p_i (1-q_j) +\frac{1-p_i}{K}\right),\\
b\in[K]\backslash\{g_j,h_j\}, &\text{w.p. } s\left(\frac{1-p_i}{K}\right),\\
0,&\text{w.p. } 1-s.
\end{cases}
\eeq
Here $p_i$ is the reliability of the $i$-th worker in giving the answer from the most plausible top two $(g_j, h_j)$.
If $p_i=0$, the worker is considered a spammer, giving random answers among $K$ choices, and a larger value of $p_i$ indicates a higher reliability of the worker.
The parameter $q_j$ represents the inherent difficulty of the task $j$ in discriminating between the top two answers: for an easy task, $q_j$ is closer to 1, and for a hard task, $q_j$ is closer to 1/2. 
We call $q_j$ the confusion probability.
Our goal is to recover the top two answers $(g_j,h_j)$ for all $j\in[m]$ with high probability at the minimum possible sampling probability $s$.
We assume that the model parameters $(\bsp,\bsq)$ are unknown.

We propose the top-two model to reflect common characteristics of public crowdsourcing datasets, as outlined in Appx. \S\ref{sec:app:model}. The most important observation is that the top-two answers dominate the overall answers, and only the second dominant answer has an incidence rate comparable to the ground truth. In other words, the incidence rate of the second answer falls within the one-sigma range of the ground truth, indicating a significant overlap. However, such an overlap is not observed with the third or fourth answers. 
This suggests that the assumption of a unique ``confusing answer'' is adequate to model confusion due to task difficulty. More details can be found in Appx. \S\ref{sec:app:model}.

\paragraph{Binary conversion.}
We provide the main observation on the structure of $\mathbf{A}$, which will be used to design algorithms for estimating the top two plausible answers $(g_j,h_j)$ and the confusion probability $q_j$ for $j\in[m]$.
The $K$-ary task can be decomposed into $(K-1)$-binary tasks as follows  \citep{karger2013efficient}: define $\bsA^{(k)}$ for $1\leq k<K$ such that the $(i,j)$-th entry $A_{ij}^{(k)}$ indicates whether the original answer $A_{ij}$ is greater than $k$ or not, i.e., $A^{(k)}_{ij}=-1$ if $1\leq  A_{ij}\leq k$; $A^{(k)}_{ij}=1$ if $k<A_{ij}\leq K$; and $A^{(k)}_{ij}=0$ if $A_{ij}=0$.
We show that $\E[\bsA^{(k)}]$ is a rank-1 matrix and the singular value decomposition (SVD) of $\E[\bsA^{(k)}]$ can reveal the top-two answers $\{(g_j,h_j)\}_{j=1}^m$ and the confusion probability vector $\bsq$.
\begin{prop}\label{prop:binary} For each $1\leq k<K$, the binary mapped matrix $\bsA^{(k)}\in\{-1,0,1\}^{n\times m}$ satisfies
$
\E[\bsA^{(k)}]-\frac{s(K-2k)}{K}\mathbbm{1}_{n \times m}=2s\bsp (\bsr^{(k)})^\top,
$
where $\bsr^{(k)}=[r_1^{(k)}\cdots r_m^{(k)}]^\top$ is defined as
\begin{equation}\nonumber\label{eqn:defr}
\begin{split}
&\text{Case I: $g_j>h_j$} \\
&r_j^{(k)}:=\begin{cases}
\frac{k}{K}& \text{ where }k<h_j ; \\
\frac{k}{K}-(1-q_j)& \text{ where }h_j\leq k< g_j;  \\
\frac{k}{K}-1& \text{ where }g_j \leq k,
\end{cases} \\
&\text{Case II: $g_j<h_j$} \\
&r_j^{(k)}:=\begin{cases}
\frac{k}{K} &\text{ where }k<g_j;\\
\frac{k}{K}-q_j &\text{ where }g_j\leq k< h_j;\\
\frac{k}{K}-1& \text{ where }h_j \leq k.
\end{cases}
\end{split}
\end{equation}

\end{prop}
The proof of Proposition \ref{prop:binary} is available in Appendix \S\ref{app:prop:binary}.
By defining $\Delta r_{j}^{(k)}:=r_j^{(k)}-r_j^{(k-1)}$ for $k\in[K]$ with $r_{j}^{(0)}:=0$ and $r_{j}^{(K)}:=0$ for all $j$, we have
\beq\label{eqn:delta_r}
\Delta r_j^{(k)}=\begin{cases}
\frac{1}{K}-q_j&\text{ where }k=g_j,\\
\frac{1}{K}-(1-q_j)&\text{ where }k=h_j,\\
\frac{1}{K}&\text{ otherwise}.
\end{cases}
\eeq
Note that $\Delta r_{j}^{(k)}$ has its minimum at $k=g_j$ and its second smallest value at $k=h_j$ for $q_j\in (1/2,1]$.
If one can specify $g_j$, the task difficulty $q_j$ can also be found from $\frac{1}{K}-\Delta r_j^{(g_j)}$. 
In the next section, we will use this structure of $\bsr^{(k)}$ for $k\in[K]$ to infer the top two answers and the confusion probability. \footnote{We assume that
$
    \eta \sqrt{n} \le \|\bsp\|_2 \le \sqrt{n}
$ for some $\eta>0$, which ensures that there are only $o(n)$ spammers ($p_i=0$). We also assume that $\|\bsr^{(k)}\|_2 = \Theta(\sqrt{m})$ for every $k\in[K]$, which can be easily satisfied except for exceptional cases from \eqref{eqn:defr}.} 

\begin{rem}\normalfont{Our top-two model can be generalized to top-$T$ ($T\ge2$) model, where it is assumed that for each task there are $T(\leq K)$ plausible answers with the associated confusion probabilities with respect to the ground truth. Even for this generalized model, we can define binary-converted observation matrices $\bsA^{(k)}$, $k\in[K]$, which enjoy the rank-1 structure, and prove the results similar to Proposition~\ref{prop:binary}, showing that the top $T$ plausible answers and the associated confusion probabilities can be inferred using the rank-1 structure. More details can be found in the Appendix~\ref{subsec:top_T}. }
\end{rem}




\section{Proposed Algorithm}\label{sec:algorithm}

In this section, we present an algorithm to estimate the top two answers $\{(g_j,h_j)\}_{j=1}^m$ and the confusion probability vector $\bsq$. 
Our algorithm consists of two stages. In Stage 1, we compute an initial estimate of the top two answers and the confusion probability $\bsq$. 
In Stage 2, we estimate the worker reliability vector $\bsp$ by using the result of the first stage, and use the estimated $\bsp$ and $\bsq$ to refine our estimates for the top two answers. 
We randomly split the entires of the original response matrix $\bsA\in \mathbb{R}^{n\times m}$ into $\bsA^1\in \mathbb{R}^{n\times m}$ and $\bsA^2\in \mathbb{R}^{n\times m}$ with probability $s_1$ and $1-s_1$, respectively, and use only $\bsA^1$ for stage 1 and $(\bsA^1,\bsA^2)$ for stage 2. 

\subsection{Stage 1: Initial estimates using SVD}\label{sec:stage1}\label{sec:st1}

The first stage of our algorithm is presented in Algorithm \ref{alg:top_two}.
In this stage, we use the data matrix $\bsA^1\in \mathbb{R}^{n\times m}$  to estimate the left singular vector  $\bsp^*:={\bsp}/{\|\bsp\|_2}$ and the scaled right singular vector $\|\bsp\|_2\bsr^{(k)}$ of $\E[\bsA^{(k)}]$ for all $k\in[K]$, which are then used to infer both the top two answers and the confusion probability by using \eqref{eqn:delta_r}.

The first stage begins with randomly splitting the entries of $\bsA^1\in \mathbb{R}^{n\times m}$ again into two independent matrices $\bsB\in \mathbb{R}^{n\times m}$ and $\bsC\in \mathbb{R}^{n\times m}$ with equal probabilities. 
We then convert $\bsB$ and $\bsC$ into $(K-1)$-binary matrices $\bsB^{(k)}$ and $\bsC^{(k)}$ for $1\leq k< K$, defined as $B^{(k)}_{ij}=-1$ if $1\leq  B_{ij}\leq k$; $B^{(k)}_{ij}=1$ if $k<B_{ij}\leq K$; and $B^{(k)}_{ij}=0$ if $B_{ij}=0$, and similarly for $\bsC^{(k)}$.
Define $\bsX^{(k)}$ and $\bsY^{(k)}$ as 
$ \bsX^{(k)} := \bsB^{(k)} - \frac{s'(K-2k)}{K}\mathbbm{1}_{n \times m}$ and $\bsY^{(k)} := \bsC^{(k)} - \frac{s'(K-2k)}{K}\mathbbm{1}_{n \times m}$ for $s'=s\cdot s_1/2$.
We have
$
\E[\bsX^{(k)}]=\E[\bsY^{(k)}]=s'\bsp (\bsr^{(k)})^\top
$ from Proposition \ref{prop:binary}.

We use $\bsX^{(k)}$ and $\bsY^{(k)}$ to estimate $\bsp^*:={\bsp}/{\|\bsp\|_2}$ and $\|\bsp\|_2\bsr^{(k)}$, respectively. 
The estimators are denoted by $\bsu^{(k)}$ and $\bsv^{(k)}$, respectively. We define $\bsu^{(k)}$ as the left singular vector of $\bsX^{(k)}$ with the largest singular value. 
Sign ambiguity of the singular vector is resolved by defining $\bsu^{(k)}$ as the one between $\{\bsu^{(k)},-\bsu^{(k)}\}$ in which at least half of the entries are positive. 
After trimming abnormally large components of $\bsu^{(k)}$ and defining the trimmed vector as $\tilde{\bsu}^{(k)}$, we calculate $\bsv^{(k)}:=\frac{1}{s'}(\bsY^{(k)})^{\top}\tilde{\bsu}^{(k)}$, which is an estimate for $\|\bsp\|_2\bsr^{(k)}$. By using $\bsv^{(k)}$ for $1\leq k<K$, we get estimates for top-two answers $(\hat{g}_j,\hat{h}_j)$ as in \eqref{alg:gh} by using the relation in \eqref{eqn:delta_r} . 
Lastly, we estimate $\|\bsp\|_2$ and use $\bsv^{(k)}/\|\bsp\|_2\approx \bsr^{(k)}$ to estimate the confusion probability vector $\bsq$ as in \eqref{eqn:est_q}.

\begin{algorithm}[t]
	\caption{Spectral Method for Initial Estimation (TopTwo1 Algorithm) }
	\label{alg:top_two}
\begin{algorithmic}[1]
  \REQUIRE{data matrix $\bsA^1\in\{0,1,\dots,K\}^{n\times m}$ and parameter $\eta>0$ where $\eta \sqrt{n} \le \|\bsp\|_2 \le \sqrt{n}$.}
  \ENSURE{estimated top-two answers $\{(\hat{g}_j,\hat{h}_j)\}_{j=1}^m$ and confusion probability vector $\hat{\bsq}$.}
	\STATE Randomly split (with equal probabilities) $\bsA^1$ into $\bsB$ and $\bsC$, and convert the two matrices into binary matrices $\bsX^{(k)}\in\{-1,0,1\}^{n\times m}$ and $\bsY^{(k)}\in\{-1,0,1\}^{n\times m}$ for $1\leq k<K$, respectively, as described in Sec. \ref{sec:stage1}. 
	\STATE  Let $\bsu^{(k)}$ be the leading normalized left singular vector of $\bsX^{(k)}$. 
Trim the abnormally large components of $\bsu^{(k)}$ by setting them to zero if $u_i^{(k)}>\frac{2}{\eta\sqrt{n}}$ and denote the resulting vector as  $\tilde{\bsu}^{(k)}$.
	\STATE Calculate the estimate of $\|\bsp\|\bsr^{(k)}$ by defining
$
\bsv^{(k)}:=\frac{1}{s'}(\bsY^{(k)})^{\top}\tilde{\bsu}^{(k)}.
$
Assume $\bsv^{(0)}:=\boldsymbol{0}$ and $\bsv^{(K)}:=\boldsymbol{0}$.
	\STATE For $k\in[K]$, calculate
$
\Delta v_j^{(k)}:= v_j^{(k)}-v_j^{(k-1)}.
$
Estimate the top-two answers for $j\in[m]$ by 
\beq
\begin{split}\label{alg:gh}
\hat{g}_j:= \argmin_{k\in[K]}\Delta v_j^{(k)};\quad \hat{h}_j:=\argmin_{k\neq \hat{g}_j, k\in[K]} \Delta v_j^{(k)}.
\end{split}
\eeq

	\STATE Estimate $\|\bsp\|_2$ by
${l}_j:= \frac{K}{K-2} \sum_{k\neq \hat{g}_j, k\neq \hat{h}_j} \Delta v_j^{(k)}$ and ${l}:=\frac{1}{m}\sum_{j=1}^m l_j$.

	\STATE Estimate ${q}_j$ for $j\in[m]$ by defining
\beq\label{eqn:est_q}
\hat{q}_j:= {1}/{K}-{\Delta v_{j}^{(\hat{g}_j)}}/{l}.
\eeq

\end{algorithmic}
\end{algorithm}

\subsection{Stage 2: Plug-in Maximum Likelihood Estimator (MLE)}\label{sec:st2}

The second stage uses the result of Stage 1 to estimate the worker reliability vector $\bsp$. 
Remind that we randomly split the original response matrix $\bsA$ into $\bsA^1$ and $\bsA^2$ with probability $s_1$ and $1-s_1$, respectively, and use $\bsA^1$ only for Alg. \ref{alg:top_two}.  Thus, the estimated top-two answers $\{(\hat{g}_j,\hat{h}_j)\}_{j=1}^m$ from Alg. \ref{alg:top_two} depend only on $\bsA^1$.  We then define the estimator $\hat{\bsp}$ for worker reliability by comparing the unused data matrix $\bsA^2$ with the estimated top two answers $\{(\hat{g}_j,\hat{h}_j)\}_{j=1}^m$ as
\beq\label{eqn:hatpi}
\hat{p}_i=\frac{K}{(K-2)}\left(\frac{1}{ms(1-s_1)}\sum_{j=1}^m \mathbbm{1}(A^2_{ij}=\hat{g}_j\text{ or } \hat{h}_j)-\frac{2}{K}\right).
\eeq

The final step refines the estimates for the top two answers by using the plug-in MLE where the estimated $(\hat{\bsp},\hat{\bsq})$ are placed in $(\bsp,\bsq)$ at the oracle MLE, which finds $(\hat{g}_j,\hat{h}_j)\in[K]^2\backslash \{(1,1),(1,2),\dots,(K,K)\}$ such that  $(\hat{g}_j,\hat{h}_j):=\argmax_{(a,b)\in[K]^2, a\neq b}\sum_{i=1}^n \log \mathbb{P}(A_{ij} \vert \bsp, q_j, (a,b))$ as in \eqref{eqn:plug-in-mle}. Our complete algorithm is presented in Algorithm \ref{alg:plugin_MLE}.

\begin{algorithm}[t]
	\caption{Plug-in MLE (TopTwo2 Algorithm) }
	\label{alg:plugin_MLE}
\begin{algorithmic}[1]
\REQUIRE{data matrix $\bsA\in\{0,1,\dots,K\}^{n\times m}$ and the sample splitting rate $s_1>0$.}
\ENSURE{estimated top two answers $\{(\hat{g}^{\mathsf{MLE}}_j,\hat{h}^{\mathsf{MLE}}_j)\}_{j=1}^m$ and confusion probability vector $\hat{\bsq}$.}
	\STATE  Randomly split $\bsA$ into $\bsA^1$ and $\bsA^2$ by defining $\bsA^1:=\bsA \circ \bsS$ and $\bsA^2=\bsA \circ (\mathbbm{1}_{n \times m}-\bsS)$ 
	where  $\bsS$ is an $n \times m$ matrix whose entries are i.i.d. with Bernoulli($s_1$) and $\circ$ is an entrywise product.
	\STATE Apply Algorithm \ref{alg:top_two} to $\bsA^1$ to yield estimates for top-two answers $\{(\hat{g}_j,\hat{h}_j)\}_{j=1}^m$ and confusion probability vector $\hat{\bsq}$.
	\STATE By using $\{(\hat{g}_j,\hat{h}_j)\}_{j=1}^m$ and $\bsA^2$, calculate the estimate for worker reliability vector $\hat{\bsp}$ as in \eqref{eqn:hatpi}.
	\STATE By using the whole $\bsA$ and $(\hat{\bsp},\hat{\bsq})$, find the plug-in MLE estimates $(\hat{g}^{\mathsf{MLE}}_j, \hat{h}^{\mathsf{MLE}}_j)$ by
	\beq\label{eqn:plug-in-mle}
        \begin{split}
        \argmax_{a,b\in[K]^2, a\neq b} \sum_{i=1}^n &\log \left(\frac{K\hat{p}_i \hat{q}_j}{1-\hat{p}_i} + 1\right) \mathbbm{1}(A_{ij} = a) \\
        &+ \log \left(\frac{K\hat{p}_i (1-\hat{q}_j)}{1-\hat{p}_i}  + 1\right) \mathbbm{1}(A_{ij} = b).
        \end{split}
	\eeq
\end{algorithmic}
\end{algorithm}

The time complexity of Alg. \ref{alg:plugin_MLE} is $O(m^2\log m+ nmK^2)$, since the SVD in Alg. \ref{alg:top_two} can be computed via power iterations within $O(m^2\log m)$ steps \citep{boutsidis2015spectral}, and the step for finding the pair of answers maximizing \eqref{eqn:plug-in-mle} requires $O(nmK^2)$ steps.

\section{Performance Analysis}\label{sec:analysis}

To state our main theoretical results, we first need to introduce some notation and assumptions.
Let $\mu^{(i,j)}_{(a,b),k}$ denote the probability that a worker $i\in[n]$ gives the label $k\in[K]$ for the assigned task $j\in[m]$, whose top two answers are $(g_j,h_j)=(a,b)$.
Note that $\mu^{(i,j)}_{(a,b),k}$  can be written in terms of $(p_i,q_j)$ from the distribution in \eqref{eqn:A_dist} for every $a,b,k\in[K]^3$.
Let $\boldsymbol{\mu}^{(i,j)}_{(a,b)}=[\mu^{(i,j)}_{(a,b),1}\;\; \mu^{(i,j)}_{(a,b),2}\;\;\cdots\;\;\mu^{(i,j)}_{(a,b),K}]^\top$.
We introduce a quantity that measures the average ability of workers in distinguishing the ground-truth pair of top-two answers $(g_j,h_j)$ from any other pair $(a,b)\in[K]^2/\{(g_j,h_j)\}$ for the task $j\in[m]$.
We define
\begin{align}
\overline{D}^{(j)}&:=\min_{(g_j,h_j)\neq (a,b)}\frac{1}{n}\sum_{i=1}^n  \mathbb{D}_\mathsf{KL}\left(\boldsymbol{\mu}^{(i,j)}_{(g_j,h_j)},\boldsymbol{\mu}^{(i,j)}_{(a,b)}\right); \\
\overline{D}&:=\min_{j\in[m]}\overline{D}^{(j)},
\end{align}
where $ \mathbb{D}_\mathsf{KL}(P,Q):=\sum_i P(i)\log(P(i)/Q(i))$ is the KL-divergence between $P$ and $Q$. 
Note that $\overline{D}^{(j)}$ is strictly positive if $q_j\in(1/2,1)$  and there exists at least one worker $i$ with $p_i>0$, so that $(g_j,h_j)$ can be distinguished from any other $(a,b)\in[K]^2/\{(g_j,h_j)\}$ statistically in \eqref{eqn:A_dist}. We define $\overline{D}$ as the minimum of $\overline{D}^{(j)}$ over $j\in[m]$, indicating the average ability of workers in distinguishing $(g_j,h_j)$ from any other $(a,b)$ for the most difficult task in the set of tasks.

Theorem \ref{thm:main} states the performance guarantees for Algorithm~\ref{alg:top_two} by providing sufficient conditions for achieving an arbitrarily accurate estimation of the top-two answers and the confusion probability. 

\begin{thm}[Performance Guarantees for Algorithm~\ref{alg:top_two}]\label{thm:main}
For any $\epsilon,\delta_1>0$, if the sampling probability $s\cdot s_1= \Omega\left(\frac{1}{\delta_1^2 \|\bsp\|_2^2}\log\frac{K}{\epsilon}\right)$, Algorithm~\ref{alg:top_two} guarantees the recovery of the ordered top-two answers $(g_j,h_j)$ with probability at least $1-\epsilon$ for any task $j\in[m]$ having $q_j\in(1/2,1)$, i.e., 
\beq\label{eqn:rec_gh_main}
\P\left((\hat{g}_j,\hat{h}_j)=(g_j,h_j)\right)\geq1-\epsilon \quad \text{for all }j\in[m], 
\eeq
and also guarantees the recovery of the confusion probability $q_j$ with
\beq\label{eqn:rec_q_main}
\P\left(|\hat{q}_j-{q}_j|<\delta_1\right)\geq 1-\epsilon\quad  \text{for all }j\in[m],
\eeq
where the number $m$ of tasks is sufficiently large and the number of workers scales as $n=O(m/\log m)$.
\end{thm}
For a task $j$ with $q_j=1$, it is impossible to recover $h_j$, since $h_j$ cannot be distinguished from the rest of wrong labels $c\in[K]\backslash{\{g_j\}}$ statistically from \eqref{eqn:A_dist}. For such tasks, we can still guarantee the recovery of $g_j$ with accuracy $\P\left(\hat{g}_j=g_j\right)\geq1-\epsilon$ under the conditions in Theorem \ref{thm:main}.
By using Theorem~\ref{thm:main}, we can also find the sufficient conditions to guarantee the recovery of paired top-two answers for \textit{all} tasks and $\bsq$ with an arbitrarily small $\ell_\infty$-norm error with probability at least $1-\epsilon$. 

\begin{cor}\label{cor:exact_alg1} For any $\epsilon,\delta_1>0$, if the sampling probability $s\cdot s_1= \Omega\left(\frac{1}{\delta_1^2 \|\bsp\|_2^2}\log\frac{mK}{\epsilon}\right)$, Algorithm~\ref{alg:top_two} guarantees the recovery of $\{(g_j,h_j)\}_{j=1}^m$ and $\bsq$ with probability at least $1-\epsilon$ as $m\to\infty$ such that
\begin{align}
&\P\left((\hat{g}_j,\hat{h}_j)=(g_j,h_j),\forall j\in[m]\right)\geq1-\epsilon, \text{ and} \\
&\P\left(\|\bsq-\hat{\bsq}\|_\infty<\delta_1\right)\geq 1-\epsilon.
\end{align}
\end{cor}
Proofs of Thm. \ref{thm:main} and Cor. \ref{cor:exact_alg1} are available in Appendix \S\ref{app:thm:main}.

We next analyze the performance of Algorithm \ref{alg:plugin_MLE}, which uses Algorithm~\ref{alg:top_two} as the first stage.
Before providing the main theorem for Algorithm \ref{alg:plugin_MLE}, we state a lemma charactering a sufficient condition for estimating the worker reliability vector $\bsp$ from \eqref{eqn:hatpi}  with an arbitrarily small $\ell_\infty$ error.

\begin{lem}\label{lem:rec_p} Conditioned on $(\hat{g}_j,\hat{h}_j)=(g_j,h_j)$ for all $j\in[m]$, if $s(1-s_1)=\Omega\left(\frac{1}{\delta_2^2 m}\log\frac{n}{\epsilon}\right)$, the estimator $\hat{p}_i$ defined in \eqref{eqn:hatpi} of Algorithm \ref{alg:plugin_MLE} guarantees 
$
\P\left(\|\bsp-\hat{\bsp}\|_\infty<\delta_2\right)\geq1-\epsilon
$
for any $\epsilon>0$.
\end{lem}
Combining Corollary \ref{cor:exact_alg1} and Lemma \ref{lem:rec_p}, we can obtain the estimators $(\hat{\bsp},\hat{\bsq})$ of the worker reliability vector $\bsp$ and the confusion probability vector $\bsq$, respectively, with $\ell_\infty$-norm error bounded by any arbitrarily small $\delta>0$ with probability at least $1-2\epsilon$ if 
\beq\label{eqn:total_s_pq}
\begin{split}
s=s\cdot s_1+s(1-s_1)&=\Omega\left(\frac{\log(mK/\epsilon)}{\delta^2 \|\bsp\|_2^2}+\frac{\log(n/\epsilon)}{\delta^2 m}\right)\\
&=\Omega\left(\frac{\log(mK/\epsilon)}{\delta^2 \|\bsp\|_2^2}\right)
\end{split}
\eeq
where the last equality is from the assumption that $\|\bsp\|_2=\Theta(\sqrt{n})$ and $n=O(m/\log m)$. 
In this regime, the sample complexity for estimating the task difficulty $\bsq$ is greater than that required for estimating worker reliability $\bsp$.
To make the sampling probability $s<1$, we need $n=\Omega(\log m)$.

Our second theorem characterizes the sufficient condition on the sampling probability $s$ to guarantee the recovery of the pair of top-two answers for all tasks by~\eqref{eqn:plug-in-mle} of Alg. \ref{alg:plugin_MLE}, when a sufficiently accurate estimation of $(\bsp,\bsq)$ is provided.

\begin{thm}\label{thm:plugmin_mle}
Assume that there is a positive scalar $\rho$ such that $\mu^{(i,j)}_{(g_j,h_j),c}\geq \rho$ for all $(i,j,g_j,h_j,c)\in[n]\times[m]\times [K]^3$.
For any $\epsilon>0$, if $(\hat{\bsp},\hat{\bsq})$ are given with
\beq\label{eqn:cond_delta}
\max\{\|\bsp-\hat{\bsp}\|_\infty, \|\bsq-\hat{\bsq}\|_\infty\}\leq \delta:=\min\left\{\frac{\rho}{4},\frac{\rho\overline{D}}{4(6+\overline{D})}\right\},
\eeq
and the sampling probability satisfies
$$s=\Omega\left(\frac{\log(1/\rho)\log(mK^2/\epsilon)+\log(m/\epsilon)}{n\overline{D}}\right),$$then for any $\epsilon>0$ the estimates of $\{(g_j,h_j)\}_{j=1}^m$ from \eqref{eqn:plug-in-mle} of Algorithm~\ref{alg:plugin_MLE} guarantees
\beq
\P\left((\hat{g}^{\mathsf{MLE}}_j,\hat{h}^{\mathsf{MLE}}_j)=(g_j,h_j),\forall j\in[m]\right)\geq1-\epsilon.
\eeq
\end{thm}
Proofs of Lemma \ref{lem:rec_p} and Theorem \ref{thm:plugmin_mle}  are available in Appendix \S\ref{app:thm:plugin_MLE}.
The assumption in Theorem \ref{thm:plugmin_mle} that $\mu^{(i,j)}_{(g_j,h_j),c}\geq \rho$ for some $\rho>0$ holds if $p_i<1$ for all $i\in[n]$, i.e., there is no perfectly reliable worker. This assumption can be easily satisfied by adding an arbitrary small random noise to the worker answers as well.
By combining the statements in Corollary \ref{cor:exact_alg1}, Lemma \ref{lem:rec_p}, and Theorem \ref{thm:plugmin_mle} with $\delta_1=\delta_2=\delta$ for $\delta$ defined in \eqref{eqn:cond_delta}, we get the overall performance guarantee for Algorithm~\ref{alg:plugin_MLE}.
\begin{cor}[Performance Guarantees for Algorithm~\ref{alg:plugin_MLE}]\label{cor:alg2} Algorithm~\ref{alg:plugin_MLE} guarantees the recovery of top-two answers for all tasks with $\P\left((\hat{g}^{\mathsf{MLE}}_j,\hat{h}^{\mathsf{MLE}}_j)=(g_j,h_j),\forall j\in[m]\right)\geq1-\epsilon$ for any $\epsilon>0$ if $s$ satisfies
\beq\label{eqn:cor2_suff}
\begin{split}
s 
=&{\Omega}\left(\frac{\log(m/\epsilon)}{\delta^2 \|\bsp\|_2^2}+\frac{\log(m/\epsilon)}{{n\overline{D}}}\right).
\end{split}
\eeq
\end{cor}
In \eqref{eqn:cor2_suff}, the first term is for guaranteeing the accurate estimate of $(\bsp,\bsq)$ with $\ell_\infty$-norm error bounded by $\delta$, and the second term is for the recovery of top-two answers from MLE with high probability. Since $\|\bsp\|_2^2=\Theta(n)$, the two terms have the same order but with different constant scaling, depending on model-specific parameters $(\bsp,\bsq)$.


Lastly, we show the optimality of the convergence rates of Algorithm~\ref{alg:top_two} and Algorithm~\ref{alg:plugin_MLE} with respect to two types of minimax errors, respectively.
The proof of Theorem \ref{thm:converse} is available in Appendix \S\ref{app:thm:converse}.
\begin{thm}\label{thm:converse}
(a) Let $\caF_{\overline{p}}$ be a set of $\bsp\in[0,1]^n$ such that the collective quality of workers, measured by $\|\bsp\|_2$, is parameterized by $\overline{p}$ as
$\caF_{\bar{p}}:=\{\bsp: \frac{1}{n}\|\bsp\|_2^2={\overline{p}}\}$. Assume that $\overline{p}\leq 2/3$. If the average number $(ns)$ of samples (queries) per task is less than $\frac{1}{2\overline{p}}\log\left(\frac{K-1}{K\epsilon}\right)$, then
\beq\label{eqn:minimax_error1}
\min_{\hat{\bsg}}\max_{\bsp\in\mathcal{F}_{\overline{p}},\;\; \bsg\in[K]^m}\frac{1}{m}\sum_{j\in[m]}\p(\hat{g}_j\neq g_j)\geq \epsilon.
\eeq
(b) There is a universal constant $c>0$ such that for any $\bsp\in[0,1]^n$ and $\bsq\in(1/2,1]^m$, if the sampling probability $s\leq {1}/{(4n\overline{D})}$, then
\beq\label{eqn:minimax_error2}
\min_{(\hat{\bsg},\hat{\bsh})}\max_{\substack{(\bsg,\bsh)\in[K]^m\times [K]^m\\ g_j\neq h_j,\forall j[m]}}\frac{1}{m}\sum_{j\in[m]}\p((\hat{g}_j,\hat{h}_j)\neq (g_j,h_j)) \geq c.
\eeq
\end{thm}

From part (a) of Theorem \ref{thm:converse}, it is necessary to have $s=\Omega\left(({1}/{\|\bsp\|_2^2})\log({1}/{\epsilon})\right)$ to make the minimax error in \eqref{eqn:minimax_error1} less than $\epsilon$. Since Thm. \ref{thm:main} shows that Alg.~\ref{alg:top_two} recovers $(\hat{g}_j,\hat{h}_j)$ with probability at least $1-\epsilon$ if $s=\Omega\left(({1}/{\|\bsp\|_2^2})\log({1}/{\epsilon})\right)$ when $s_1=1$, we can conclude that Alg.~\ref{alg:top_two} achieves the minimax optimal rate for a fixed collective intelligence of workers, measured by $\|\bsp\|_2$.
From part (b) of Theorem \ref{thm:converse}, for any $(\bsp,\bsq)$, unless we have $s> {1}/{(4n\overline{D})}$ there exists a constant fraction of tasks for which the recovered top-two answers are incorrect. This bound matches with our sufficient condition on $s$ in \eqref{eqn:cor2_suff} from Alg. \ref{alg:plugin_MLE} upto logarithmic factors, as long as $\delta^2\|\bsp\|_2 \gtrsim n\overline{D}$, showing the minimax optimality of Alg. \ref{alg:plugin_MLE}. 
More discussions on the theoretical results are available at Appendix \S \ref{app:diss:theory}.

It is also worth comparing our algorithm with the simple \textit{majority voting} (MV) scheme where we infer the top-two answers by counting the majority of the received answers. Simple analysis shows that the MV scheme requires the sampling probability $s$ to be $ns= \Theta\left((\frac{1}{n}\sum_i p_i)^{-2}\log\frac{1}{\epsilon}\right)$ to recover $(g_j,h_j)$ with probability $1-\epsilon$. Since $\frac{1}{n}\|\bsp\|_2^2=\frac{1}{n}\sum_{i} p_i^2 \geq\left(\frac{1}{n}\sum_i p_i\right)^2$, Algorithm \ref{alg:top_two} achieves strictly better trade-offs unless $p_i$ is the same for all workers $i\in[n]$. For a spammer-hammer model \citep{Karger} where $\alpha\in(0,1)$ fraction of workers are hammers with $p_i=1$ and the rest are spammers with $p_i=0$,  Algorithm \ref{alg:top_two}  requires $ns=\Theta\left(\frac{1}{\alpha}\log\frac{1}{\epsilon}\right)$ samples per task, while MV requires $ns=\Theta\left(\frac{1}{\alpha^2}\log\frac{1}{\epsilon}\right)$ samples per task to recover top-two answers with probability $1-\epsilon$.

\begin{figure*}[t]
    \centering
    \includegraphics[width=0.9\textwidth]{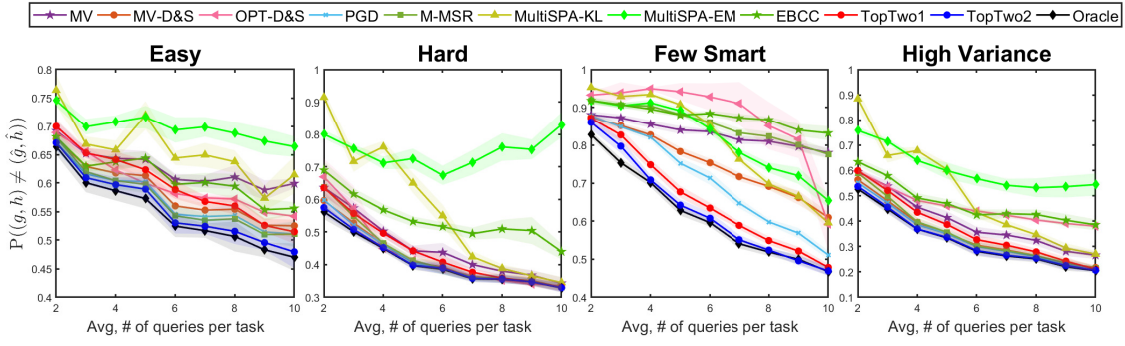}
    \caption{Prediction error in recovering the ordered top-two answers $(g,h)$ for four different scenarios, summarized in Table \ref{tab:synth_setup1},  as the avg. number of queries per task changes. Our TopTwo2 algorithm achieves the best performance, near the oracle MLE for all the scenarios. }
    \label{fig:synthetic_acc}
\end{figure*}

\section{Experiments}\label{sec:experiment}

We evaluate the proposed algorithm under diverse scenarios of synthetic datasets in Sec. \ref{sec:synth}, and for two applications--in identifying difficult tasks in real datasets in Sec. \ref{sec:real_diff}, and in training neural network models with soft labels defined from the top-two plausible labels in Sec. \ref{sec:NN_soft}.

\subsection{Experiments on synthetic dataset}\label{sec:synth}



\begin{table}[b]
    \centering
     \caption{Parameters for synthetic data experiments under diverse scenarios.}
    \label{tab:synth_setup1}
    \vspace{0.2em}
   \small{
    \begin{tabular}{c| c| c}
       \toprule
        & Worker  & Task  \\
       \midrule
       Easy & $p_i \in [0, 1]$ & $q_j \in [0.9, 1]$ \\
       \midrule
       Hard & $p_i \in [0, 1]$ & $q_j \in (0.5, 0.6]$ \\
       \midrule
       \multirow{2}{*}{Few-smart} & 90\% $p_i\in[0,0.1]$  & \multirow{2}{*}{$q_j\in(0.5,1]$} \\
       & 10\% $p_i\in [0.9,1]$ & \\
       \midrule
       \multirow{2}{*}{High-variance} & \multirow{2}{*}{$p_i\in[0,0.1]$} & 50\% $q_j\in(0.5,0.6]$ \\
       & & 50\% $q_j\in[0.9,1.0]$ \\
        \bottomrule
    \end{tabular}}
\end{table}

We compare the empirical performance of Algorithm \ref{alg:top_two} and Algorithm \ref{alg:plugin_MLE} (referred as TopTwo1 and TopTwo2) with other baselines: majority voting(MV), 
MV-D\&S and OPT-D\&S \citep{SEM}, PGD \citep{PGD}, M-MSR \citep{M-MSR}, MultiSPA-KL and MultiSPA-EM \citep{ibrahim2019crowdsourcing}, EBCC \citep{li2019exploiting} and oracle-MLE. 
OTP-D\&S and MV-D\&S assume the D\&S model and use the EM algorithm, initialized with worker confusion matrices estimated by spectral method or MV, respectively.  PGD, on the other hand, assumes a simpler single-coin D\&S model, which is equivalent to our model \eqref{eqn:A_dist} with a fixed $q_j=1$ for all tasks, and estimates $p_i$ of each worker and uses this estimate to compute the MLE. 
We choose these baselines because they have the strongest established guarantees in the literature, and they are all MLE-based approaches, from which the top-two answers can be inferred. 
Obviously, oracle-MLE, which uses the ground-truth model parameters, provides the best possible performance since oracle MLE uses the ground-truth $(\bsp,\bsq)$ from which the synthetic data is generated.
See Appendix \S\ref{app:baseline} for more details of these baselines. 

We devise four scenarios described in Table \ref{tab:synth_setup1} to verify the robustness of our model for different $(\bsp,\bsq)$ ranges, at $(n,m)=(50,500)$ with $s\in(0,0.2]$. The number of choices for each task is fixed as 5.
Fig. \ref{fig:synthetic_acc} reports the empirical error probability $\frac{1}{m}\sum_{j=1}^m \P((\hat{g}_j,\hat{h}_j)\neq (g_j,h_j))$  averaged over 30 runs, with 95\% confidence intervals (shaded region). 
Four columns correspond to the four scenarios, respectively. The prediction errors for $g_j$ and $h_j$ are plotted in Fig. \ref{fig:app:synthetic_acc} of Appendix. \S\ref{sec:app:more_plots}.



We can observe that for all considered scenarios, TopTwo2 achieves the best performance, close to the oracle MLE, in recovering $(g_j,h_j)$. Depending on the scenario, the reason for TopTwo2's outperformance can be explained differently. For the \underline{Easy} scenario, since $q_j$ is close to 1, it is easy to distinguish $g_j$ from $h_j$, but hard to distinguish $h_j$ from other labels. Our algorithm achieves the best performance in estimating $h_j$ by a large margin (Fig. \ref{fig:app:synthetic_acc}), which also leads to a better performance in estimating $(g_j,h_j)$ compared to other baselines. For the \underline{Hard} scenario, it is difficult to distinguish $g_j$ from $h_j$, but our algorithm using an accurate $\hat{q}_j$ can better distinguish $g_j$ from $h_j$. 
For \underline{Few-smart}, our algorithm achieves the largest gain compared to other methods, since our algorithm can effectively distinguish few smart workers from spammers.  \underline{High-variance} shows the effect of having diverse $q_j$ in a dataset. 

We remark that our algorithm (TopTwo2) achieves the best performance, close to that of the oracle MLE, for all scenarios, while the next performer changes depending on the scenario. For example, the OPT D\&S is the second best performer in the \underline{Hard} scenario, while it is the worst performer in the \underline{Few-smart} scenario. 
We also show the robustness of our algorithm to changes in model parameters in Appendix \S\ref{app:synthetic_exp}.

\begin{figure*}
\centering
	\subfloat[][The average prediction error on color comparison tasks  \label{fig:color_error}]{\includegraphics[width=0.6\textwidth]{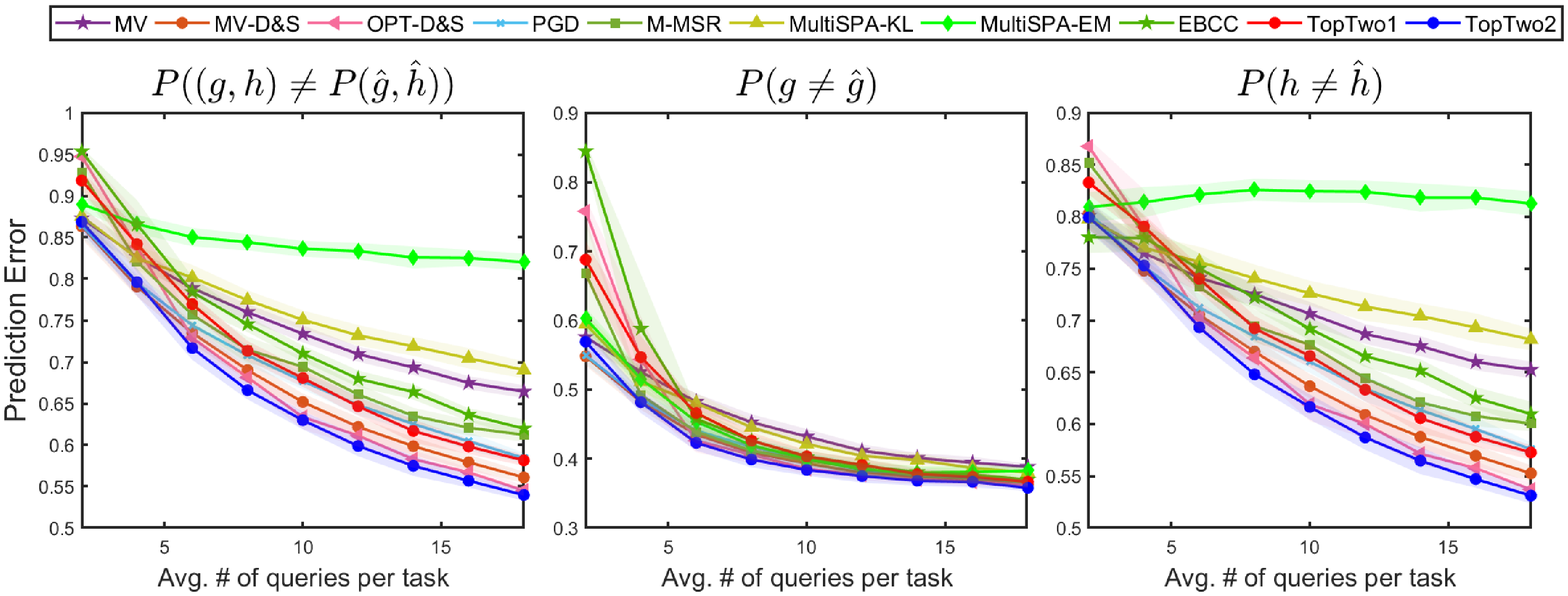}} \qquad
    \subfloat[][Histogram of dist. gap\label{fig:color_histo}]{ \includegraphics[width=0.25\textwidth]{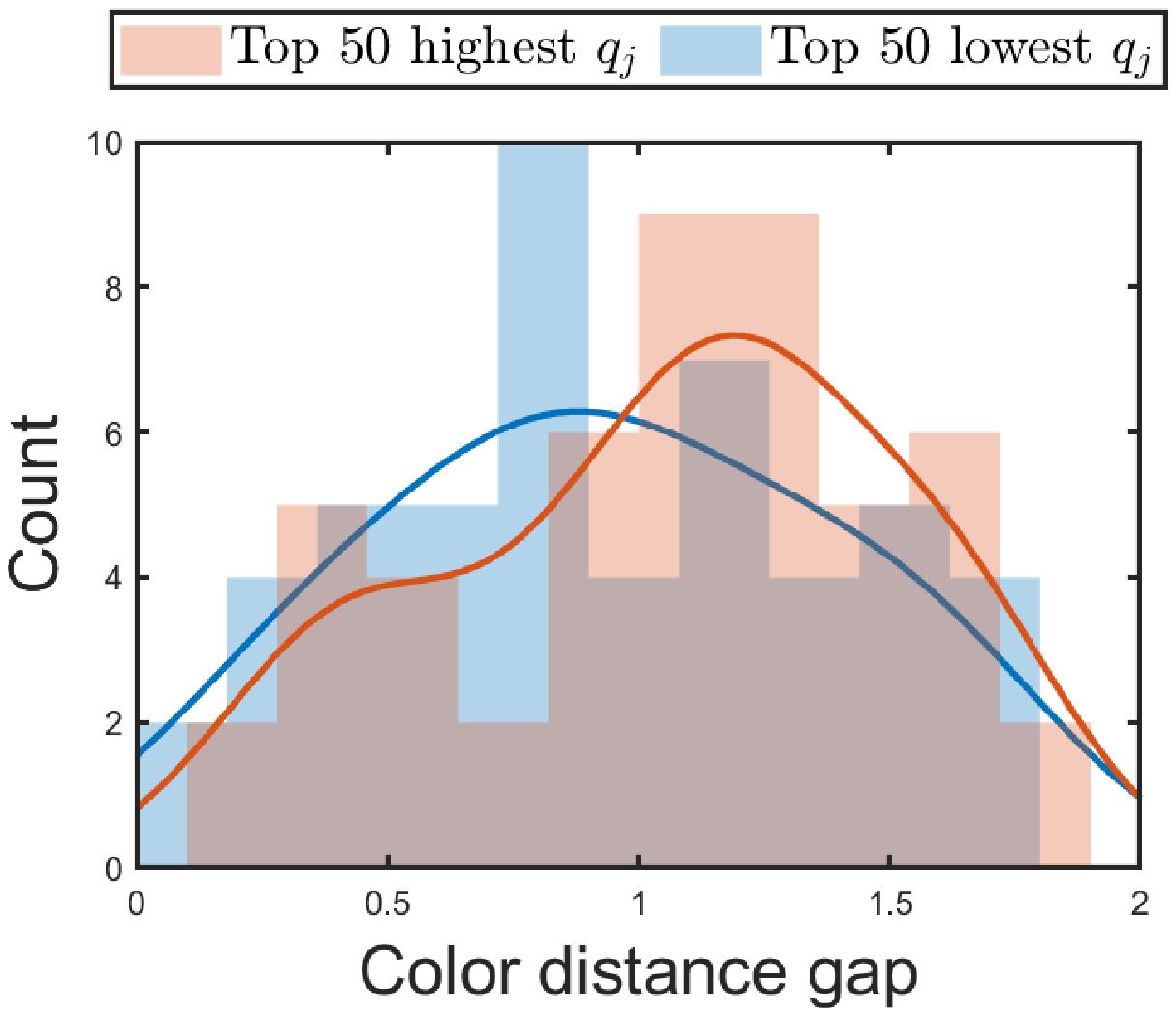}}
    \caption{(a) Prediction error for $(g_j,h_j)$, $g_j$ and $h_j$ (from left to right) for color comparison task using real data collected by MTurk. TopTwo2 algorithm achieves the best performance. (b) Histogram of color distance gap for two task groups, the easy group with the highest $q_j$(red)  and the  difficult group with the lowest $q_j$(blue). The difficult group  tends to have a smaller color distance gap. }
    \label{fig:25_gaussian}
\end{figure*}
\subsection{Experiments on real-world dataset: inferring task difficulties}\label{sec:real_diff}

We provide experimental results using real-world data collected by MTurk and show that our algorithm can be used to infer task difficulty.
Since publicly available datasets do not provide information about confusing answers or task difficulty, we designed a new set of multiple-choice tasks for which we can identify both. 
We designed a color comparison task in which we asked the crowd to choose, from six given choices, the color that looks the most like a reference color of each task. See Fig. \ref{tab:color_example} in Appx. \S\ref{subsec:dataset} for example tasks. 
After randomly generating a reference color and the six choices, we identified the ground truth and the most confusing answer for each task by measuring the distance between colors using the CIEDE2000 color difference formula \citep{CIEDE2000}. If the distance from the reference color to the ground truth is much shorter than the distance to the most confusing answer, then the task is considered easy.
We designed 1,000 tasks and distributed them to 200 workers, collecting 19.5 responses for each task. After collecting the data, we subsampled it to simulate how the prediction error decreases as the number of responses per task increases.
Fig. \ref{fig:color_error} shows the performance in detecting $(g_j,h_j)$, $g_j$ and $h_j$, averaged over 10 random sampling,  with a 95\% confidence interval (shaded region). 

First, we can verify that the ground truth and the most confusing answer we identified by the measured color distance are valid with the collected data, since the prediction error actually decreases as the number of queries per task increases. 
As shown in Fig. \ref{fig:color_error}, TopTwo2 algorithm achieves the best performance in detecting $(g_j,h_j)$, $g_j$ and $h_j$ in all ranges.
We further investigate the correlation between the task difficulty - quantified by the distance gap between the ground truth and the most confusing answer from the reference color - and the estimated confusion probability $q_j$ across tasks.
We select the top 50 most difficult/easiest tasks according to the estimated confusion probability $q_j$ and plot the histograms of the distance gap for the two groups in Fig \ref{fig:color_histo}. We can see that the difficult group (blue, with the lowest $q_j$) tends to have a smaller distance gap than the easy group (red). This result shows that our algorithm can  identify difficult tasks in real datasets.

\subsection{Training neural networks with soft labels having top-two information}\label{sec:NN_soft}

An appealing example where we can use the knowledge of the second best answer is in training deep neural networks for classification tasks. Traditionally, a hard label (one ground-truth label per image) has been used to train a classifier. Recent work has shown that using a soft label (a full label distribution that reflects human perceptual uncertainty) is sometimes advantageous in obtaining a model with better generalization ability \citep{peterson2019human}. However, obtaining an accurate full label distribution requires much higher sample complexity than just recovering only the ground-truth.  For example, \citet{peterson2019human} provided a CIFAR10H dataset with full human label distributions for 10,000 instances of CIFAR10 test examples by collecting an average of 50 judgments per image, which is about 5-10 times larger than the usual datasets (Table \ref{tab:data_info} in Appendix \ref{subsec:dataset}).

Our top-two model, on the other hand, can effectively reduce the required sample complexity while still providing the benefit of the soft-label training. 
To demonstrate this idea, we train two deep neural network models, VGG-19 and ResNet18, with the soft label vectors having the top-two (top2) structure extracted from the CIFAR10H dataset\footnote{As in \citep{peterson2019human}, we used the original 10,000 test examples from CIFAR10 for training and 50,000 training examples for testing. Thus, the final accuracy is lower than usual. Since CIFAR10H is collected from selected `reliable' workers who answered a set of test examples with an accuracy higher than 75\%, we directly used the top two dominant answers and the fraction between them to obtain the soft label vector (top2). }. 
We then compare the training and testing results of our method with those of the hard label (hard) and full label (full) training. Experimental details are given in Appendix \S\ref{sec:app:NN_soft}.
Compared to the original training with hard labels, training with the top-two soft labels achieves 1.56\% and 4.09\% higher test accuracy in VGG-19 and ResNet18, respectively (averaged over three runs, 150 epochs), as shown in Table \ref{tab:data_CIFAR10}. This result is also comparable to that of the full label distribution. 
It shows that training with the top-two soft labels can achieve better generalization (test accuracy) than training with hard labels, because the top-two soft label contains simple but helpful side information, the most confusable class, and the confusion probability. In Sec. \ref{sec:toptwo_t_robust}, we also report an additional experiment showing that training with the top-two labels is more robust to the label noise than  training with the full label distribution. 
 \begin{table}[t]
 \centering
    \caption{Comparison of performances for CIFAR10H dataset with hard/soft label training}
    \label{tab:data_CIFAR10}
    \vspace{0.2em}
    \renewcommand{\arraystretch}{1.2}
    \small{
    \begin{tabular}{c c c}
        \toprule
         Network & Train accuracy & Test accuracy  \\
        \midrule
        VGG-19 (hard) & \textbf{97.46$\pm$0.59\%} & 77.64$\pm$1.54\% \\
        VGG-19 (top2) & 97.00$\pm$0.51\% & \textbf{79.20$\pm$1.04\%} \\
        VGG-19 (full) & 96.69$\pm$0.48\% & 78.66$\pm$0.97\%  \\
        \midrule
        ResNet18 (hard) & 98.47$\pm$0.320\% & 76.49\%$\pm$1.80\%  \\
        ResNet18 (top2) & 98.67$\pm$0.491\% & 80.58\%$\pm$2.36\% \\
        ResNet18 (full) & \textbf{99.19$\pm$0.125\%} & \textbf{80.93\%$\pm$2.66\%} \\
        \bottomrule
    \end{tabular}}
\end{table}

 \section{Discussion}\label{sec:diss}

We proposed a new model for multiple-choice crowdsourcing with top-two confusable answers and varying confusion probabilities across tasks. We provided an algorithm to infer the top-two answers and the confusion probability. 
This work can benefit various query-based data collection systems, such as MTurk or review systems, by providing additional information about the task, such as the most plausible answer other than the ground truth and how plausible it is. This information can be used to quantify the accuracy of the ground truth or to classify the tasks based on difficulty. We also demonstrated possible applications of our top-two model in designing soft labels for training neural networks.
\section{Acknowledgements}
This research was supported by the National Research Foundation of Korea under grant 2021R1C1C11008539, and  by the MSIT(Ministry of Science and ICT), Korea, under the ITRC(Information Technology Research Center) support program(IITP-2023-2018-0-01402) supervised by the IITP(Institute for Information \& Communications Technology Planning \& Evaluation).

\bibliography{main_toptwo}
\bibliographystyle{icml2023}

\newpage
\appendix
\onecolumn
\section{Verification for the Proposed Top-Two Model}\label{sec:app:model}

We proposed the top-two model to reflect the key attributes of seven datasets including Adult2, Dog, Web, Flag, Food, Plot, and Color, of which the details are summarized in Appendix \ref{subsec:dataset}.

Table \ref{tab:data_info_3} shows empirical distributions of the mean incidence of responses for the top-three dominating answers, sorted by the dominance proportions, for the six public datasets and the Color dataset that we collected, with the standard deviation over the tasks in the dataset. In Fig. \ref{fig:label_dist}, we also plot empirical distributions of the mean incidence of responses sorted by the dominant proportion with error bars indicating the standard deviation. The $i$-th data point represents the average incidence of the $i$-th highest response in each task. For example, in Adult2 dataset, the most dominating answer takes 0.8 portion of the total answers, and the next dominating answer takes 0.14 portion of the total answers on average.
\begin{table}[h]
    \centering
    \caption{Proportions of top-three dominating answers in public datasets}
    \label{tab:data_info_3}
    \vspace{0.2em}
    \renewcommand{\arraystretch}{1.2}
    \begin{tabular}{c c c c }
        \toprule
         Dataset & Ground truth& 2nd dominating answer & 3rd dominating answer\\
        \midrule
        Adult2 & 0.80$\pm$0.19 & 0.14$\pm$0.13 & 0.04$\pm$0.07 \\
        Dog & 0.76$\pm$0.15 & 0.22$\pm$0.14 & 0.01$\pm$0.04 \\
        Web & 0.59$\pm$0.20 & 0.25$\pm$0.12 &0.12$\pm$0.09 \\
        Flag & 0.90$\pm$0.16 & 0.09$\pm$0.13 & 0.01$\pm$0.03 \\
        Food & 0.80$\pm$0.18 & 0.17$\pm$0.15 & 0.02$\pm$0.05 \\
        Plot & 0.62$\pm$0.21 & 0.30$\pm$0.16 & 0.06$\pm$0.07 \\
        Color & 0.43$\pm$0.1 & 0.23$\pm$0.06 & 0.15$\pm$0.05 \\
        \bottomrule
    \end{tabular}
\end{table}

\begin{figure}[tbh]
    \centering{
    \subfloat[\label{fig:real_adult}]{\includegraphics[width=0.25\textwidth]{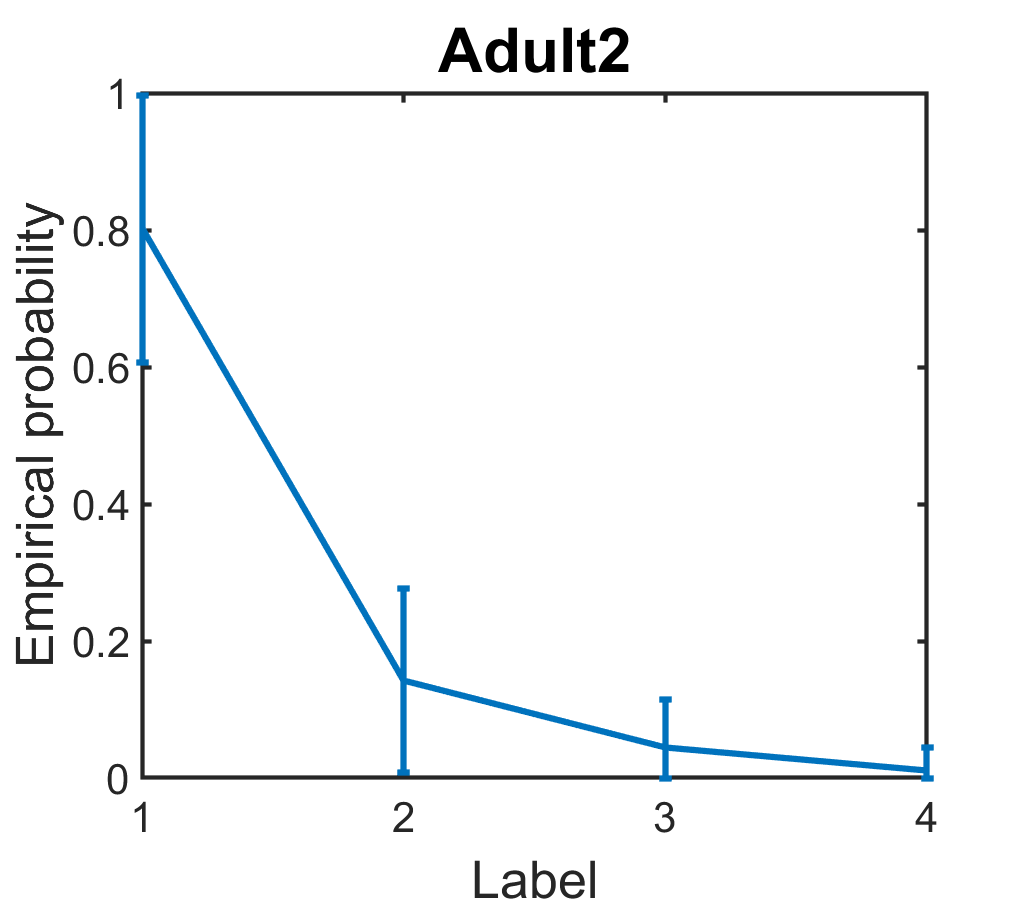}}
    \subfloat[\label{fig:real_dog}]{\includegraphics[width=0.25\textwidth]{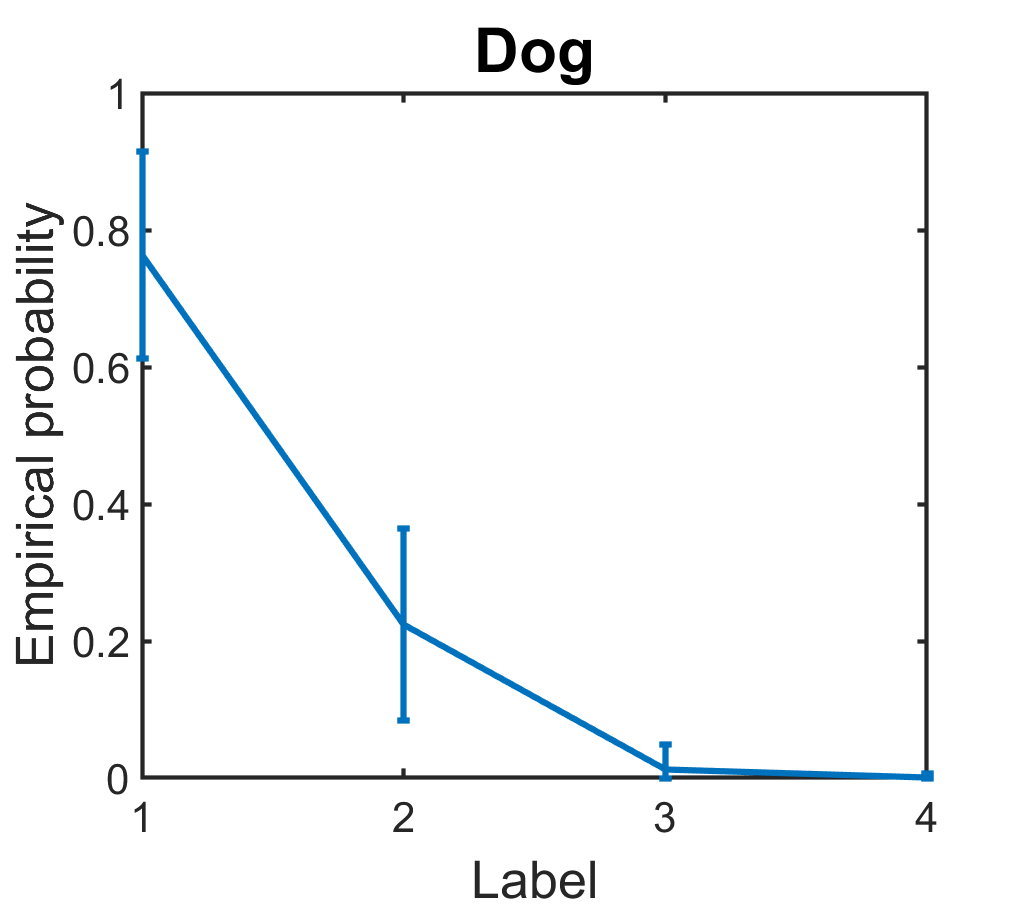}}
    \subfloat[\label{fig:real_web}]{\includegraphics[width=0.25\textwidth]{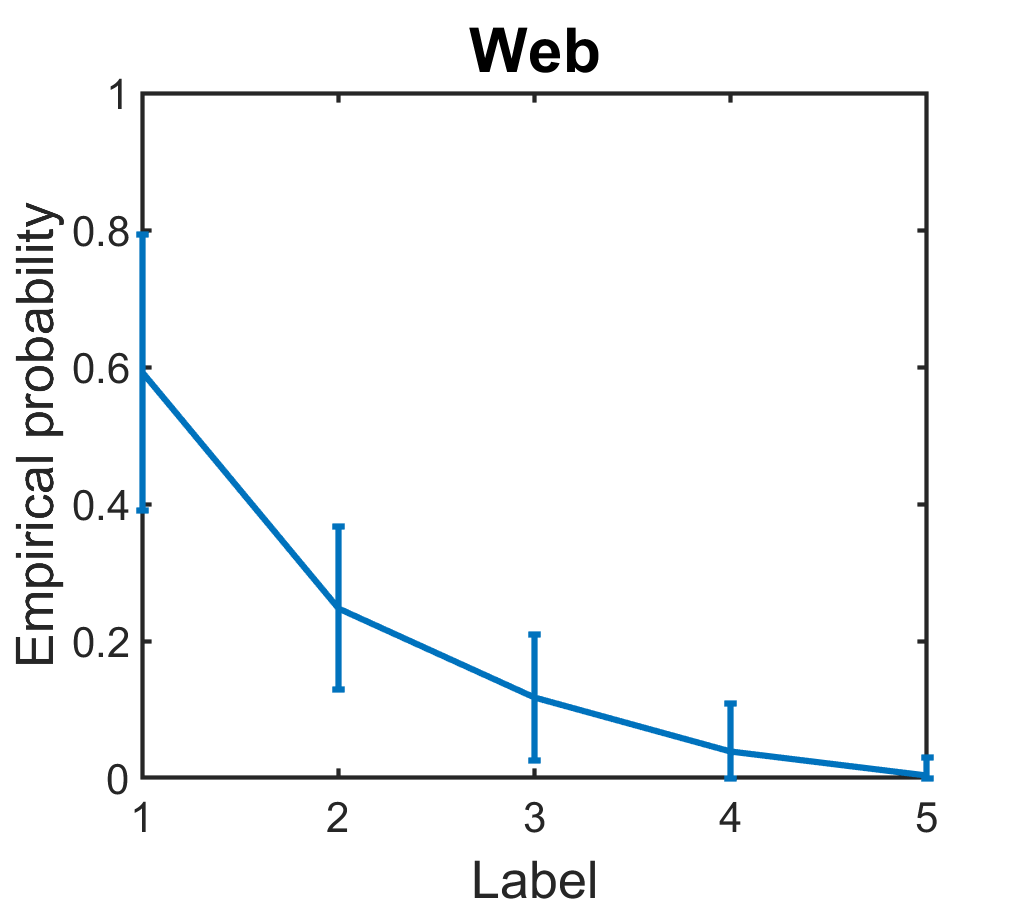}}
    \subfloat[\label{fig:real_flag}]{\includegraphics[width=0.25\textwidth]{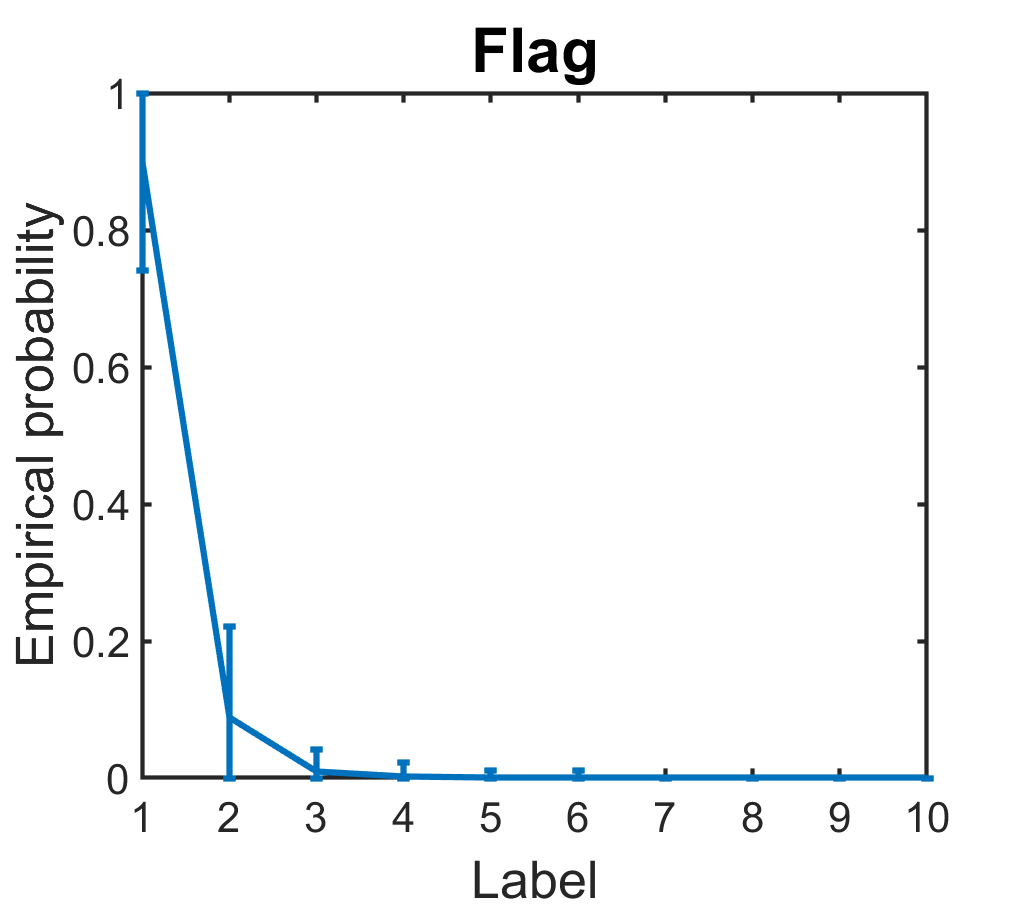}}
    \\
    \subfloat[\label{fig:real_food}]{\includegraphics[width=0.25\textwidth]{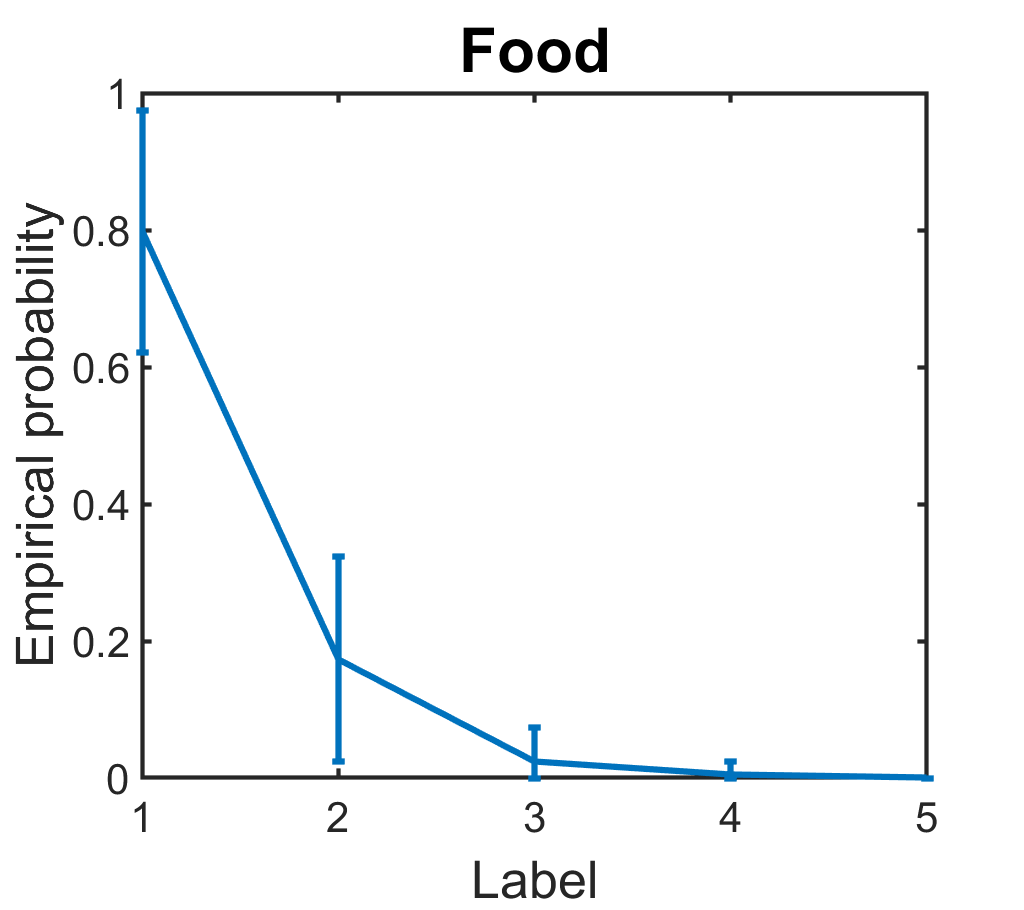}}
    \subfloat[\label{fig:real_plot}]{\includegraphics[width=0.25\textwidth]{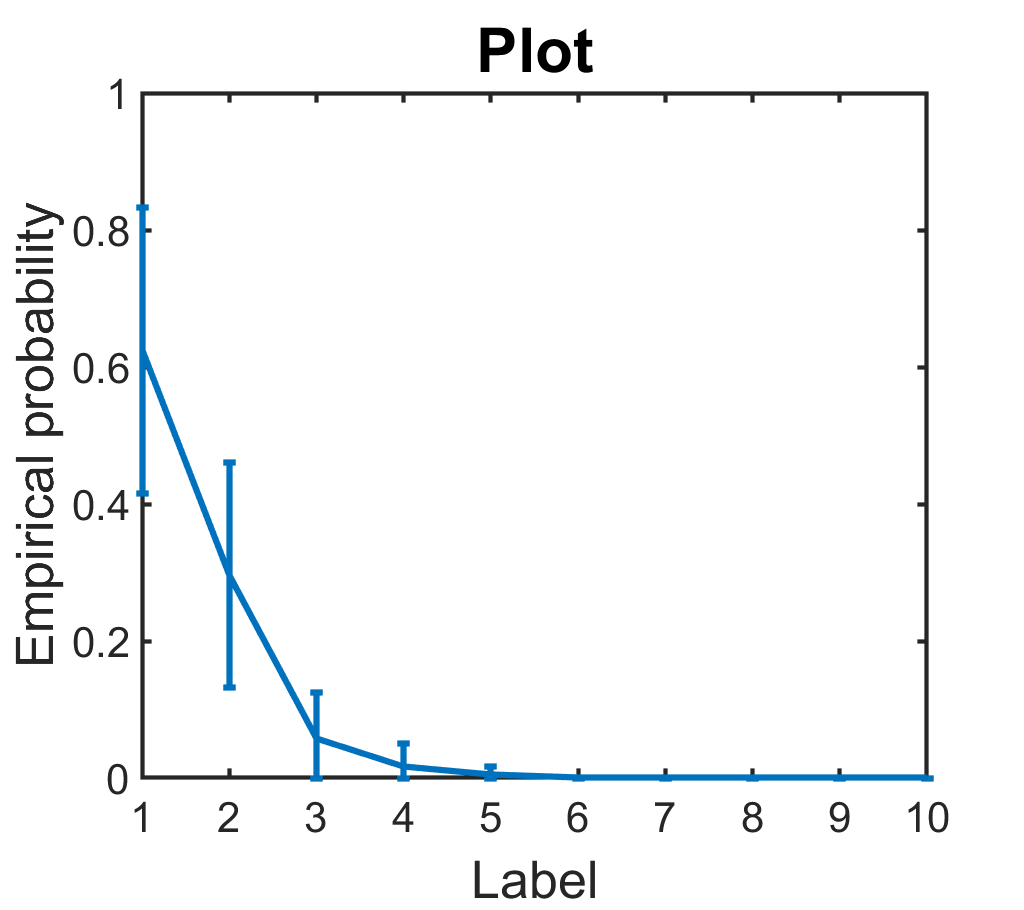}}
    \subfloat[\label{fig:real_color}]{\includegraphics[width=0.25\textwidth]{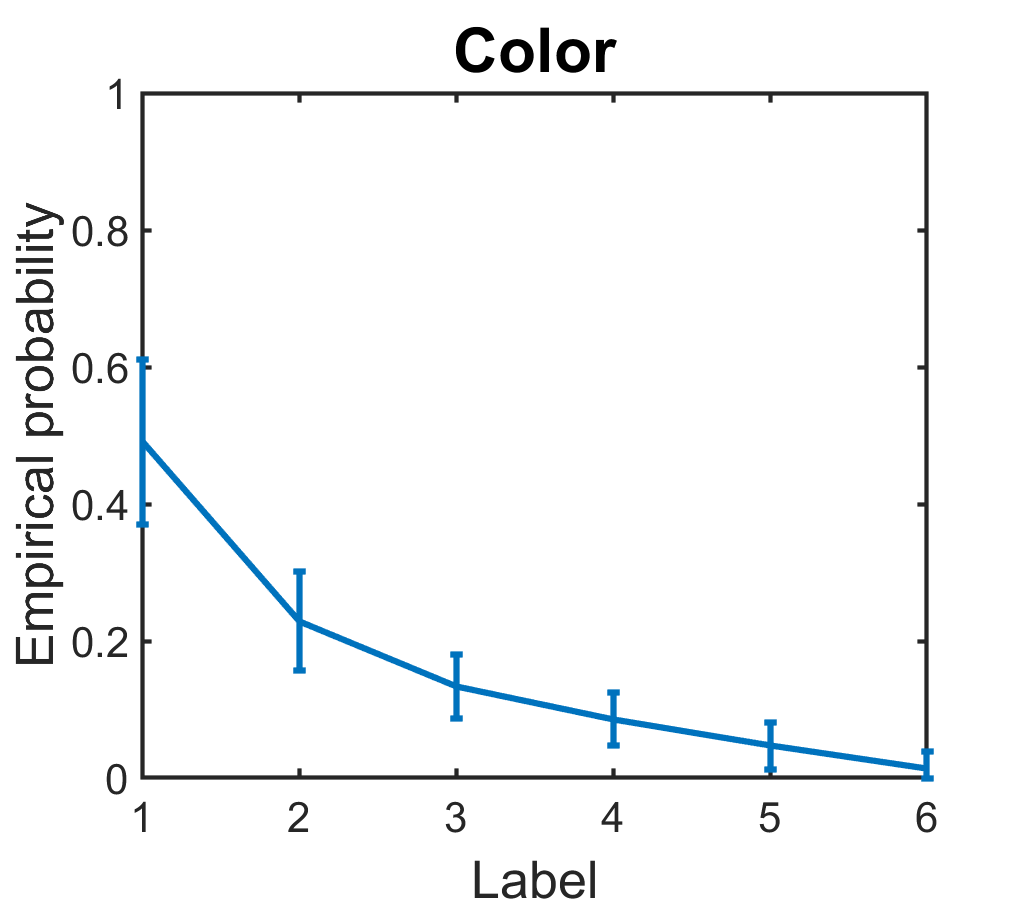}}
}    
    \caption{Empirical distribution of the mean incidence of responses sorted by the dominant proportion, averaged over  all tasks in each dataset. The $i$-th data point represents the average incidence of the $i$-th highest response in each task.  The error bars indicate the standard deviation of the mean incidence of the $i$-th dominating answer over the tasks in the dataset. }
    \label{fig:label_dist}
\end{figure}

From the table and figure, we can observe that for all the considered public datasets the top-two answers dominate the overall answers, i.e., about 65-90\% of the total answers belong to the top two. Furthermore, the average ratio from the most dominating answer to the second one is 4:1, while that between the second and the third is 7.5:1. 
There often exist overlaps in the error bars between the ground truth and the second dominating answer, e.g., for Web, Plot, and Color datasets, but no such overlap is found between the ground truth and the third dominating answer. 
What we can call a \textit{`confusing answer'} is an answer that has an incidence rate comparable to that of the ground truth. In all the considered datasets, only the second dominating answer shows such a tendency, and thus, we can conclude that the third dominating answer cannot be called a `confusing answer', and the top-two model in \eqref{eqn:A_dist} well describes the errors in answers caused by confusion. 


Moreover, from the public datasets, we also observe that the task difficulty can be quantified by the confusion probability between the top-two answers. As an example, for the Web dataset, when we select the easiest 500 tasks and hardest 500 tasks by ordering tasks with the ratio of correct answers, the ratio between the ground-truth to the 2nd best answer was 10.7:1 for the easiest group, while it was 1.5:1 for the hardest group. This observation shows that the ratio between the top-two answers indeed captures task difficulty as does our model parameter for task difficulty $q_j$ in \eqref{eqn:A_dist}.

\subsection{Datasets}\label{subsec:dataset}
We collect six publicly available multi-class datasets: Adult2, Dog, Web, Flag, Food and Plot. Since these datasets do not provide information about the most confusing answer or the task difficulty, we additionally create a new dataset called `Color', for which we can identify the most confusing answer and also quantify the task difficulty for all the included tasks.

\begin{itemize}
	\item \textbf{Color} is a dataset where the task is to find the most similar color to the reference color among six different choices. For each task, we randomly create a reference color and then choose six choices of colors. The distance from the reference color to the ground truth color is in between 4.5 and 5.5, the distance to the most confusing answer is in between 5.5 and 6.5, and the distance to the rest of the choices is between 11 and 12, where the distance between the pairs of colors is measured by CIEDE2000 \citep{CIEDE2000} color difference formulation. The tasks are ordered in terms of their difficulty levels by measuring the gap between: the distance from the reference color to the ground truth; and that to the most confusing answer. If the distance from the reference color to the ground truth is much shorter than that to the most confusing answer, then the task is considered easy.
Using MTurk, we collected 19600 labels from 196 workers for 1000 tasks. Each Human Intelligence Task (HIT) is composed of randomly selected 100 tasks, and we pay \$1 to each worker who completed a HIT. 
Fig. \ref{tab:color_example} shows an example task for the Color dataset. 


\begin{figure}[t]
    \centering
    \subfloat[$g_j = 6$ and $h_j = 5$]{\includegraphics[width=0.4\textwidth]{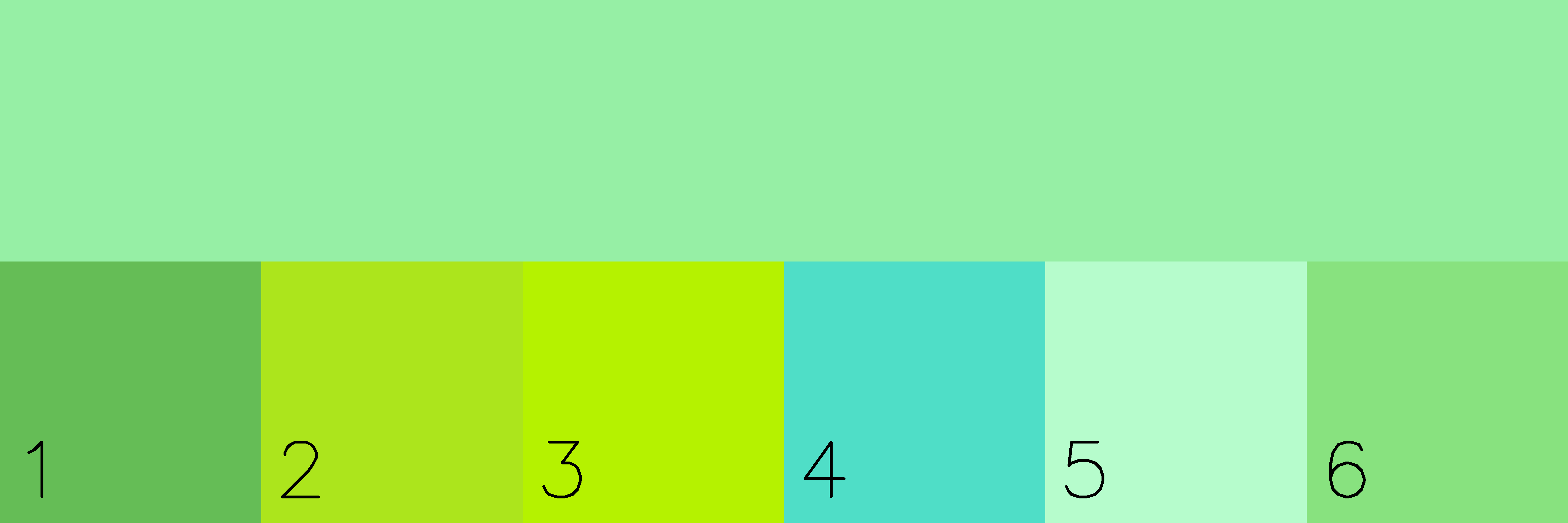}} \qquad
    \subfloat[$g_j = 4$ and $h_j = 3$]{\includegraphics[width=0.4\textwidth]{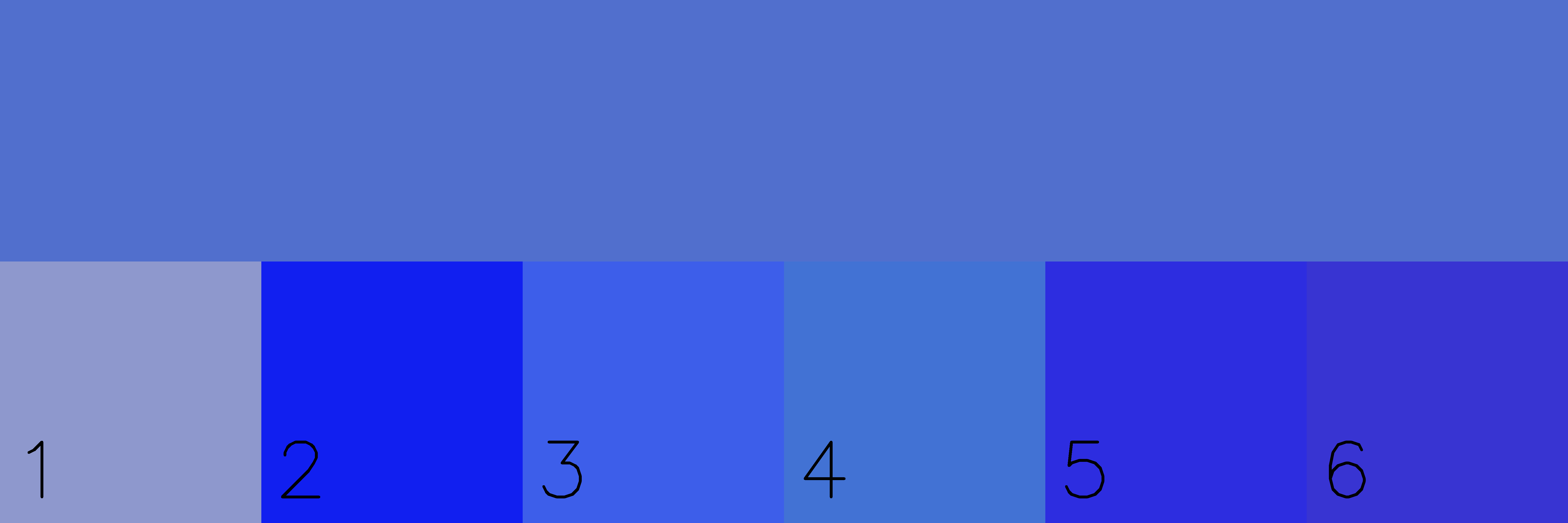}} \\
    \subfloat[$g_j = 5$ and $h_j = 3$]{\includegraphics[width=0.4\textwidth]{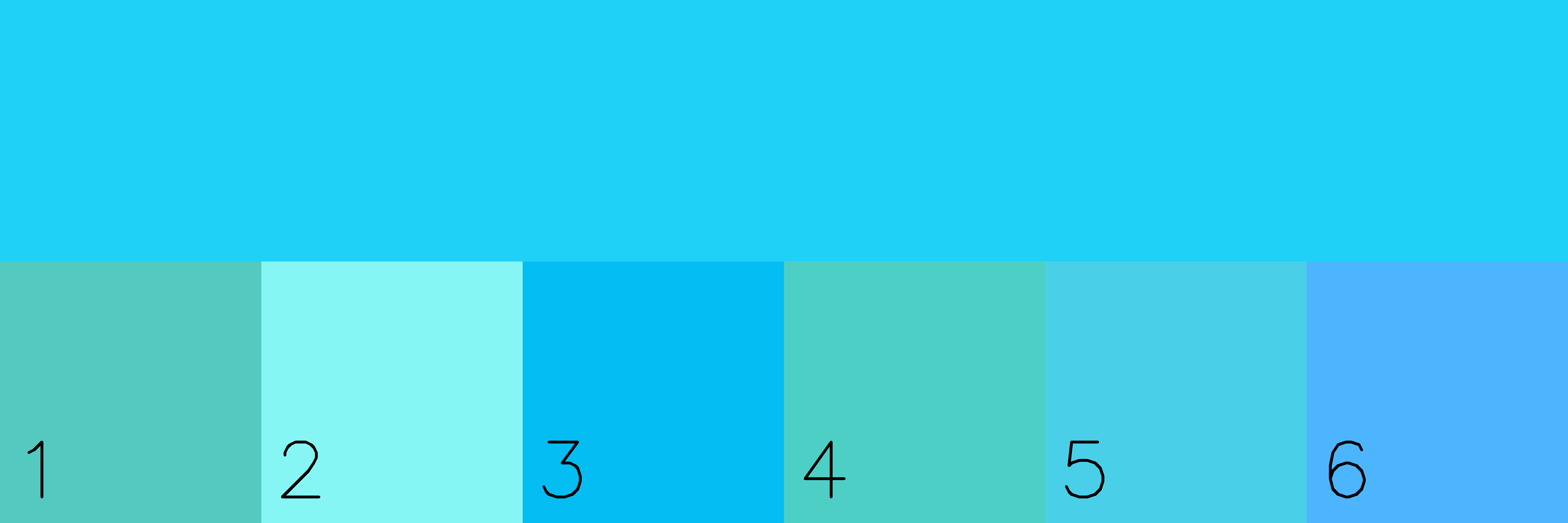}} \qquad
    \subfloat[$g_j = 6$ and $h_j = 2$]{\includegraphics[width=0.4\textwidth]{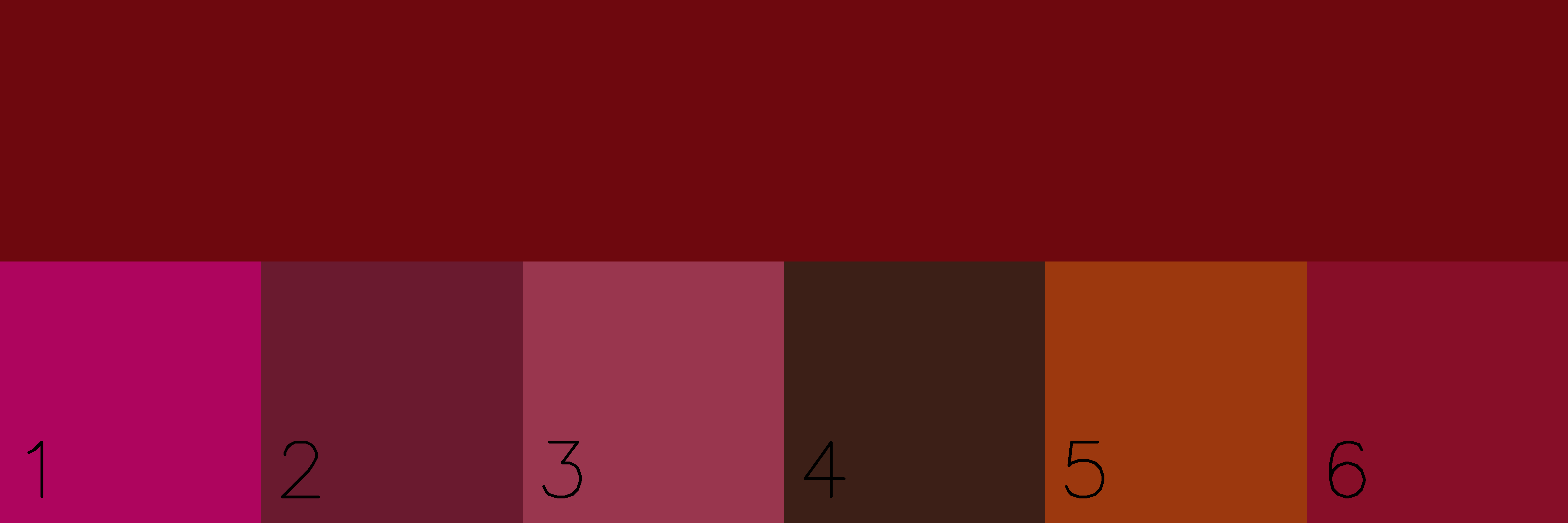}}
    \caption{Example tasks for `Color' dataset where the ground truth $g$ and the most confusing answer $h$ are determined by the color distance from the reference color (top).}
    \label{tab:color_example}
\end{figure}

    \item \textbf{Adult2} \citep{ipeirotis2010quality} is a 4-class dataset where the task is to classify the web pages into four categories (G, PG, R, X) depending on the adult level of the websites. This dataset contains 3317 labels for 333 websites which are offered by 269 workers. 
    
    \item \textbf{Dog }\citep{SEM} is a 4-class dataset where the task is to discriminate a breed (out of Norfolk Terrire, Norwich Terrier, Irish Wolfhound, and Scottich Deerhound) for a given dog. This dataset contains 7354 labels collected from 52 workers for 807 tasks. 
    
    \item \textbf{Web} \citep{zhou2012learning} is a 5-class dataset where the task is to determine the relevance of query-URL pairs with a 5-level rating (from 1 to 5). The dataset contains 15567 labels for the 2665 query-URL pairs offered by 177 workers. 

    \item \textbf{Flag} \citep{krivosheev2020detecting} is a dataset for multiple-choice tasks where each task is to identify the country for a given flag from 10 given choices. A total of 1600 votes are collected from 220 workers for the 100 tasks. 
    
    \item \textbf{Food} \citep{krivosheev2020detecting} is a dataset for multiple-choice tasks where each task asks to identify a picture of a given food or dish from 5 given choices. 
    This dataset contains 1220 labels for 76 tasks collected from 177 workers. 
    
    \item \textbf{Plot} \citep{krivosheev2020detecting} is a dataset for multiple-choice tasks where the task is to identify a movie from a description of its plot from 10 given choices. Only workers who correctly solved the first 10 test questions can answer the rest of the tasks. A total of 1937 labels are collected from 122 workers for 100 tasks. 

\end{itemize}
Table~\ref{tab:data_info} shows a summarized information for the introduced datasets. 
\begin{table}[h]
    \centering
    \caption{Dataset information}
    \label{tab:data_info}
    \vspace{0.2em}
    \renewcommand{\arraystretch}{1.2}
    \begin{tabular}{c c c c c c c}
        \toprule
         Dataset & \# workers & \# tasks & \# labels or choices & sparsity& $d_{task}$ & $d_{worker}$\\
        \midrule
        Adult2 & 269 & 333 & 4 & 0.037 & 10.0 & 12.4 \\
        Dog & 109 & 807 & 4 & 0.092 & 10.0 & 74.0 \\
        Web & 176 & 2653 & 5 & 0.033 & 5.9 & 88.3 \\
        Flag & 220 & 100 & 10 & 0.073 & 16.0 & 7.3 \\
        Food & 177 & 54 & 5 & 0.125 & 22.1 & 6.7 \\
        Plot & 122 & 56 & 10 & 0.293 & 35.7 & 16.4 \\
        Color & 196 & 1000 & 6 & 0.1 &19.5 &99.4 \\
        \bottomrule
    \end{tabular}
\end{table}

\subsection{Top-$T$ model: extension of the Top-Two model}\label{subsec:top_T}
In this section, we also show that our top-two model \eqref{eqn:A_dist} can be generalized to have $T\geq 2$ plausible answers. The distribution of the response $A_{ij}$ can be defined as follows:
\beq\label{eqn:topk_dist}
A_{ij}=\begin{cases}
g_{jt} \text{ for } t \in [T] & \text{w.p. } s\left(p_i q_{jt}+\frac{1-p_i}{K}\right);\\
\text{each }b\in[K]\backslash\{g_{j1},\dots,g_{jT}\} &  \text{w.p. } s\left(p_i (1-\sum_{t=1}^T q_{jt})+\frac{1-p_i}{K}\right);\\
0,&\text{w.p. } 1-s,
\end{cases}
\eeq
where $g_{j1}, \dots, g_{jT}$ represent the $T$ plausible answers, and $q_{j1}, \dots, q_{jT}$ are the associated confusion probabilities with respect to the ground truth. Without loss of generality, let the ground truth answer of the task $j$  be $g_{j1}$, where we assume $q_{j1} \ge q_{j2} \dots \ge q_{jT}>1-\sum_{t=1}^T q_{jt}$. Similar to the top-two model, we can define a binary converted observation matrices $\bsA^{(k)}$ for $1\leq k<K$, which enjoy the rank-1 structure.
The analysis of the binary converted observation matrices reveals that $\Delta r_j^{(k)}$ in \eqref{eqn:delta_r} can be represented as below
\beq
 \Delta r_j^{(k)}= \begin{cases}
\frac{1}{K}-q_{jt}&\text{ for }k=g_{jt}, t\in[T], \\
\frac{1}{K} &\text{ otherwise}.\\
\end{cases}
\eeq
Thus, we can estimate the top-$T$ plausible answers for each task by finding the lowest-$T$ values of  $\Delta r_{j}^{(k)}$, $k\in[K]$. We can also obtain $q_{jt}$ from $\frac{1}{K}-\Delta r_j^{(t)}$. 
Based on this observation, we can generalize Algorithm \ref{alg:top_two} of the top-two model to  Algorithm \ref{alg:top_T} of the top-$T$ model.
\begin{algorithm}[h]
	\caption{Spectral Method for Initial Estimation (Top-$T$1 Algorithm) }
	\label{alg:top_T}
\begin{algorithmic}[1]
	\STATE {\bfseries Input:} data matrix $\bsA^1\in\{0,1,\dots,K\}^{n\times m}$ and parameter $\eta>0$ where $
    \eta \sqrt{n} \le \|\bsp\|_2 \le \sqrt{n}
$.
	\STATE Randomly split (with equal probabilities) and convert $\bsA^1$ into binary matrices $\bsX^{(k)}\in\{-1,0,1\}^{n\times m}$ and $\bsY^{(k)}\in\{-1,0,1\}^{n\times m}$ for $1\leq k<K$ as described in Sec. \ref{sec:stage1}. 
	\STATE  Let $\bsu^{(k)}$ be the leading normalized left singular vector of $\bsX^{(k)}$. 
Trim the abnormally large components of $\bsu^{(k)}$ by setting them to zero if $u_i^{(k)}>\frac{2}{\eta\sqrt{n}}$ and denote the resulting vector as  $\tilde{\bsu}^{(k)}$.
	\STATE Calculate the estimate of $\|\bsp\|\bsr^{(k)}$ by defining
$
\bsv^{(k)}:=\frac{1}{s'}(\bsY^{(k)})^{\top}\tilde{\bsu}^{(k)}.
$
Assume $\bsv^{(0)}:=\boldsymbol{0}$ and $\bsv^{(K)}:=\boldsymbol{0}$.
	\STATE For $k\in[K]$, calculate
$
\Delta v_j^{(k)}:= v_j^{(k)}-v_j^{(k-1)}.
$
Estimate the top-$T$ answers for $j\in[m]$ by 
\beq
\begin{split}\label{alg:gh_T}
\hat{g}_{jt}:= \argmin_{k\in[K], k \neq \hat{g}_{j1}, \dots, \hat{g}_{j(t-1)}}\Delta v_j^{(k)}, t \in [T]. 
\end{split}
\eeq

	\STATE Estimate $\|\bsp\|_2$ by
$l_j:= \frac{K}{K-T} \sum_{k \notin \{g_{j1}, \dots, g_{jT}\}} \Delta v_j^{(k)}$ and ${l}:=\frac{1}{m}\sum_{j=1}^m l_j$.

	\STATE Estimate ${q}_{jt}$ for $j\in[m]$ and $t \in [T]$ by defining
\beq\label{eqn:est_qT}
\hat{q}_{jt}:= {1}/{K}-{\Delta v_{j}^{(\hat{g}_{jt})}}/{l}.
\eeq

	\STATE {\bfseries Output:} estimated top-T answers $\{\hat{g}_{j1}, \dots, \hat{g}_{jT}\}_{j=1}^m$ and confusion probability matrix $\hat{\bsq}$.
\end{algorithmic}
\end{algorithm}
\newpage
To proceed to the second stage, we also generalize Algorithm \ref{alg:plugin_MLE} of the top-two model to Algorithm \ref{alg:plugin_MLE_T} of the top-$T$ model by defining the estimate of the worker reliability in a similar way as  \eqref{eqn:hatpi}, but for the case of the top-$T$ model:
\beq\label{eqn:hatpi_T}
\hat{p}_i=\frac{K}{(K-T)}\left(\frac{1}{ms(1-s_1)}\sum_{j=1}^m \mathbbm{1}(A^2_{ij} \in \{\hat{g}_{j1}, \dots, \hat{g}_{jT}\})-\frac{T}{K}\right).
\eeq
We then apply the Maximum Likelihood Estimator (MLE) using $\hat{\bsp}$ and $\hat{\bsq}$. See Algorithm \ref{alg:plugin_MLE_T} for details. 
\begin{algorithm}[h]
	\caption{Plug-in MLE (Top-$T$2 Algorithm) }
	\label{alg:plugin_MLE_T}
\begin{algorithmic}[1]
	\STATE {\bfseries Input:} data matrix $\bsA\in\{0,1,\dots,K\}^{n\times m}$ and the sample splitting rate $s_1>0$.
	\STATE  Randomly split $\bsA$ into $\bsA^1$ and $\bsA^2$ by defining $\bsA^1:=\bsA \circ \bsS$ and $\bsA^2=\bsA \circ (\mathbbm{1}_{n \times m}-\bsS)$ 
	where  $\bsS$ is an $n \times m$ matrix whose entries are i.i.d. with Bern($s_1$) and $\circ$ is an entrywise product.
	\STATE Apply Algorithm \ref{alg:top_two} to $\bsA^1$ to yield estimates for top-T answers $\{\hat{g}_{j1}, \dots, \hat{g}_{jT}\}_{j=1}^m$ and confusion probability vector $\hat{\bsq}$.
	\STATE By using $\{\hat{g}_{j1}, \dots, \hat{g}_{jT}\}_{j=1}^m$ and $\bsA^2$, calculate the estimate $\hat{\bsp}$ as in \eqref{eqn:hatpi_T}.
	\STATE By using the whole $\bsA$ and $(\hat{\bsp},\hat{\bsq})$, find the plug-in MLE estimates $\{\hat{g}^{\mathsf{MLE}}_{j1}, \dots, \hat{g}^{\mathsf{MLE}}_{jT}\}_{j=1}^m$ by
	\beq\label{eqn:plug-in-mle_T}
        \argmax_{a_1, \dots, a_T \in[K]^T} \sum_{i=1}^n \sum_{t=1}^T \log \left(\frac{K\hat{p}_i \hat{q}_{jt}}{1-\hat{p}_i} + 1\right) \mathbbm{1}(A_{ij} = a_t) 
	\eeq
	\STATE {\bfseries Output:} estimated top-two answers $\{\hat{g}^{\mathsf{MLE}}_{j1}, \dots, \hat{g}^{\mathsf{MLE}}_{jT}\}_{j=1}^m$
\end{algorithmic}
\end{algorithm}

Although theoretical analysis needs to be changed accordingly, the model and algorithms can be easily extended to the general case of $T \ge 2$ plausible answers as above, since the binary-converted observation matrices still enjoy the rank-1 structure. Generalizing the theoretical analysis will be an interesting open problem.

\newpage
\section{Experimental Details for Neural Network Training}\label{sec:app:NN_soft}

We show the details of the experiments presented in Sec. \ref{sec:NN_soft}. 

\subsection{Datasets}
The CIFAR10H dataset \citep{peterson2019human} consists of 511,400 human classifications by 2,571 participants which were collected via Amazon Mechanical Turk. Each participant classified 200 images, 20 from each category. Every 20 tasks, a trivial question is presented to prevent random guessing, and participants who scored below 75\% were excluded from the dataset. We present the images with the lowest/highest $q$ from the training samples in Fig \ref{fig:CIFAR10Hlowhigh}. The image with a lower $q$ means that the first answer and the second answer are hard to distinguish.

\begin{figure*}[!tb]
\centering
	\subfloat[Images with lowest $q$ (considered to be hard)\label{fig:lowest_q}]{
	\includegraphics[width=0.45\textwidth]{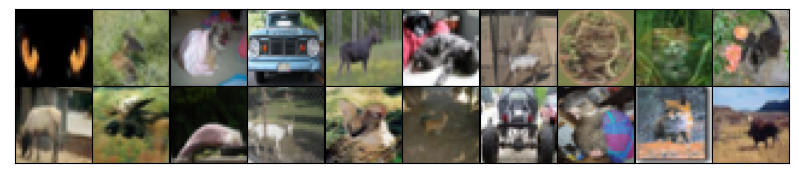}} \qquad
    \subfloat[Images with highest $q$ (considered to be easy)\label{fig:highest_q}]{ \includegraphics[width=0.45\textwidth]{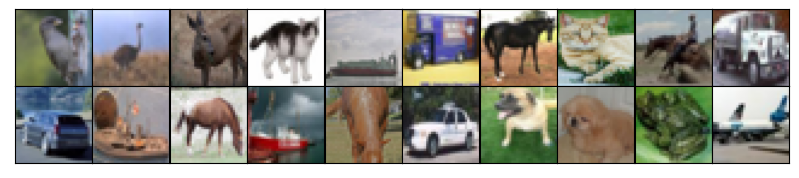}}
    \caption{Training images with (a) lowest and (b) highest confusion probabilities.}\label{fig:CIFAR10Hlowhigh}
\end{figure*}

\subsection{Model}
We train two simple CNN architectures, VGG-19 and ResNet-18, to show the usefulness of the second answer and the confusion probability. For each model, our loss function is defined as the cross-entropy between the softmax output and the two-hot vector (in which the values are $q$ and $1-q$ for $g$ and $h$, respectively). We compare the results of our top-two label training with those of full-distribution training and hard label (one-hot vector) training. 

\subsection{Training}
We train each model using 10-fold cross validation (using 90\% of images for training and 10\% images for validation) and average the results across 5 runs. We run a grid search over learning rates, with the base learning rate chosen from \{0.1, 0.01, 0.001\}. We find 0.1 to be optimal in all cases. We train each model for a maximum of 150 epochs using SGD optimizer with a momentum of 0.9 and a weight decay of 0.0001. Our neural networks are trained using NVIDIA GeForce 3090 GPUs.

\subsection{Training neural networks with corrupted CIFAR10H datasets}\label{sec:toptwo_t_robust}

The CIFAR10H dataset is collected from workers whose reliability is above 75\%, so that the full label distribution is in fact almost the same as the top-two distribution. To analyze the robustness against the label noise, we conduct an additional experiment by adding different portions of random responses to the original CIFAR10H dataset. In the experiment, we add the responses from spammers, who provide random labels on each image, to the original dataset, with the varying ratio of $[0.1, 0.2, 0.3, 0.4, 0.5]$. For example, if the ratio of spammers is 0.5, it means that we add the  same number of responses from spammers as the original dataset. The exact number of the added responses is
\beq
    \text{(\# of added responses)} = \frac{\text{(spammer ratio)}}{1 - \text{(spammer ratio)}} \times \text{(\# of total responses)}. 
\eeq

As in the experiments of Sec.\ref{sec:NN_soft}, we train two neural networks, ResNet18 and VGG-19, with the top-two label distribution and the full label distribution as the spammer ratio increases. Table~\ref{tab:corrupt_CIFAR10} shows the test accuracy of the trained neural networks. As shown in the table, the top-two label training outperforms the full label training in the high spammer ratio regime. This is because the training with the full label distribution tries to fit the model to all the collected answers, which include the responses from spammers. On the other hand, training with the top-two labels is more robust against the label noise, since it focuses on the simple yet meaningful side information, the ground-truth label and the most confusing label with the ratio between the two in the collected answers.

 \begin{table}[h]
 \centering
    \caption{Comparison of performances for the corrupted CIFAR10H dataset with top2/full label training}
    \label{tab:corrupt_CIFAR10}
    \vspace{0.2em}
    \renewcommand{\arraystretch}{1.2}
    \begin{tabular}{c | c c | c c}
        \toprule
        \multirow{2}{*}{spammer ratio} & \multicolumn{2}{c|}{ResNet18} & \multicolumn{2}{c}{VGG-19}\\ \cline{2-5} 
        & top-two & full & top-two & full\\
        \midrule
        0.1	& 80.18$\pm$1.30\% & \textbf{80.73$\pm$0.79\%} & \textbf{78.90$\pm$0.72\%} & 78.67$\pm$1.45\% \\
        0.2	& \textbf{80.30$\pm$1.81\%} & 79.79$\pm$0.59\% & \textbf{79.10$\pm$0.64\%} & 78.65$\pm$0.91\% \\
        0.3	& \textbf{79.80$\pm$0.44\%} & 79.23$\pm$0.79\% & \textbf{79.08$\pm$1.22\%} & 77.80$\pm$1.08\% \\
        0.4	& \textbf{79.05$\pm$0.78\%} & 76.82$\pm$0.75\% & \textbf{79.15$\pm$1.46\%} & 77.40$\pm$1.09\% \\
        0.5	& \textbf{78.40$\pm$0.96\%} & 75.88$\pm$0.93\% & \textbf{78.22$\pm$0.69\%} & 76.11$\pm$1.53\% \\
        \bottomrule
    \end{tabular}
\end{table}

\section{Baseline Methods}\label{app:baseline}

In this section, we explain the baseline methods with which we compare the performance of our algorithms.
To analyze the performance in recovering the top-two answers, we considered the ML-based algorithms, including the \textbf{Spectral-EM algorithm (MV-D\&S and OPT-D\&S)} \citep{SEM}, \textbf{Projected Gradient Descent (PGD)} \citep{PGD}, \textbf{M-MSR} \citep{M-MSR}, \textbf{MultiSPA}  \citep{M-MSR}, and \textbf{EBCC} \citep{li2019exploiting}, which provide a ``score'' for each label so that we can recover the top-two answers. 

\begin{itemize}
    \item \textbf{Spectral-EM algorithm (MV-D\&S and OPT-D\&S)} \citep{SEM} is a two-stage algorithm for multi-class crowd labeling problems. These algorithms are built for the D\&S model where each worker has his/her own confusion matrix. In the first stage of the algorithm, the confusion matrix of each worker is estimated via spectral method (OPT-D\&S) or majority voting (MV-D\&S), respectively, and in the second stage, the estimates for the confusion matrices are refined by optimizing the objective function of the D\&S estimator via the Expectation Maximization (EM) algorithm.
        \item\textbf{Projected Gradient Descent (PGD)} \citep{PGD} is an approach to estimate the skills of each worker in the single-coin D\&S model. The authors formulate the skill estimation problem as a rank-one correlation-matrix completion problem. They propose a projected gradient descent method to solve the correlation-matrix completion problem. 



    \item \textbf{M-MSR} \citep{M-MSR} algorithm is an approach to estimate the reliability of each worker in the multi-class D\&S model. M-MSR algorithm utilizes that the rank of the response matrix is one. To estimate the reliability of the workers, they use update rules to find the left singular vector and right singular vector of the response matrix. In this process, the extreme values are filtered out to guarantee the stable convergence of the algorithm.

    \item \textbf{MultiSPA-EM} \citep{ibrahim2019crowdsourcing} is an approach to estimate  each worker's confusion matrix using pairwise co-occurrence matrix. To estimate the confusion matrices, three SPA (successive projection algorithm)-based algorithms are proposed; MultiSPA, MultiSPA-KL and MultiSPA-EM. MultiSPA utilizes the second order statistics to obtain the confusion matrices and the ground truth. MultiSPA-KL is an iterative optimization method to minimize the KL-divergence between the expectation of the co-occurrences and the empirical co-occurrences, where the initial estimates are obtained from MultiSPA. MultiSPA-EM is an EM based algorithm where the initial estimates are obtained from MultiSPA. Since the MultiSPA-EM outperforms MultiSPA and MultiSPA-KL, we only include these two in our baselines.
    
    \item \textbf{EBCC} \citep{li2019exploiting} algorithm is an enhanced version of the Baysian classifier combination model. The authors assume that each label has its own subtypes. Each subtype has different probability distribution even if the label is the same. EBCC algorithm utilizes the Expectation-Maximization (EM) algorithms to recover the hidden variables and estimates the true labels. 
    
\end{itemize}

\section{Synthetic Experiments}\label{app:synthetic_exp}
\begin{figure}[h]
    \centering
    \includegraphics[width=\textwidth]{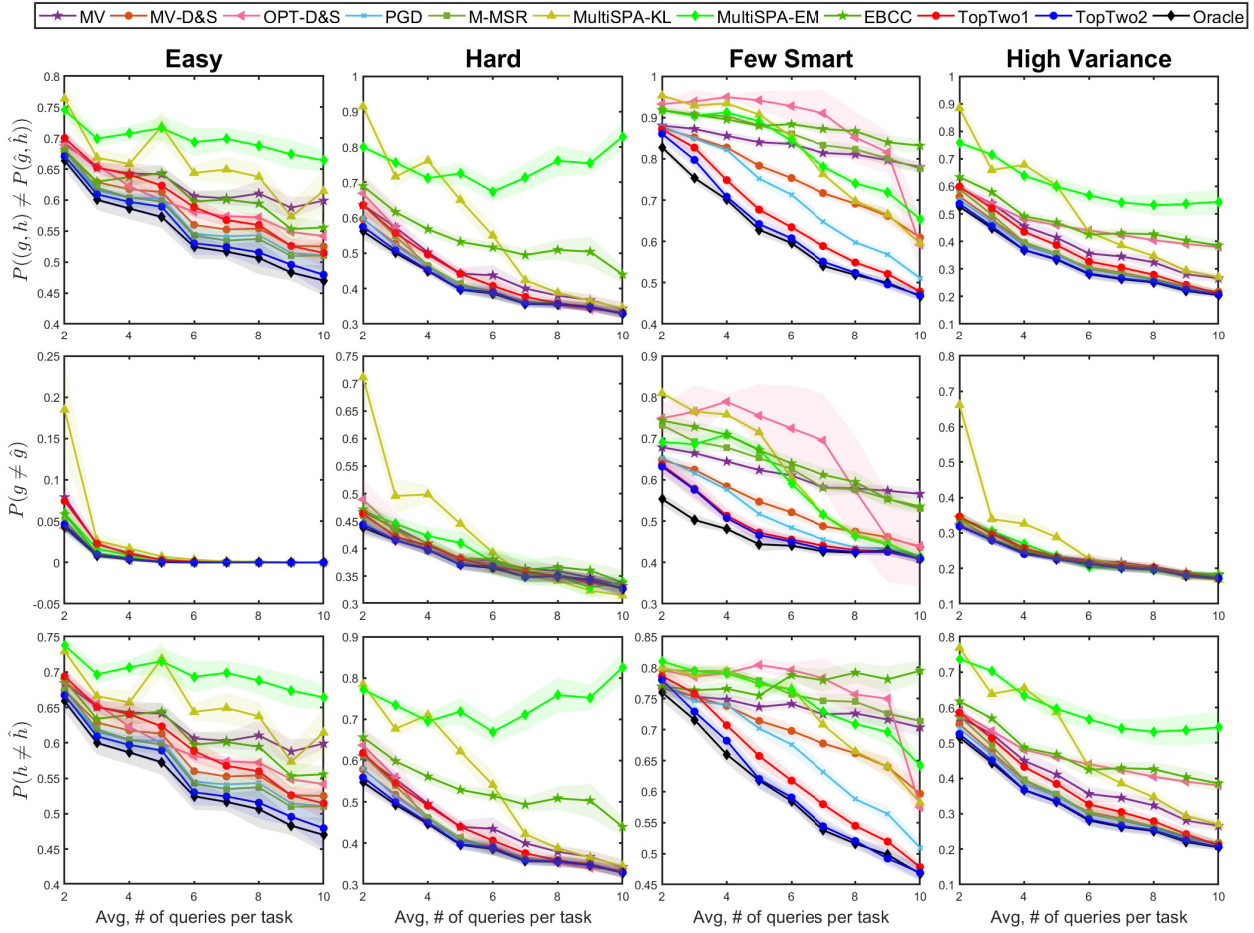}
    \caption{Prediction error for $(g_j,h_j)$ (top row), $g_j$ (middle) and $h_j$ (bottom) for four scenarios. Our algorithm (TopTwo2) achieves the best performance, near the oracle MLE for all the scenarios. }
    \label{fig:app:synthetic_acc}
\end{figure}

\subsection{Additional plots for synthetic data experiments in Sec. \ref{sec:synth}}\label{sec:app:more_plots}
In Section \ref{sec:synth}, we devised four scenarios described in Table \ref{tab:synth_setup1} to verify the robustness of our model for various $(\bsp,\bsq)$ ranges, with $(n,m,s)=(50,500,0.2)$. 
The performance of algorithms is measured by the empirical average error probabilities in recovering $g_j$, $h_j$ and $(g_j,h_j)$, i.e.,  $\frac{1}{m}\sum_{j=1}^m \P(\hat{g}_j\neq g_j)$, $\frac{1}{m}\sum_{j=1}^m \P(\hat{h}_j\neq h_j)$ and $\frac{1}{m}\sum_{j=1}^m \P((\hat{g}_j,\hat{h}_j)\neq (g_j,h_j))$ and plotted in Fig. \ref{fig:app:synthetic_acc}.
We can observe that for all the considered scenarios TopTwo2 achieves the best performance, near the oracle MLE, in recovering $(g_j,h_j)$. Depending on scenarios though, the reason TopTwo2 outperforms can be explained differently. For \underline{Easy} scenario, since $q_j$ is close to 1, it becomes easy to distinguish $g_j$ from $h_j$ but hard to distinguish $h_j$ from other labels. Our algorithm achieves the best performance in estimating $h_j$ by a large margin. For \underline{Hard} scenario, it becomes hard to distinguish $g_j$ and $h_j$, but our algorithm, which uses an accurate $\hat{q}_j$, can better distinguish $g_j$ and $h_j$. 
For \underline{Few-smart}, our algorithm achieves the largest gain compared to other methods, since our algorithm can effectively distinguish few smart workers from spammers.  \underline{High-variance} show the effect of having diverse $q_j$ in a dataset. 
\subsection{Robustness of our methods}
\begin{figure}
    \centering
    \subfloat[Effect of the number of workers on the performance\label{fig:c_workers}]{\includegraphics[width=0.645\textwidth]{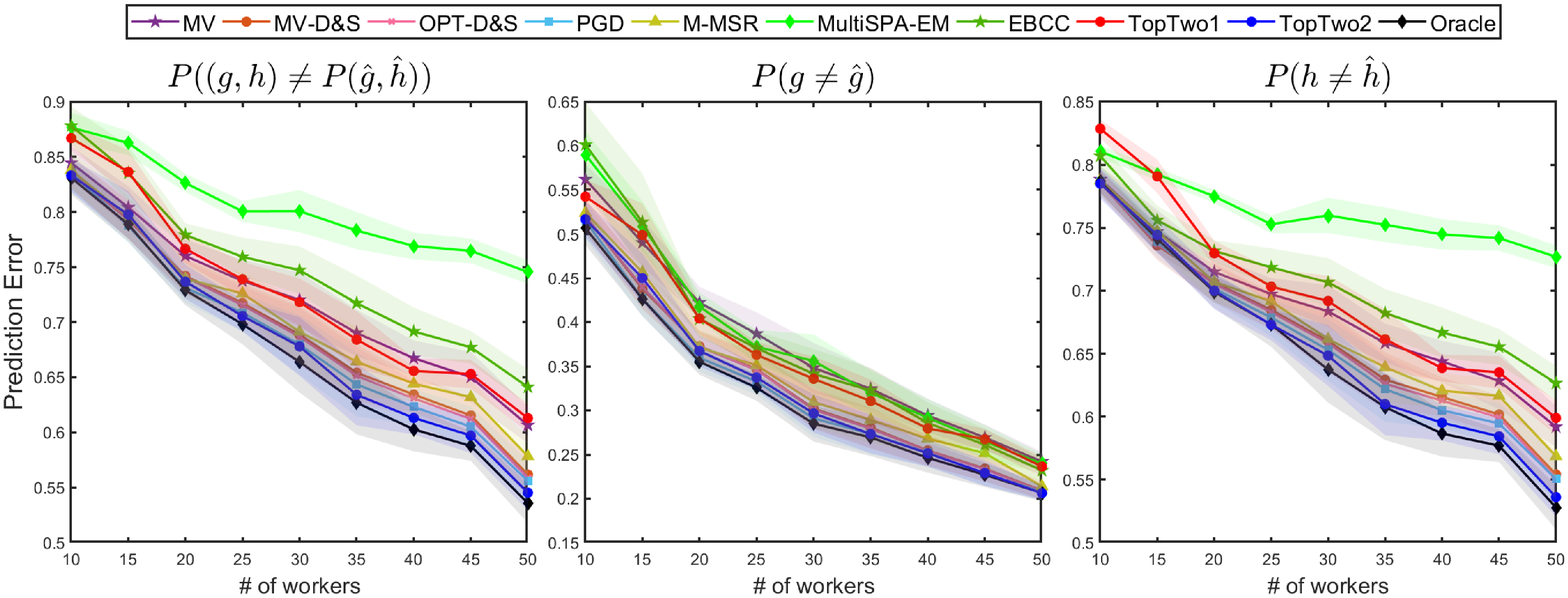}}  \vspace{-1em}\\
    \subfloat[Effect of the number of tasks on the performance\label{fig:c_tasks}]{\includegraphics[width=0.65\textwidth]{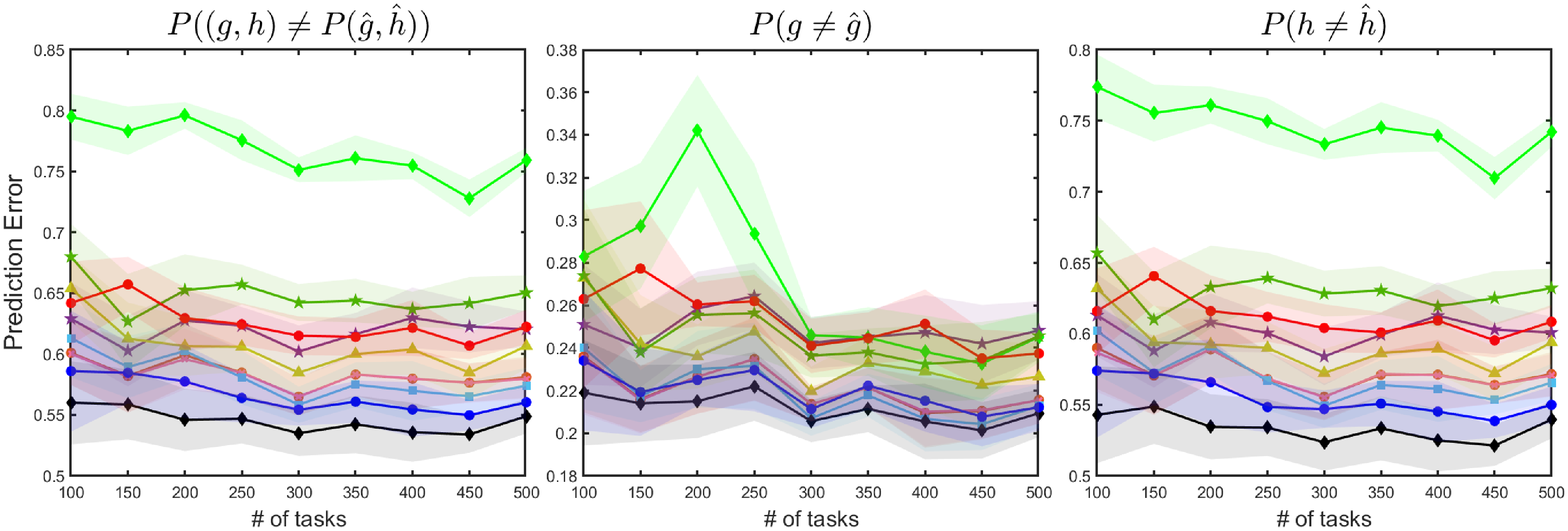}}  \vspace{-1em}\\
    \subfloat[Effect of the variance of worker reliability on the performance\label{fig:c_std_worker}]{\includegraphics[width=0.65\textwidth]{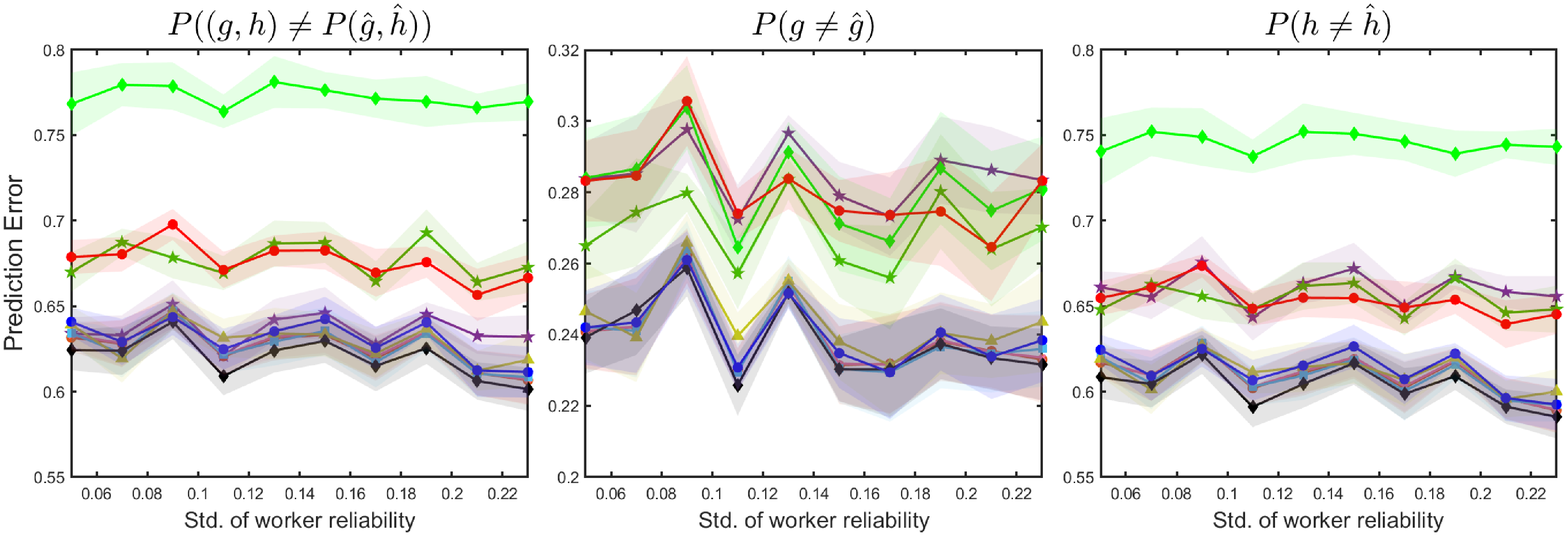}}  \vspace{-1em}\\
    \subfloat[Effect of the variance of task difficulty on the performance\label{fig:c_std_task}]{\includegraphics[width=0.65\textwidth]{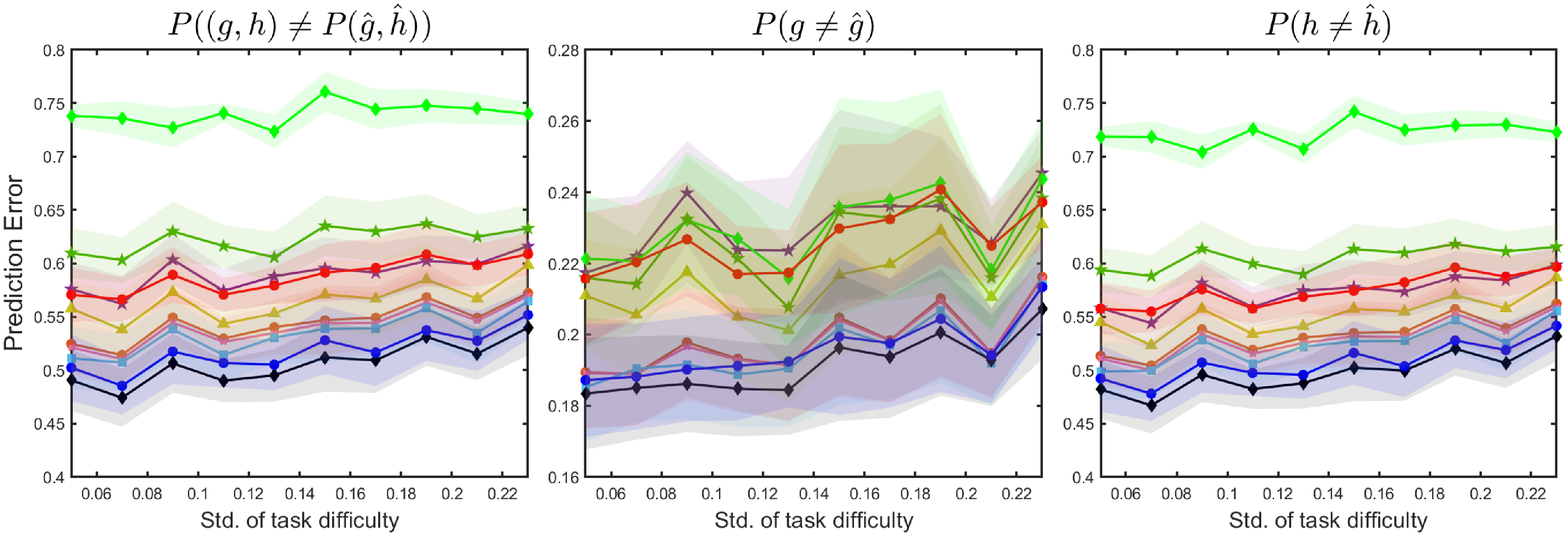}}  \vspace{-1em}\\
    \subfloat[Effect of the portion of spammers on the performance\label{fig:c_spammers}]{\includegraphics[width=0.65\textwidth]{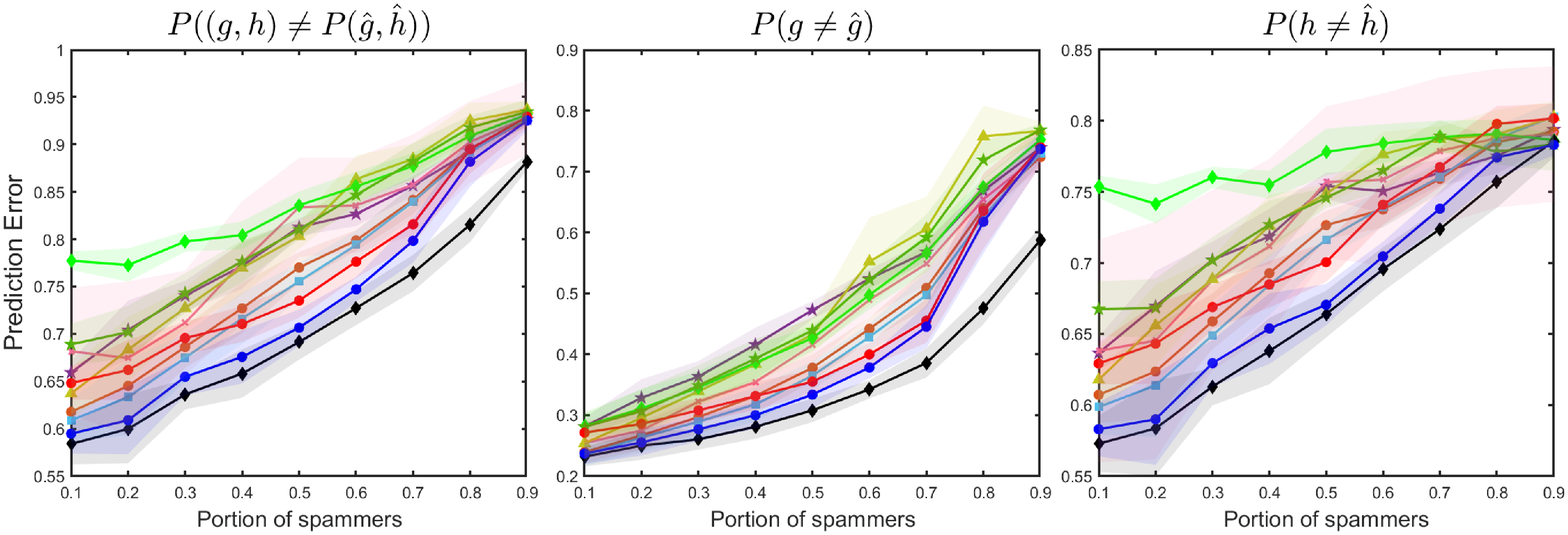}} 
    \caption{Prediction error for $(g_j, h_j)$ (first column), $g_j$ (second column), and $h_j$ (third column) for five different setups. The solid lines represent the mean prediction errors of each algorithm averaged over 10 runs, and the shaded regions represent the standard deviations. }
    \label{fig:app_exp}
\end{figure}
In this section, we present a set of four additional synthetic experiments to demonstrate the robustness of our methods, Alg. \ref{alg:top_two} and Alg. \ref{alg:plugin_MLE} (referred to as TopTwo1 and TopTwo2). In each experiment, we change a parameter of our synthetic error model and compare the prediction error of our algorithms to the baselines: majority voting(MV), MV-D\&S \cite{SEM}, PGD \cite{PGD}, MultiSPA-KL and MultiSPA-EM\cite{ibrahim2019crowdsourcing}, EBCC\cite{li2019exploiting} and Oracle-MLE. We measure the performance of each algorithm by the empirical average error probabilties in recovering the ground truth $g_j$, the most confusing answer $h_j$ and the pair of top two $(g_j,h_j)$, i.e.,  $\frac{1}{m}\sum_{j=1}^m \P(\hat{g}_j\neq g_j)$, $\frac{1}{m}\sum_{j=1}^m \P(\hat{h}_j\neq h_j)$ and $\frac{1}{m}\sum_{j=1}^m \P((\hat{g}_j,\hat{h}_j)\neq (g_j,h_j))$.
Obviously, Oracle-MLE provides a lower bound for the performance. 

\textbf{Changing the dimension of observed matrix}: We first check the robustness of our methods against the change of dimensions of the observation matrix $\bsA\in\{0,1\dots,K\}^{n\times m}$ with $n\leq m$. We vary the number of workers ($n$) or the number of tasks ($m$) while fixing the other dimension. The default values of $n$ and $m$ are 50 and 500, respectively, and the sampling probability $s$ is fixed as $0.1$ throughout the experiments. The worker reliability $p_i$ and the task difficulty $q_j$ is sampled uniformly at random from $[0,1]$ and $(1/2,1]$, respectively, for all $i\in[n]$ and $j\in[m]$.

In Fig. \ref{fig:c_workers} and \ref{fig:c_tasks}, we report the results when we change $n$ for a fixed $m$ and $s$, or when we change $m$ for a fixed $n$ and $s$, respectively.
From Fig. \ref{fig:c_workers}, we can see that as the number of workers increases, the performance of every algorithm improves since the number of samples per task scales as $ns$ for a fixed $s$.
Our algorithm achieves the performance close to the Oracle-MLE for all the considered range, which implies that the worker reliabilities $\{p_i\}$ are well estimated with our methods. 
From Fig. \ref{fig:c_tasks}, we can see that our algorithm achieves a robust performance against the change in the number of tasks, although the performance gets closer to that of Oracle-MLE as the number of tasks increases. Since our method uses SVD in the first stage, the larger dimension is beneficial for the concentration of the random perturbation matrix with respect to the expectation of the observation matrix.  This phenomenon is observed for other baseline methods as well, which are based on the spectral method. 

\textbf{Changing the variance of worker reliability}: 
In this experiment, we change the range of $p_i$, the parameter for worker skill/reliability, for $i\in[n]$, with a fixed mean in order to observe the impact of the variance of the worker reliability on the prediction error. We randomly sample $p_i$ from the window $[0.5-x,0.5+x]$ with $x$ varying from $0.05$ to $0.25$. The mean of the worker reliability is fixed as $0.5$.

As shown in Fig.~\ref{fig:c_std_worker}, when the variance of the worker reliability increases, the baseline methods estimating worker reliabilities perform better than the majority voting. Our TopTwo2 algorithm achieves the best performance close to Oracle-MLE, as the standard deviation increases, i.e., as the workers become more heterogeneous.

\textbf{Changing the variance of task difficulty}: We also design an experiment to observe the impact of the variance of $q_j$, $j\in[m]$, the parameter for task difficulty, on the prediction error. 
We randomly sample $q_j$ from the window $[0.75-x,0.75+x]$ with $x$ varying from $0.05$ to $0.25$. The mean of the worker reliability is fixed as $0.75$.
If the variance of the task difficulty is small, it could be sufficient to only estimate the worker reliability since all the tasks have almost the similar task difficulties. 

As shown in Fig.~\ref{fig:c_std_task}, when the variance of the task difficulty increases, our TopTwo2 algorithm performs better than the other baselines. This is the evidence for the validity of our method in estimating the task difficulty. 

\textbf{Changing the portion of spammers}: Spammers who provide random answers always exist in crowdsourcing systems. To improve the inference performance, it is important to distinguish spammers from reliable workers. In our experimental setup, we define a spammer as a worker whose reliability parameter $p_i$ is in the range $[0, 0.1]$. We change the portion of spammers among the workers from $0.1$ to $0.9$ and compare the prediction error of our methods to those of other baseline methods. 

In Fig. \ref{fig:c_spammers}, we can see that  our algorithm achieves the best performance among all the considered baselines except Oracle-MLE, which can exactly distinguish spammers from reliable workers. This result demonstrates the superiority of our methods in detecting spammers compared to other methods.

\subsection{Estimating the worker reliability vector and the task difficulty vector}

In this section, we examine the accuracy of our estimates for the worker reliability vector $\bsp$ and the task difficulty vector $\bsq$.
The worker reliability is estimated by $\hat{\bsp}$ defined in \eqref{eqn:hatpi} of Algorithm \ref{alg:plugin_MLE} and the task difficulty is estimated by $\hat{\bsq}$ defined in \eqref{eqn:est_q} of Algorithm \ref{alg:top_two}.
To analyze the accuracy of these estimators, we compute the mean squared error (MSE), $\frac{1}{n}\|\hat{\bsp}-\bsp\|_2^2$ and $\frac{1}{m}\|\hat{\bsq}-\bsq\|_2^2$, respectively. 

To analyze the estimation accuracy for the worker reliability, we first sample $p_i$ uniformly at random from $[0, 1]$ for all $i\in[n]$ and fix the worker reliability vector $\bsp$. Then, we randomly sample the task difficulty vector $\bsq\in(1/2,1]^m$ fifty times and then sample the observation matrices from the distribution \eqref{eqn:A_dist} for each $(\bsp,\bsq)$ pair with a fixed $\bsp$. For each observation matrix, we subsample the data with varying probabilities and  apply Algorithm \ref{alg:plugin_MLE} to get the estimate $\hat{\bsp}$, which is then used to calculate the MSE of $\bsp$. We report the MSE averaged over these fifty cases. Similarly, to analyze the estimation accuracy for the task difficulty, we randomly sample and fix a task difficulty vector $\bsq\in(1/2,1]^m$ and generate fifty different observation matrices while varying the worker reliability vector $\bsp$. We again report the MSE averaged over these fifty cases. The number of workers and that of tasks is set to be $(50, 500)$ for the worker reliability estimation, and to be $(100, 1000)$ for the task difficulty estimation. 

In Fig. \ref{fig:estm_p} and \ref{fig:estm_q}, we plot the MSE for $\bsp$ and $\bsq$, respectively, as the average number of queries per task increases. 
We can see that both for $\bsp$ and $\bsq$, the MSEs converge to near zero as the average number of queries per task increases.
However, estimating the task difficulty requires more number of samples as our theory \eqref{eqn:total_s_pq} suggests.


\begin{figure}[thb]
    \centering
    \subfloat[Mean squared error $\frac{1}{n}\|\hat{\bsp}-\bsp\|_2^2$\label{fig:estm_p}]{\includegraphics[width=0.4\textwidth]{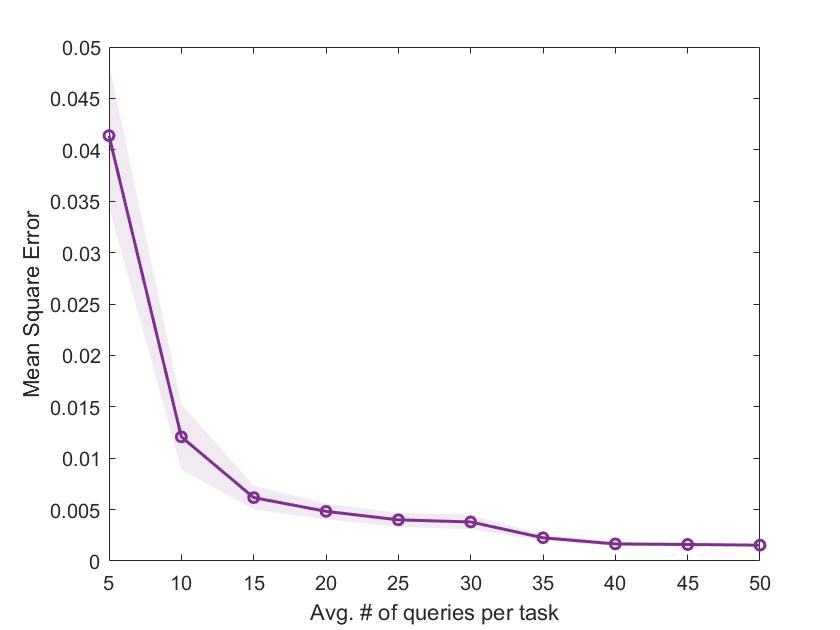}}
    \subfloat[Mean squared error  $\frac{1}{m}\|\hat{\bsq}-\bsq\|_2^2$ \label{fig:estm_q}]{\includegraphics[width=0.4\textwidth]{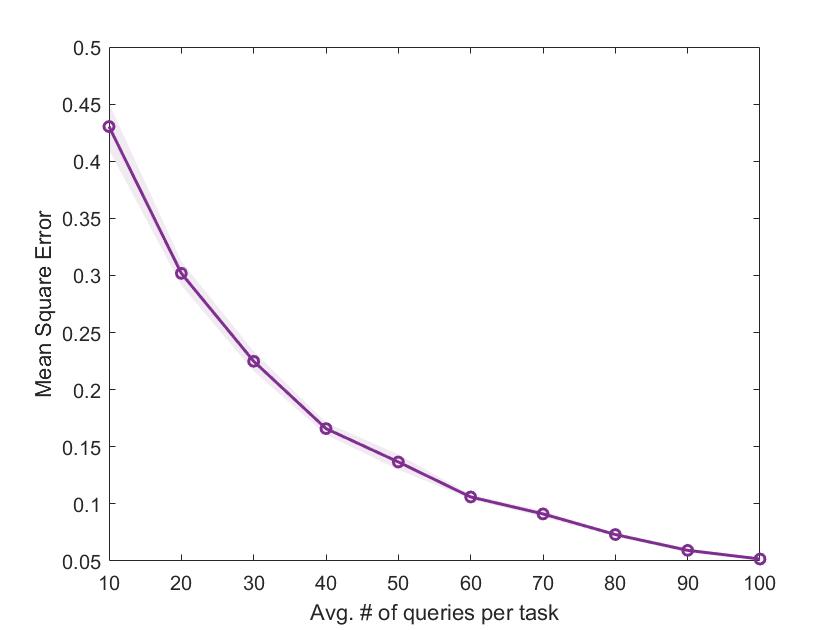}}
    \caption{Mean squared errors in estimating the worker reliability vector $\bsp$ (left) and the task difficulty vector $\bsq$ (right), respectively. }
    \label{fig:estm_pq}
\end{figure}

\section{Discussion of theoretical results}\label{app:diss:theory}

In this section, we present a discussion of the main theoretical results.
\begin{itemize}
\item Theorem \ref{thm:main} asserts that the sampling probability of $\Omega\left(\frac{1}{\delta_1^2 \|\bsp\|_2^2}\log\frac{K}{\epsilon}\right)$ is sufficient to recover the top-two answers $(g_j,h_j)$ for any task $j\in[m]$ and to estimate the confusion probability $q_j$ with accuracy of $|\hat{q}_j-q_j|<\delta_1$ by Algorithm \ref{alg:top_two} with probability at least $1-\epsilon$. Combined with Theorem \ref{thm:converse} part (a), we can see that this sample complexity is the minimax optimal rate for a fixed collective quality of workers, measured by $\|\bsp\|_2^2$. 

\item  Theorem \ref{thm:plugmin_mle} shows that when we have an entrywise bound on the estimated worker reliability vector $\bsp$ and the task difficulty vector $\bsq$, the plug-in MLE estimator, used in Algorithm \ref{alg:plugin_MLE}, guarantees the recovery of top-two answers if the sampling probability $s=\Omega(\frac{\log(m/\epsilon)}{n\bar{D}})$ where $\bar{D}$, which depend on $(\bsp,\bsq)$, indicates the average reliability of workers in distinguishing the top-two answers from any other pairs for the most difficult task. Combined with Theorem \ref{thm:converse} part (b), we can see that this sample complexity is the minimax optimal rate for any $(\bsp,\bsq)$, ignoring the logarithmic terms.
\item  Combining the conditions for the accurate estimation of model parameters in \eqref{eqn:total_s_pq} and the convergence of the plug-in MLE (Theorem \ref{thm:plugmin_mle}), Corollary \ref{cor:alg2} shows the condition on the sample complexity to guarantee the performance of Algorithm \ref{alg:plugin_MLE}. 
\end{itemize}

\section{Proof of Proposition \ref{prop:binary}}\label{app:prop:binary}

For each task $j$ and label $k$, define four indicator functions: 
\beq\label{eqn:defn_pi}
\begin{split}
\Pi_a(j,k):=&\mathbbm{1}(g_j>k,h_j>k),\\
\Pi_b(j,k):=&\mathbbm{1}(g_j\leq k,h_j>k),\\
\Pi_c(j,k):=&\mathbbm{1}(g_j> k,h_j\leq k),\\
\Pi_d(j,k):=&\mathbbm{1}(g_j\leq k,h_j\leq k),
\end{split}
\eeq
which satisfy $\Pi_a(j,k)+\Pi_b(j,k)+\Pi_c(j,k)+\Pi_d(j,k)=1$. For notational simplicity, we will often drop $(j,k)$ fron $\Pi_*$.
The pmf of $\bsA^{(k)}$ is given by
\beq
A_{ij}^{(k)}=
\begin{cases}
-1 &\text{with probability }s(1-\rho_{ij}^{(k)}),\\
1&\text{with probability } s\rho_{ij}^{(k)},\\
0&\text{with probability }1-s,
\end{cases}
\eeq
where $\rho_{ij}^{(k)}=\Pi_a (j,k)p_i+\Pi_b(j,k) p_i(1-q_j)+\Pi_c(j,k) p_iq_j+\frac{(K-k)(1-p_i)}{K}$, and its expectation is
$
\E[A_{ij}^{(k)}]=s(2\rho_{ij}^{(k)}-1).
$
Note that by using  $\Pi_a=1-\Pi_b-\Pi_c-\Pi_d$, the probability $\rho_{ij}^{(k)}$ can be written as
$
\rho_{ij}^{(k)}=p_i\left(q_j(\Pi_c-\Pi_b)-(\Pi_c+\Pi_d)+\frac{k}{K}\right)+\frac{K-k}{K}.
$
Thus, by defining
\beq\label{eqn:def_rjk}
r_{j}^{(k)}:=q_j(\Pi_c-\Pi_b) -(\Pi_c+\Pi_d)+\frac{k}{K},
\eeq
the expectation of $A_{ij}^{(k)}$ can be written as
\beq
\E[A_{ij}^{(k)}]=s(2\rho_{ij}^{(k)}-1)=s\left(2p_ir_j^{(k)}+\frac{K-2k}{K}\right),
\eeq
and
\beq\label{eqn:rank1}
\E[\bsA^{(k)}]-\frac{s(K-2k)}{K}\mathbbm{1}_{n \times m}=2s\bsp (\bsr^{(k)})^\top.
\eeq

Note that
\begin{equation}\label{eqn:defpicase}
 \begin{split}
&\text{Case I: $g_j>h_j$} \\
&\Pi_a(j,k)=1 \text{ where }k<h_j,\\
&\Pi_c(j,k)=1 \text{ where }h_j\leq k< g_j,\\
&\Pi_d(j,k)=1 \text{ where }g_j \leq k;
\end{split}\qquad\qquad
 \begin{split}
&  \text{Case II: $g_j<h_j$} \\
&\Pi_a(j,k)=1 \text{ where }k<g_j,\\
&\Pi_b(j,k)=1 \text{ where }g_j\leq k< h_j,\\
&\Pi_d(j,k)=1 \text{ where }h_j \leq k.
\end{split}\end{equation}
Thus, $r_j^{(k)}$ in \eqref{eqn:def_rjk} is equal to
\begin{equation*}
 \begin{split}
&\text{Case I: $g_j>h_j$} \\
&r_j^{(k)}=\begin{cases}
\frac{k}{K}& \text{ where }k<h_j ; \\
\frac{k}{K}-(1-q_j)& \text{ where }h_j\leq k< g_j;  \\
\frac{k}{K}-1& \text{ where }g_j \leq k,
\end{cases}
\end{split}\quad\quad
 \begin{split}
&  \text{Case II: $g_j<h_j$} \\
&r_j^{(k)}=\begin{cases}
\frac{k}{K} &\text{ where }k<g_j;\\
\frac{k}{K}-q_j &\text{ where }g_j\leq k< h_j;\\
\frac{k}{K}-1& \text{ where }h_j \leq k.
\end{cases}
\end{split}\end{equation*}

\section{Performance Analysis of Algorithm \ref{alg:top_two}}\label{app:thm:main}

\subsection{Proofs of Theorem \ref{thm:main} and Corollary \ref{cor:exact_alg1}}
In Algorithm \ref{alg:top_two}, we use the data matrix $\bsA^1$, which is obtained by randomly splitting the original data matrix $\bsA$ into $\bsA^1$ and $\bsA^2$ with probability $s_1$ and $(1-s_1)$, respectively. 
Then, the first stage of Algorithm \ref{alg:top_two} begins with randomly splitting  $\bsA^1$ again  into two independent matrices $\bsB$ and $\bsC$ with equal probabilities, 
and then converting $\bsB$ and $\bsC$ into $(K-1)$-binary matrices $\bsB^{(k)}$ and $\bsC^{(k)}$ as explained  in Sec. \ref{sec:model}. 
We define $\bsX^{(k)}$ and $\bsY^{(k)}$ as 
$ \bsX^{(k)} := \bsB^{(k)} - \frac{s'(K-2k)}{K}\mathbbm{1}_{n \times m}$ and $\bsY^{(k)} := \bsC^{(k)} - \frac{s'(K-2k)}{K}\mathbbm{1}_{n \times m}$ where $s'=s\cdot s_1/2$.
We have
$
\E[\bsX^{(k)}]=\E[\bsY^{(k)}]=s'\bsp (\bsr^{(k)})^\top
$ from Prop. \ref{prop:binary}. For notational simplicity, we will ignore this random splitting for a moment and just pretend that $\bsX^{(k)}$ and $\bsY^{(k)}$ are sampled independently with $s'=s$ throughout this section. 

We first outline the proof. 
Based on the observation that $\E[\bsX^{(k)}]=s\bsp (\bsr^{(k)})^\top$, if $\E[\bsX^{(k)}]$ is available we can recover $\bsp^*=\frac{\bsp}{\|\bsp\|_2}$ by SVD, and by using $\bsp^*$ it is possible to recover $\|\bsp\|_2 \bsr^{(k)}$, which then reveals $\{(g_j,h_j)\}_{j=1}^m$ as well as $\bsq$ from the relation in \eqref{eqn:delta_r}.
To estimate $\bsp^*$ from $\bsX^{(k)}$, we first bound the spectral norm of the perturbation, $\|\bsX^{(k)}-\E[\bsX^{(k)}]\|_2$. We then use this bound and Wedin Sin$\Theta$ theorem to bound $\sin\theta(\bsu^{(k)}, \bsp^*)$ where $\bsu^{(k)}$ is the left singular vector of $\bsX^{(k)}$ with the largest singular value. We trim the abnormally large components of $\bsu^{(k)}$ and denote the resulting vector by $\tilde{\bsu}^{(k)}$. After trimming, it is still possible to show that $\sin\theta(\tilde{\bsu}^{(k)}, \bsp^*)$ can be bounded in the same order as that of $\sin\theta(\bsu^{(k)}, \bsp^*)$. Finally, we provide an entrywise bound between $\bsv^{(k)}=\frac{2}{s}(\bsY^{(k)})^{\top}\tilde{\bsu}^{(k)}$ and $\|\bsp\|_2\bsr^{(k)}$ in Lemma \ref{lem:entrywisebd_new}, which is the main lemma to prove Theorem \ref{thm:main}.  We state our main technical lemmas first and then prove Theorem \ref{thm:main}.

Let us define the perturbation matrix
\beq\label{eqn:perturbation}
\bsE:=\bsX^{(k)}-\E[\bsX^{(k)}]=\bsB^{(k)}-\frac{s(K-2k)}{K}{\mathbbm{1}_{n\times m}}-s\bsp (\bsr^{(k)})^\top=\bsB^{(k)}-\E[\bsB^{(k)}]
\eeq 
where 
\beq
B_{ij}^{(k)}=
\begin{cases}
-1 &\text{w.p. }s(1-\rho_{ij}^{(k)}),\\
1&\text{w.p. } s\rho_{ij}^{(k)},\\
0&\text{w.p. }1-s,
\end{cases}
\eeq
and $\rho_{ij}^{(k)}=\Pi_a (j,k)p_i+\Pi_b(j,k) p_i(1-q_j)+\Pi_c(j,k) p_iq_j+\frac{(K-k)(1-p_i)}{K}$ for $(\Pi_a,\Pi_b,\Pi_c,\Pi_d)$ defined in \eqref{eqn:defn_pi}.

For the perturbation matrix $\bsE$ in \eqref{eqn:perturbation}, we have
\beq
\E[E_{i,j}]=0,\quad\text{and}\quad |E_{i,j}|\leq 2, \quad1\leq i\leq n,\;\;1\leq j\leq m,
\eeq
and also
\beq
\begin{split}
\var(E_{ij})&=\var(B^{(k)}_{ij})=\E[(B_{ij}^{(k)})^2]-(\E[B_{ij}^{(k)}])^2\\
&=s-(s(\rho_{ij}^{(k)}-1/2))^2\leq s.
\end{split}	
\eeq
Note that $\{E_{ij}\}$ are independent across all $i,j$.
Define
\beq
\nu:=\max\left\{\max_i \sum_j\E[E_{i,j}^2],\;\;{\max_j \sum_i\E[E_{i,j}^2]}\right\}\leq \max\{m,n\}s.
\eeq
By applying the spectral norm bound to random matrices with independent entires, appeared in \cite{bandeira2016sharp} and summarized in Theorem \ref{lem:bd_spnorm}, we can bound the spectral norm of $\bsE$ as below.
\begin{lem}[Spectral norm bound of $\bsE$]\label{lem:spnorm}
With probability $1-(n+m)^{-8}$, we have
\beq
\|\bsE\|\leq 4\sqrt{s\max{(m,n)}}+\tilde{c}\sqrt{\log (n+m)}
\eeq
for some constant $\tilde{c}>0$ when $m\geq n$. For some sufficiently large $m$, assuming $n=o(m)$ and $s=\Omega(\log(n+m)/m)$, the spectral norm of $\bsE$ can be further bounded by
\beq\label{eqn:bd_spnorm_E}
\|\bsE\|\leq 5\sqrt{sm}.
\eeq 
\end{lem}

Using the bounded spectral norm of $\bsE$ in \eqref{eqn:bd_spnorm_E} and applying the Wedin Sin$\Theta$ theorem, summarized in Theorem \ref{thm:WedinSin}, we can bound the angle between $\bsu^{(k)}$ and $\bsp^*$.

\begin{lem}\label{lem:sin_uk_p} For some sufficiently large $m$, assuming $n=o(m)$ and $s=\Omega(\log(n+m)/m)$, we have
\beq\label{eqn:sin_uk_p}
\sin\theta(\bsu^{(k)},\bsp^*)\leq \Theta(1/\sqrt{sn})
\eeq
with probability at least $1-(n+m)^{-8}$.
\end{lem}
\begin{proof}
By applying the Wedin Sin$\Theta$ Theorem (Theorem \ref{thm:WedinSin}), we have
\beq 
\sin\theta(\bsu^{(k)},\bsp^*)\leq\frac{\sqrt{2}\|\bsE\|}{s\|\bsp\|_2\cdot\|\bsr^{(k)}\|_2-\|\bsE\|}.
\eeq
We have $\|\bsp\|_2=\Theta(\sqrt{n})$ and $\|\bsr^{(k)}\|_2=\Theta(\sqrt{m})$ by assumptions on model parameters.  By Lemma~\ref{lem:spnorm}, for some sufficiently large $m$, assuming $n=o(m)$ and $s=\Omega(\log(n+m)/m)$, we have $\|\bsE\|\leq 5\sqrt{sm}$ with probability at least $1-(n+m)^{-8}$. Combining these bounds, we get
\beq
\sin\theta(\bsu^{(k)},\bsp^*)\leq\frac{\Theta(\sqrt{sm})}{\Theta(s\sqrt{mn})-\Theta(\sqrt{sm})}=\frac{1}{\Theta\left(\sqrt{{sn}}\right)}.
\eeq

\end{proof}

We trim the abnormally large components of $\bsu^{(k)}$ by letting it zero if $u_i^{(k)}>{2}/({\eta\sqrt{n}})$ and denote the resulting vector as  $\tilde{\bsu}^{(k)}$.
This process is required to control the maximum entry size of $\tilde{\bsu}^{(k)}$, which is used later in the proof.
For the next lemma, we show that after the trimming process, the norm of $\tilde{\bsu}^{(k)}$ is still close to 1 and the angle between $\tilde{\bsu}^{(k)}$ and $\bsp^*$ has the same order as that of $\sin\theta({\bsu}^{(k)},\bsp^*)$.
\begin{lem}\label{lem:sin_uktilde_p} Given $\|\bsp^*\|_2\geq \eta\sqrt{n}$, we have
\begin{align}
\|\tilde{\bsu}^{(k)}\|_2&\geq\sqrt{1-50\sin^2\theta(\bsu^{(k)},\bsp^*)},\label{lem:utilde_norm}\\
\sin\theta({\tilde{\bsu}}^{(k)},\bsp^*)&\leq 6\sqrt{2} \sin\theta({\bsu}^{(k)},\bsp^*).\label{lem:utilde_sin}
\end{align}
\end{lem}
The proof of Lemma \ref{lem:sin_uktilde_p} is provided in Section \ref{app:sec:lem:sin_uktilde}.

Finally, we provide our main lemma giving the entrywise bound on the difference between $\bsv^{(k)}=\frac{1}{s}(\bsY^{(k)})^{\top}\tilde{\bsu}^{(k)}$ and $\|\bsp\|_2\bsr^{(k)}$.
\begin{lem}[Entrywise Bound]\label{lem:entrywisebd_new}
For any $\delta_1, \epsilon>0$, and any task $j\in[m]$ and label index $k\in[K]$, if the sampling probability 
$
s\geq \Theta\left(\frac{1}{\delta_1^2\|\bsp\|_2^2}\log\frac{1}{\epsilon}\right),$
then we can guarantee 
\beq\label{eqn:entry_bd1}
\P\left(\left| \frac{1}{s}\left\langle \bsY^{(k)}_{*j}, \tilde{\bsu}^{(k)}\right\rangle -\|\bsp\|_2 r_{j}^{(k)} \right|<\delta_1\|\bsp\|_2\right)>1-\epsilon
\eeq
as $m\to\infty$ when $n=O(m/\log m)$.
\end{lem}
\begin{proof}
For notional simplicity, denote $\theta(\tilde{\bsu}^{(k)},\bsp^*)$ by $\theta$.
To prove~\eqref{eqn:entry_bd1}, we show bounds on two probabilities,
\begin{align}
 \P\left(\left|\frac{1}{s}\left\langle \bsY^{(k)}_{*j}, \tilde{\bsu}^{(k)}\right\rangle-\|\tilde{\bsu}^{(k)}\|_2\|\bsp\|_2 r_{j}^{(k)}\cos\theta \right|>\frac{\delta_1\|\bsp\|_2}{2}\right)<\epsilon/2,\label{eqn:mainlem_part1}\\
  \P\left(\left|\|\tilde{\bsu}^{(k)}\|_2 \|\bsp\|_2 r_{j}^{(k)}\cos\theta -\|\bsp\|_2 r_{j}^{(k)}\right|>\frac{\delta_1\|\bsp\|_2}{2}\right)<\epsilon/2.\label{eqn:mainlem_part2}
\end{align}
Then, the triangle inequality implies \eqref{eqn:entry_bd1}. 

We first prove \eqref{eqn:mainlem_part1}.
Remind that we do the random splitting of the input matrix $\bsA$ and define the two independent binary-converted matrices as $\bsX^{(k)}$ and $\bsY^{(k)}$, for $1\leq k<K$, which are used to estimate $\tilde{\bsu}^{(k)}$ and $\bsv^{(k)}$, respectively. Thus, $\tilde{\bsu}^{(k)}$ is independent from $\bsY^{(k)}$ and this independence is used when we bound the first and second moments of $v_j^{(k)}=\frac{1}{s}\langle\bsY^{(k)}_{*j}, \tilde{\bsu}^{(k)}\rangle$.
For any $1\leq j\leq m$, the first and second moments of $v_j^{(k)}=\frac{1}{s}\langle\bsY^{(k)}_{*j}, \tilde{\bsu}^{(k)}\rangle$ satisfy
\beq
\E\left[\frac{1}{s}\left\langle\bsY^{(k)}_{*j}, \tilde{\bsu}^{(k)}\right\rangle\right]=\langle \bsp,\tilde{\bsu}^{(k)}\rangle r_j^{(k)}=\|\bsp\|_2 \| \tilde{\bsu}^{(k)}\|_2(\cos\theta)r_j^{(k)}=\Theta(\sqrt{n})
\eeq
if  $r_j^{(k)}\neq 0$ by Lemma~\ref{lem:sin_uk_p} and~\ref{lem:sin_uktilde_p}, and
\beq
\var\left(\frac{1}{s}\left\langle\bsY^{(k)}_{*j}, \tilde{\bsu}^{(k)}\right\rangle\right)\leq\frac{1}{s^2}\sum_{i=1}^n (\tilde{u}_i^{(k)})^2\E[(Y_{ij}^{(k)})^2]=\Theta\left(\frac{1}{s}\right)
\eeq
since $\E[(Y_{ij}^{(k)})^2]=\Theta(s)$ and $\sum_{i=1}^n (\tilde{u}_i^{(k)})^2=\Theta(1)$ by Lemma~\ref{lem:sin_uk_p} and~\ref{lem:sin_uktilde_p}.
Furthermore, we have $\max_{1\leq i\leq m}| Y^{(k)}_{ij}\tilde{u}_i^{(k)} |\leq \Theta\left(\frac{1}{\sqrt{n}}\right) $ since  $\tilde{u}_i^{(k)}\leq \frac{2}{\eta\sqrt{n}}$. By applying the Bernstein's inequality, we can show that
\beq
\begin{split}
\P\left(\left|\frac{1}{s}\left\langle \bsY^{(k)}_{*j}, \tilde{\bsu}^{(k)}\right\rangle-\|\tilde{\bsu}^{(k)}\|_2\|\bsp\|_2 r_{j}^{(k)}\cos\theta \right|>\frac{\delta_1\|\bsp\|_2}{2}\right)&\leq 
2\exp\left(-\frac{\Theta(\delta_1^2\|\bsp\|_2^2)}{\Theta\left(\frac{1}{s}\right)+ \Theta\left(\delta_1\|\bsp\|_2/\sqrt{n}\right)}\right)\\
&\leq \exp\left(-\Theta(s\delta_1^2\|\bsp\|_2^2)\right)
\end{split}
\eeq
where the second inequality is due to the assumption $\|\bsp\|_2=\Theta(\sqrt{n})$. 
To make this probability less than $\frac{\epsilon}{2}$, it is sufficient to have
$
s\geq \Omega\left(\frac{1}{\delta_1^2\|\bsp\|_2^2}\log\frac{1}{\epsilon}\right).
$

We next prove \eqref{eqn:mainlem_part2} by bounding $\left|\|\tilde{\bsu}^{(k)}\|_2 \|\bsp\|_2 r_{j}^{(k)}\cos\theta -\|\bsp\|_2 r_{j}^{(k)}\right|$. By the triangle inequality, we have
\beq
\begin{split}
\left|\|\tilde{\bsu}^{(k)}\|_2 \|\bsp\|_2 r_{j}^{(k)}\cos\theta -\|\bsp\|_2 r_{j}^{(k)}\right|&\leq \left|\|\tilde{\bsu}^{(k)}\|_2 \|\bsp\|_2 r_{j}^{(k)}\cos\theta -\|\bsp\|_2 r_{j}^{(k)}\cos\theta\right|\\
&\quad+\left|\|\bsp\|_2 r_{j}^{(k)}\cos\theta -\|\bsp\|_2 r_{j}^{(k)}\right|.
\end{split}
\eeq
Note that 
\beq
\begin{split}
\frac{1}{\|\bsp\|_2 }\cdot\left|\|\tilde{\bsu}^{(k)}\|_2 \|\bsp\|_2 r_{j}^{(k)}\cos\theta -\|\bsp\|_2 r_{j}^{(k)}\cos\theta\right|&= r_j^{(k)}\cos\theta \left|\|\tilde{\bsu}^{(k)}\|_2 -1\right|\\
&\leq \Theta(\sin^2\theta(\bsu^{(k)},\bsp^*))=\frac{1}{\Theta\left({ns}\right)},
\end{split}
\eeq
with probability $1-(n+m)^{-8}$ by Lemma~\ref{lem:sin_uk_p} and~\ref{lem:sin_uktilde_p}, and also note that
\beq
\begin{split}
\frac{1}{\|\bsp\|_2 }\cdot\left|\|\bsp\|_2 r_{j}^{(k)}\cos\theta -\|\bsp\|_2 r_{j}^{(k)}\right|&=r_j^{(k)}(1-\cos\theta)\\
&\leq \Theta(\sin^2\theta(\bsu^{(k)},\bsp^*))=\frac{1}{\Theta\left(ns\right)},
\end{split}
\eeq 
with probability $1-(n+m)^{-8}$ by Lemma \ref{lem:sin_uk_p} and~\ref{lem:sin_uktilde_p}.
To make these errors of order $1/{\Theta\left(ns\right)}$ less than $\frac{\delta_1}{2}$, it is sufficient to have
$
s\geq\Omega\left( \frac{1}{\delta_1n}\right).
$

By combining the above results, it can be guaranteed that
$
\left| \frac{1}{2s}\left\langle \bsY^{(k)}_{*j}, \tilde{\bsu}^{(k)}\right\rangle -\|\bsp\|_2 r_{j}^{(k)} \right|<\delta\|\bsp\|_2
$ with probability at least $1-\epsilon$, 
if the sampling probability
\beq\label{eqn:main_lem_cond_s}
s\geq \max\left\{\Omega\left( \frac{1}{\delta_1^2\|\bsp\|_2^2}\log\frac{1}{\epsilon}\right),\Omega\left(\frac{1}{\delta_1 n}\right) \right\}=\Omega\left( \frac{1}{\delta_1^2\|\bsp\|_2^2}\log\frac{1}{\epsilon}\right)
\eeq
where the last equality is due to $\|\bsp\|_2=\Theta(\sqrt{n})$. The condition $s=\Omega(\log(n+m)/m)$ in Lemma \ref{lem:sin_uk_p} is immediately satisfied by \eqref{eqn:main_lem_cond_s} when $n=O(m/\log m)$.
\end{proof}

\paragraph{Proof of Theorem \ref{thm:main}.}
By using Lemma \ref{lem:entrywisebd_new}, we next prove Theorem \ref{thm:main}.
By applying the union bound over $k\in[K]$, if $s\geq \Theta\left(\frac{1}{\delta_1^2\|\bsp\|_2^2}\log\frac{K}{\epsilon}\right)$ then we have 
\beq\label{eqn:boundvj}
\|\bsp\|_2(r_j^{(k)}-\delta_1) \leq v_j^{(k)}=\frac{1}{s}\left\langle \bsY^{(k)}_{*j}, \tilde{\bsu}^{(k)}\right\rangle \leq \|\bsp\|_2(r_j^{(k)}+\delta_1), \;\;\forall k\in[K]
\eeq
for any $\delta_1>0$ and $j\in[m]$ with probability at least $1-\epsilon$.
Under the condition~\eqref{eqn:boundvj}, for any $q_j\in(1/2,1)$ and $\delta< \min\left\{\frac{2q_j-1}{2},\frac{1-q_j}{2}\right\}$, we can guarantee that
\beq
\frac{1}{K}-q_j+\delta<\frac{1}{K}-(1-q_j)-\delta\;\; \text{ and }\;\; \frac{1}{K}-(1-q_j)+\delta<\frac{1}{K}-\delta,
\eeq
which implies $(\hat{g}_j,\hat{h}_j)=(g_j,h_j)$ for $(\hat{g}_j,\hat{h}_j)$ defined in~\eqref{alg:gh}. This proves \eqref{eqn:rec_gh_main} of Theorem \ref{thm:main}.

We next prove \eqref{eqn:rec_q_main}, the accuracy guarantee in estimating the task difficulty vector $\bsq$.
After estimating $\|\bsp\|_2\bsr^{(k)}$ by $\bsv^{(k)}=\frac{1}{s}(\bsY^{(k)})^{\top}\tilde{\bsu}^{(k)}$, we estimate $\|\bsp\|_2$ by calculating $l$ where ${l}_j:=\frac{K}{K-2}\sum_{k\neq \hat{g}_j, k\neq \hat{h}_j} \Delta v_j^{(k)}$ and ${l}:=\frac{1}{m}\sum_{j=1}^m l_j$.
Assume that
$
|\|\bsp\|_2-l|\leq \|\bsp\|_2\delta'.
$
We will specify the required order of $\delta'$ later. 
Remind that the estimate for $q_j$ is defined as 
$
\hat{q}_j:= \frac{1}{K}-\frac{\Delta v_{j}^{(\hat{g}_j)}}{l}.
$
Under the condition that $\hat{g}_j=g_j$ and $|v_j-\|\bsp\|_2r_j^{(k)}|\leq \|\bsp\|_2\delta_1$, both of which are satisfied under the conditions of Lemma \ref{lem:entrywisebd_new},
we have
\beq
\frac{\left(\frac{1}{K}-q_j-2\delta_1\right)}{1+\delta'}\leq \frac{\Delta v_j^{(\hat{g}_j)}}{l} \leq \frac{\left(\frac{1}{K}-q_j+2\delta_1\right)}{1-\delta'}.
\eeq 
By the Taylor expansion for $\frac{1}{1-x}=1+x+\Theta(x^2)$ as $x\to 0$, we have
\beq
|\hat{q}_j-q_j|\leq 2\delta_1 +\delta'\left(\frac{1}{K}-q_j+2\delta_1\right)+\Theta(\delta'^2)=\Theta(\delta_1+\delta').
\eeq
Thus, both the order of $\delta'$, which is the estimation error of $\|\bsp\|_2$, and that of $\delta$, which is the estimation error of $\|\bsp\|_2 r_j^{(k)}$, govern the estimation accuracy of $q_j$.
We next show that we can have $\delta'=\Theta(\delta_1)$. By Lemma \ref{lem:entrywisebd_new}, we have $|v_j-\|\bsp\|_2r_j^{(k)}|\leq \|\bsp\|_2\delta_1$, which implies
\beq
\|\bsp\|_2( \Delta r_j^{(k)}-2\delta_1)\leq \Delta v_j^{(k)}\leq \|\bsp\|_2( \Delta r_j^{(k)}+2\delta_1).
\eeq
Under the condition $(\hat{g}_j,\hat{h}_j)=(g_j,h_j)$, since $\Delta r_j^{(k)}=\frac{1}{K}$ for $k\neq \hat{g}_j,\hat{h}_j$, we have
\beq
 \|\bsp\|_2-\|\bsp\|_2\frac{2\delta_1 K}{K-2} \leq l_j=\frac{K}{K-2}\sum_{k\neq \hat{g}_j, k\neq \hat{h}_j} \Delta v_j^{(k)}\leq \|\bsp\|_2+\|\bsp\|_2\frac{2\delta_1 K}{K-2},
\eeq
and thus $\delta'=\frac{2\delta_1 K}{K-2}=\Theta(\delta_1)$. Thus, it is enough to have $s= \Omega\left(\frac{1}{\delta_1^2\|\bsp\|_2^2}\log\frac{K}{\epsilon}\right)$ to guarantee \eqref{eqn:rec_q_main}.

\paragraph{Proof of Corollary \ref{cor:exact_alg1}.}

By using Lemma \ref{lem:entrywisebd_new} and taking the union bound over all tasks $j\in[m]$ as well as $k\in[K]$, we can prove  Corollary \ref{cor:exact_alg1} in a similar way as that of Theorem \ref{thm:main}.


\subsection{Proof of Lemma \ref{lem:sin_uktilde_p}}\label{app:sec:lem:sin_uktilde}
We first prove \eqref{lem:utilde_norm},
\begin{equation*}
\|\tilde{\bsu}^{(k)}\|_2\geq\sqrt{1-50\sin^2\theta(\bsu^{(k)},\bsp^*)}.
\end{equation*}

Let $I$ be the set of indices $1\leq i\leq n$ such that $u_i^{(k)}\geq \frac{2}{\eta\sqrt{n}}$. Then, we have $u_i^{(k)}-p^*_i\geq \frac{1}{\eta\sqrt{n}}$ for all $i\in I$ since $p_i^*={p_i}/{\|\bsp\|_2}\leq \frac{1}{\eta\sqrt{n}}$ due to the assumption that $\|\bsp\|_2^2\geq \eta^2n$. Thus, we have
\beq
\frac{|I|}{\eta^2n} \leq \sum_{i\in I} (u_i^{(k)}-p^*_i)^2\leq \|\bsu^{(k)}-\bsp^*\|_2^2.
\eeq
By using the triangle inequality, we can show that
\beq
\begin{split}\label{eqn:bd_u_diff}
\sqrt{\sum_{i\in I} \left(u_i^{(k)}\right)^2}&\leq \sqrt{\sum_{i\in I} \left(u_i^{(k)}-\frac{2}{\eta\sqrt{n}}\right)^2}+\sqrt{\frac{4|I|}{\eta^2n}}\\
&\leq \sqrt{\sum_{i\in I} \left(p_i^{*}-\frac{2}{\eta\sqrt{n}}\right)^2}+\sqrt{\sum_{i\in I} \left(u_i^{(k)}-p_i^{*}\right)^2}+\sqrt{\frac{4|I|}{\eta^2n}}\\
&\leq \sqrt{\frac{4|I|}{\eta^2n}}+\sqrt{\sum_{i\in I} \left(u_i^{(k)}-p_i^{*}\right)^2}+\sqrt{\frac{4|I|}{\eta^2n}}\\
&\leq 5 \|\bsu^{(k)}-\bsp^*\|_2.
\end{split}
\eeq
Therefore, we get
\beq\label{eqn:lower_bd_norm_utilde}
1\geq \|\tilde{\bsu}^{(k)}\|_2^2=1-\sum_{i\in I} (u_i^{(k)})^2\geq1-25\|\bsu^{(k)}-\bsp^*\|_2^2.
\eeq

By the law of cosine, we have
\beq
\begin{split}\label{eqn:law_of_cos}
\|\bsp^*-\bsu^{(k)}\|_2^2&=\sin^2 \theta(\bsu^{(k)},\bsp^*)+(1-\cos\theta(\bsu^{(k)},\bsp^*))^2=2-2\cos\theta(\bsu^{(k)},\bsp^*)\\
&=2\left(1-\sqrt{1- \sin^2\theta(\bsu^{(k)},\bsp^*)}\right)=2\frac{\sin^2\theta(\bsu^{(k)},\bsp^*)}{1+\sqrt{1-\sin^2\theta(\bsu^{(k)},\bsp^*)}}\\
&\leq 2\sin^2\theta(\bsu^{(k)},\bsp^*).
\end{split}
\eeq
Combining \eqref{eqn:lower_bd_norm_utilde} and \eqref {eqn:law_of_cos} proves \eqref{lem:utilde_norm}.

We next prove \eqref{lem:utilde_sin},
\begin{equation*}
\sin\theta({\tilde{\bsu}}^{(k)},\bsp^*)\leq 6\sqrt{2} \sin\theta({\bsu}^{(k)},\bsp^*).
\end{equation*}
First, note that $\|\tilde{\bsu}^{(k)}-\bsu^{(k)}\|_2^2=\sum_{i\in I} \left(u_i^{(k)}\right)^2$. We have
\beq
\sin\theta(\tilde{\bsu}^{(k)},\bsp^*)\leq \|\tilde{\bsu}^{(k)}-\bsp\|_2\leq  \|\tilde{\bsu}^{(k)}-\bsu^{(k)}\|_2+ \|\bsu^{(k)}-\bsp^*\|_2\leq 6\|\bsu^{(k)}-\bsp^*\|_2
\eeq
where the last inequality is from \eqref{eqn:bd_u_diff}. Combined with \eqref{eqn:law_of_cos}, we get \eqref{lem:utilde_sin}.

\section{Performance Analysis of Algorithm \ref{alg:plugin_MLE}}\label{app:thm:plugin_MLE}

\subsection{Proof of Lemma \ref{lem:rec_p}}
In this lemma, we show that conditioned on $(\hat{g}_j,\hat{h}_j)=(g_j,h_j)$ for all $j\in[m]$, if $s(1-s_1)=\Omega\left(\frac{1}{\delta_2 m}\log\frac{n}{\epsilon}\right)$, the estimator $\hat{p}_i$ defined in \eqref{eqn:hatpi},
\begin{equation*}
\hat{p}_i=\frac{K}{(K-2)}\left(\frac{1}{s(1-s_1)}\left(\frac{1}{m}\sum_{j=1}^m \mathbbm{1}(A^2_{ij}=\hat{g}_j\text{ or } \hat{h}_j)\right)-\frac{2}{K}\right),
\end{equation*}
guarantees 
$
\P\left(\|\bsp-\hat{\bsp}\|_\infty<\delta_2\right)\geq1-\epsilon
$
for any $\epsilon>0$.

Given $(\hat{g}_j,\hat{h}_j)=(g_j,h_j)$ for all $j\in[m]$, since $\bsA^2$ is independent of $(\hat{g}_j,\hat{h}_j)$, we have
\beq
\begin{split}
\E\left[ \mathbbm{1}(A^2_{ij}=\hat{g}_j\text{ or } \hat{h}_j)\right]&=\P(A^2_{ij}=\hat{g}_j \text{ or }\hat{h}_j)=s(1-s_1)\left(\frac{K-2}{K}p_i+\frac{2}{K}\right),\\
\var\left( \mathbbm{1}(A^2_{ij}=\hat{g}_j\text{ or } \hat{h}_j)\right)&\leq s(1-s_1).
\end{split}
\eeq
 By applying the Bernstein's inequality, we can show that 
 \beq
 \begin{split}
& \P\left(\left|\sum_{j=1}^m \left(\mathbbm{1}(A^2_{ij}=\hat{g}_j\text{ or } \hat{h}_j)-s(1-s_1)\left(\frac{K-2}{K}p_i+\frac{2}{K}\right)\right)\right|>\frac{(K-2)ms(1-s_1)\delta_2}{K}\right)\\
&\leq \exp\left(-\frac{\frac{1}{2}\left(\frac{(K-2)ms(1-s_1)\delta_2}{K}\right)^2}{ms(1-s_1)+\frac{1}{3}\frac{(K-2)ms(1-s_1)\delta_2}{K}}\right)\leq\exp\left(-\Theta\left(ms(1-s_1)\delta_2^2\right)\right).
 \end{split}
 \eeq
 
 Thus, if the sampling probability satisfies
 \beq
 s(1-s_1)=\Omega\left(\frac{1}{m\delta_2^2}\log\frac{1}{\epsilon}\right),
 \eeq
then we can guarantee that
$
\P(|\hat{p}_i-p_i|<\delta_2)\geq 1-\epsilon.
$
By taking the union bound over $i\in[n]$, if the sampling probability satisfies
 \beq
 s(1-s_1)=\Omega\left(\frac{1}{m\delta_2^2}\log\frac{n}{\epsilon}\right),
 \eeq
 then we can guarantee that $\P\left(\|\hat\bsp-{\bsp}\|_\infty<\delta_2\right)\geq1-\epsilon$.

\subsection{Proof of Theorem \ref{thm:plugmin_mle}}
To prove this theorem, we use similar proof techniques from \cite{SEM}. Since the work in \cite{SEM} focuses on the recovery of only the ground-truth label for each task, we generalize the techniques to recover not only the ground-truth label but also the most confusing answer.

We first introduce some notations. Let $\mu^{(i,j)}_{(a,b),k}$ denote the probability that a worker $i\in[n]$ gives label $k\in[K]$ for the assigned task $j\in[m]$ of which the top-two answers are $(g_j,h_j)=(a,b)$.
Let $\boldsymbol{\mu}^{(i,j)}_{(a,b)}=[\mu^{(i,j)}_{(a,b),1}\;\; \mu^{(i,j)}_{(a,b),2}\;\;\cdots\;\;\mu^{(i,j)}_{(a,b),K}]^\top$.
We introduce a quantity that measures the average ability of workers in distinguishing the ground-truth pair of top-two answers $(g_j,h_j)$ from any other pair $(a,b)\in[K]^2/\{(g_j,h_j)\}$ for the task $j\in[m]$.
We define
\beq
\begin{split}\label{defn:overlineD}
\overline{D}^{(j)}:=\min_{(g_j,h_j)\neq (a,b)}\frac{1}{n}\sum_{i=1}^n  \mathbb{D}_\mathsf{KL}\left(\boldsymbol{\mu}^{(i,j)}_{(g_j,h_j)},\boldsymbol{\mu}^{(i,j)}_{(a,b)}\right);\quad
\overline{D}:=\min_{j\in[m]}\overline{D}^{(j)},
\end{split}
\eeq
where $ \mathbb{D}_\mathsf{KL}(P,Q):=\sum_i P(i)\log(P(i)/Q(i))$ is the KL-divergence between $P$ and $Q$. 
Note that $\overline{D}^{(j)}$ is strictly positive if $q_j\in(1/2,1)$ and there exists at least one worker $i$ with $p_i>0$ for the distribution \eqref{eqn:A_dist}, so that $(g_j,h_j)$ can be distinguished from any other $(a,b)\in[K]^2/\{(g_j,h_j)\}$ statistically. We define $\overline{D}$ as the minimum of $\overline{D}^{(j)}$ over $j\in[m]$, indicating the average ability of workers in distinguishing $(g_j,h_j)$ from any other $(a,b)$ for the most difficult task in the set.

Let us define an event that will be shown holding with high probability,
\beq\label{eqn:good_event}
\caE: \sum_{i=1}^n \sum_{k=1}^K \mathbbm{1}(A_{ij}=k)\log\left(\frac{\mu^{(i,j)}_{(g_j,h_j),k}}{\mu^{(i,j)}_{(a,b),k}}\right)\geq ns\overline{D}/2 \text{ for all }j\in[m]\text{ and }(a,b)\in[K]\times[K]\backslash {(g_j,h_j)}.
\eeq

Define
\beq
l_i:=\sum_{k=1}^K \mathbbm{1}(A_{ij}=k)\log\left(\mu^{(i,j)}_{(g_j,h_j),k}/\mu^{(i,j)}_{(a,b),k}\right).
\eeq
We can see that $l_1,\dots, l_n$ are mutually independent on any value of $(g_j,h_j)$, and each $l_i$ belongs to the interval $[0,\log(1/\rho)]$ where $\mu^{(i,j)}_{(g_j,h_j),c}\geq \rho$ for all $(i,j,g_j,h_j,c)\in[n]\times[m]\times [K]^3$.
We can easily show that
\beq
\E\left[\sum_{i=1}^n l_i\Bigg|(g_j,h_j)\right]=\sum_{i=1}^n s \mathbb{D}_{\mathsf KL}\left(\boldsymbol{\mu}^{(i,j)}_{(g_j,h_j)},\boldsymbol{\mu}^{(i,j)}_{(a,b)}\right).
\eeq
We define
\beq
\begin{split}
D&:=\sum_{i=1}^n  \mathbb{D}_{\mathsf KL}\left(\boldsymbol{\mu}^{(i,j)}_{(g_j,h_j)},\boldsymbol{\mu}^{(i,j)}_{(a,b)}\right).
\end{split}
\eeq

The following lemma shows that the second moment of $l_i$ is bounded above by the KL-divergence between the label distribution under  $(g_j,h_j)$ pair and the label distribution under $(a,b)$ pair.
\begin{lem}\label{lem:second_si}
Conditioning on any value of $(g_j,h_j)$, we have
\beq
\E\left[ l_i^2|(g_j,h_j)\right]\leq \frac{2\log(1/\rho)}{1-\rho}s  \mathbb{D}_{\mathsf KL}\left(\boldsymbol{\mu}^{(i,j)}_{(g_j,h_j)},\boldsymbol{\mu}^{(i,j)}_{(a,b)}\right).
\eeq
\end{lem}
The proof of this lemma can be obtained by following the proof of the similar result, Lemma 4 of \cite{SEM}.

According to Lemma~\ref{lem:second_si}, the aggregated second moment of $l_i$ is bounded by
\beq
\begin{split}
\E\left[\sum_{i=1}^n l_i^2\Bigg|(g_j,h_j)\right]&\leq \frac{2\log(1/\rho)}{1-\rho}\sum_{i=1}^{n}  s\mathbb{D}_{\mathsf KL}\left(\boldsymbol{\mu}^{(i,j)}_{(g_j,h_j)},\boldsymbol{\mu}^{(i,j)}_{(a,b)}\right)\\
&=\frac{2\log(1/\rho)}{1-\rho}sD.
\end{split}
\eeq
Thus, applying the Bernstein's inequality, we have
\beq
\P\left[\sum_{i=1}^n l_i\geq sD/2\Bigg|(g_j,h_j)\right]\geq 1-\exp\left(-\frac{\frac{1}{2}(sD/2)^2}{\frac{2\log(1/\rho)}{1-\rho}sD+\frac{1}{3}(2\log(1/\rho))(sD/2)}\right).
\eeq
Since $\rho\leq 1/2$ and $D\geq n\overline{D}^{(j)}\geq n\overline{D}$, combining the above inequality with union bound over $j\in[m]$, we have
\beq
\begin{split}\label{eqn:prob_E}
\P\left[\caE\right] &\geq 1-mK^2 \exp\left(-\frac{ns\overline{D}}{33\log(1/\rho)}\right).
\end{split}
\eeq

The maximum likelihood estimator finds a pair of $(a,b)\in[K]^2$, $a\neq b$, maximizing
\begin{align}
        (\hat{g}_j, \hat{h}_j) &= \argmax_{(a,b)\in[K]^2, a \neq b} \prod_{i=1}^n \mathbb{P}(A_{ij}\vert \bsp, q_j, (a,b)) \nonumber\\
        &= \argmax_{(a,b)\in[K]^2, a \neq b} \sum_{i=1}^n \log \mathbb{P}(A_{ij}\vert \bsp, q_j, (a,b)) \nonumber\\
        &= \argmax_{(a,b)\in[K]^2, a \neq b} \sum_{i=1}^n\sum_{k=1}^K \mathbbm{1}(A_{ij}=k)\log \mu^{(i,j)}_{(a,b),k}.
\end{align}
The plug-in MLE in \eqref{eqn:plug-in-mle}, on the other hand, finds a pair of $(a,b)\in[K]^2$, $a\neq b$, maximizing
\begin{align}\label{eqn:rule_pluginmle}
        (\hat{g}_j, \hat{h}_j)        &= \argmax_{(a,b)\in[K]^2, a \neq b} \sum_{i=1}^n\sum_{k=1}^K \mathbbm{1}(A_{ij}=k)\log \hat{\mu}^{(i,j)}_{(a,b),k}
\end{align}
where $\hat{\mu}^{(i,j)}_{(a,b),k}$ is the estimated probability that a worker $i\in[n]$ gives label $k\in[K]$ for the assigned task $j\in[m]$ of which the top two answers are $(g_j,h_j)=(a,b)$ assuming $p_i=\hat{p}_i$ from \eqref{eqn:hatpi} and $q_j=\hat{q}_j$ from \eqref{eqn:est_q} in the distribution \eqref{eqn:A_dist}.
Thus, for the plug-in MLE to correctly find the ground-truth top two answers $(g_j,h_j)$, we need to satisfy the following event:
\beq
\sum_{i=1}^n \sum_{k=1}^K \mathbbm{1}(A_{ij}=k)\log\left(\hat{\mu}^{(i,j)}_{(g_j,h_j),k}/\hat{\mu}^{(i,j)}_{(a,b),k}\right)\geq 0 \text{ for all }(a,b)\in[K]\times[K]\backslash {(g_j,h_j)}.
\eeq
For any arbitrary $(a,b)\neq (g_j,h_j)$, consider the quantity
\beq
Q_{(a,b)}:=\sum_{i=1}^n \sum_{k=1}^K \mathbbm{1}(A_{ij}=k)\log\left(\hat{\mu}^{(i,j)}_{(g_j,h_j),k}/\hat{\mu}^{(i,j)}_{(a,b),k}\right),
\eeq
which can be written as
\beq
\begin{split}
Q_{(a,b)}&=\sum_{i=1}^n \sum_{k=1}^K \mathbbm{1}(A_{ij}=k)\log\frac{{\mu}^{(i,j)}_{(g_j,h_j),k}}{{\mu}^{(i,j)}_{(a,b),k}} +\sum_{i=1}^n \sum_{k=1}^K\mathbbm{1}(A_{ij}=k)\left[\log\left(\frac{\hat{\mu}^{(i,j)}_{(g_j,h_j),k}}{{\mu}^{(i,j)}_{(g_j,h_j),k}}\right)-\log\left(\frac{\hat{\mu}^{(i,j)}_{(a,b),k}}{{\mu}^{(i,j)}_{(a,b),k}}\right)\right].
\end{split}
\eeq

Assuming that there exist $\rho>\delta_3$ such that
\beq
{\mu}^{(i,j)}_{(a,b),k}\geq \rho \text{ and } |\hat{\mu}^{(i,j)}_{(a,b),k}-{\mu}^{(i,j)}_{(a,b),k}|\leq \delta_3\text{ for all }i\in[n], j\in[m], (a,b)\in[K]^2,
\eeq
we have
\beq
\max_{i\in[n],k\in[K]} \left[\log\left(\frac{\hat{\mu}^{(i,j)}_{(g_j,h_j),k}}{{\mu}^{(i,j)}_{(g_j,h_j),k}}\right)-\log\left(\frac{\hat{\mu}^{(i,j)}_{(a,b),k}}{{\mu}^{(i,j)}_{(a,b),k}}\right)\right]\leq 2\log\left(\frac{\rho}{\rho-\delta_3}\right).
\eeq
By the Bernstein's inequality, we also have
\beq\label{eqn:Aij_cont}
\P\left[\left|\sum_{i=1}^n\sum_{k=1}^K \mathbbm{1}(A_{ij}=k)-ns\right|> ns/2\right]\leq \exp\left(-\frac{\frac{1}{2}(ns/2)^2}{ns+\frac{1}{3}(ns/2)}\right)=\exp\left(-\frac{3ns}{28}\right).
\eeq
By taking the union bound over $j\in[m]$, we have
\beq\label{eqn:Aij_cont_union}
\P\left[\left|\sum_{i=1}^n\sum_{k=1}^K \mathbbm{1}(A_{ij}=k)-ns\right|>ns/2\text{ for any }j\in[m]\right]\leq m\exp\left(-\frac{3ns}{28}\right).
\eeq
Under the intersection of the event $\left|\sum_{i=1}^n\sum_{k=1}^K \mathbbm{1}(A_{ij}=k)-ns\right|\leq ns/2$ for all $j\in[m]$  and the event $\mathcal{E}$, we can guarantee
\beq
\begin{split}
Q_{(a,b)}&=\sum_{i=1}^n \sum_{k=1}^K \mathbbm{1}(A_{ij}=k)\log\frac{{\mu}^{(i,j)}_{(g_j,h_j),k}}{{\mu}^{(i,j)}_{(a,b),k}} +\sum_{i=1}^n \sum_{k=1}^K\mathbbm{1}(A_{ij}=k)\left[\log\left(\frac{\hat{\mu}^{(i,j)}_{(g_j,h_j),k}}{{\mu}^{(i,j)}_{(g_j,h_j),k}}\right)-\log\left(\frac{\hat{\mu}^{(i,j)}_{(a,b),k}}{{\mu}^{(i,j)}_{(a,b),k}}\right)\right]\\
&\geq \frac{ns\overline{D}}{2}-3ns\log\left(\frac{\rho}{\rho-\delta_3}\right)\geq ns \left(\frac{\overline{D}}{2}-\frac{3\delta_3}{\rho-\delta_3}\right)>0
\end{split}
\eeq
for every $j\in[m]$ 
where the last inequality holds if
\beq
\delta_3< \rho\frac{\overline{D}}{6+\overline{D}}.
\eeq
In summary, under that the event $\left|\sum_{i=1}^n\sum_{k=1}^K \mathbbm{1}(A_{ij}=k)-ns\right|\leq ns/2$ for all $j\in[m]$ and the event $\mathcal{E}$ hold, if we have $\delta_3$ such that
\beq
|\hat{\mu}^{(i,j)}_{(a,b),k}-{\mu}^{(i,j)}_{(a,b),k}|\leq \delta_3\text{ for all }i\in[n], j\in[m], (a,b)\in[K]^2
\eeq
and
\beq
\delta_3<\rho\text{ and }\quad\delta_3< \rho\frac{\overline{D}}{6+\overline{D}},
\eeq
then we can guarantee that the plug-in MLE in \eqref{eqn:rule_pluginmle} successfully recovers the pair of top two $(g_j,h_j)$ for all the tasks $j\in[m]$.
To make the right-hand side of \eqref{eqn:prob_E} and \eqref{eqn:Aij_cont_union} less than $\epsilon/2$, it is sufficient to have
\beq
s=\Omega\left(\frac{\log(1/\rho)\log(mK^2/\epsilon)+\log(m/\epsilon)}{n\overline{D}}\right).
\eeq

Lastly, when we have 
\beq\label{eqn:cond_p_q_inf_delta}
\max\{\|\bsp-\hat{\bsp}\|_\infty, \|\bsq-\hat{\bsq}\|_\infty\}\leq \delta,
\eeq
we can guarantee that 
\beq
|\hat{\mu}^{(i,j)}_{(a,b),k}-{\mu}^{(i,j)}_{(a,b),k}|\leq 4\delta:=\delta_3.
\eeq
Thus, it is sufficient to guarantee \eqref{eqn:cond_p_q_inf_delta} with
\beq
\delta<\min\left\{\frac{\rho}{4},\frac{\rho\overline{D}}{4(6+\overline{D})}\right\}.
\eeq


\section{Proof of Theorem \ref{thm:converse}}\label{app:thm:converse}

\subsection{Proof of part (a)}

To prove this minimax bound, we use the similar arguments from \cite{Karger}.
In particular, we consider a spammer-hammer model such that
\beq
p_i=\begin{cases}
0, \text{ for } 1\leq i\leq \floor{(1-\overline{p})n}\\
1,\text{ otherwise. }
\end{cases}
\eeq
Assume that total $l_j$ workers randomly sampled from $[n]$ provide answers for the task $j$. 
Under the spammer-hammer model, the oracle estimator makes a mistake on task $j$ with probability $(K-1)/K$ if it is only assigned to spammers. When $l_j$ is the number of assignments, we have
\beq
\p(\hat{g}_j\neq g_j)=\frac{K-1}{K}(1-\overline{p})^{l_j}.
\eeq
By convexity and using Jensen's inequality, the average probability of error is lower bounded by
\beq
\frac{1}{m}\sum_{j\in[m]}\p(\hat{g}_j\neq g_j)\geq\frac{K-1}{K}(1-\overline{p})^{l}
\eeq
where $\frac{1}{m}\sum_{i\in[m]}l_i\leq l$. By assuming $\overline{p}\leq 2/3$, we have $(1-\overline{p})\geq e^{-(\overline{p}+\overline{p}^2)}$. Thus,
\beq\label{eqn:minmax1_bd_f}
\min_{\hat{\bsg}}\max_{\bsp\in\mathcal{F}_{\overline{p}},\;\; \bsg\in[K]^m}\frac{1}{m}\sum_{j\in[m]}\p(\hat{g}_j\neq g_j)\geq \frac{K-1}{K}e^{-(\overline{p}+\overline{p}^2)l}\geq \frac{K-1}{K}e^{-2\overline{p}l}.
\eeq
The inequality in \eqref{eqn:minmax1_bd_f} implies that if $l$ is less than $\frac{1}{2\overline{p}}\log\left(\frac{K-1}{K\epsilon}\right)$, then no algorithm can make the minimax error in \eqref{eqn:minmax1_bd_f} less than $\epsilon$.
Since the average number of queries per task in our model is $ns$, it implies that it is necessary to have $s=\Omega\left(\frac{1}{\|\bsp\|_2^2}\log\frac{1}{\epsilon}\right)$.

\subsection{Proof of part (b)}
To prove the second part of the theorem, we use proof techniques from \cite{SEM}, but generalizes the results for pair of top two answers.
We assume that $j_c\in[m]$, $(g_c,h_c)\in[K]^2$ and $(a_c,b_c)\in[K]^2$ are the task index and the pairs of labels such that
\beq
\overline{D}=\frac{1}{n}\sum_{i=1}^n  \mathbb{D}_\mathsf{KL}\left(\boldsymbol{\mu}^{(i,j_c)}_{(g_c,h_c)},\boldsymbol{\mu}^{(i,j_c)}_{(a_c,b_c)}\right)
\eeq
for $\overline{D}$ defined in \eqref{defn:overlineD}.

Let $\mathbb{Q}$ be a uniform distribution over the set $\{(g_c,h_c),(a_c,b_c)\}^m$. For any $(\hat{g},\hat{h})$, we have
\beq
\begin{split}\label{eqn:conv2_1}
&\max_{\substack{(\bsv,\bsu)\in[K]^m\times [K]^m\\ v_j\neq u_j,\forall j[m]}}\E\left[\sum_{j=1}^m\mathbbm{1}((\hat{g}_j,\hat{h}_j)\neq (g_j,h_j))\Big|(\bsg,\bsh)=(\bsv,\bsu)\right]\\
&\geq \sum_{j=1}^m \sum_{(\bsv,\bsu)\in \{(g_c,h_c),(a_c,b_c)\}^m}\mathbb{Q}((\bsv,\bsu))\E\left[\mathbbm{1}((\hat{g}_j,\hat{h}_j)\neq (g_j,h_j))\Big|(\bsg,\bsh)=(\bsv,\bsu)\right]
\end{split}
\eeq
Let $\bsA:=\{A_{ij}:i\in[n],j\in[m]\}$ be the set of observations. Define two probability measures $\mathbb{P}_0$ and $\mathbb{P}_1$, such that $\mathbb{P}_0$ is the measure of $\bsA$ conditioned on $(g_{j},h_{j})=(g_c,h_c)$, while $\mathbb{P}_1$ is that on $(g_{j},h_{j})=(a_c,b_c)$.
Then, we can have
\beq
\begin{split}\label{eqn:conv2_2}
&\sum_{(\bsv,\bsu)\in \{(g_c,h_c),(a_c,b_c)\}^m}\mathbb{Q}((\bsv,\bsu))\E\left[\mathbbm{1}((\hat{g}_j,\hat{h}_j)\neq (g_j,h_j))\Big|(\bsg,\bsh)=(\bsv,\bsu)\right]\\
&=\mathbb{Q}((g_{j},h_{j})=(g_c,h_c))\mathbb{P}_0((\hat{g}_j,\hat{h}_j)\neq (g_c,h_c))+\mathbb{Q}((g_{j},h_{j})=(a_c,b_c))\mathbb{P}_1((\hat{g}_j,\hat{h}_j)\neq (a_c,b_c))\\
&\geq\frac{1}{2}-\frac{1}{2}\|\mathbb{P}_0-\mathbb{P}_1\|_{\text{TV}}\\
&\geq\frac{1}{2}-\frac{1}{4}\sqrt{\mathbb{D}_\mathsf{KL}(\mathbb{P}_0,\mathbb{P}_1)}.
\end{split}
\eeq
where the second to the last inequality is by Le Cam's method and the last inequality is by Pinsker's inequality.\footnote{The total variation distance between probability distributions $P$ and $Q$ defined on a set $\mathcal{X}$ is defined as the maximum difference between probabilities they assign on subsets of $\mathcal{X}$: $\|P-Q\|_\text{TV}:=\sup_{\mathcal{A}\subset \mathcal{X}}|P(\mathcal{A})-Q(\mathcal{A})|$.}
 
Conditioned on $(g_j,h_j)$, the set of random variables $A_j:=\{A_{ij}: i\in[n]\}$ are independent of $\bsA\backslash A_j$ for both $\mathbb{P}_0$ and $\mathbb{P}_1$, and thus
\beq\label{eqn:conv2_3}
\mathbb{D}_\mathsf{KL}(\mathbb{P}_0,\mathbb{P}_1)=\mathbb{D}_\mathsf{KL}(\mathbb{P}_0(A_j),\mathbb{P}_1(A_j))+\mathbb{D}_\mathsf{KL}(\mathbb{P}_0(\bsA\backslash A_j),\mathbb{P}_1(\bsA\backslash A_j))=\mathbb{D}_\mathsf{KL}(\mathbb{P}_0(A_j),\mathbb{P}_1(A_j))
\eeq
where $\mathbb{P}(X)$ denote the distribution of $X$ with respect to the probability measure $\mathbb{P}$.
Given $(g_j,h_j)$, since $A_{1j},\dots, A_{nj}$ are independent, we can show that 
\beq
\begin{split}\label{eqn:conv2_4}
\mathbb{D}_\mathsf{KL}(\mathbb{P}_0(A_j),\mathbb{P}_1(A_j))&=\sum_{i=1}^n\mathbb{D}_\mathsf{KL}( \mathbb{P}_0(A_{ij}),\mathbb{P}_1(A_{ij}))\\
&=\sum_{i=1}^n\left((1-s)\log\frac{1-s}{1-s}+s \mathbb{D}_\mathsf{KL}\left(\boldsymbol{\mu}^{(i,j)}_{(g_c,h_c)},\boldsymbol{\mu}^{(i,j)}_{(a_c,b_c)}\right)\right)\\
&\geq s n \overline{D}.
\end{split}
\eeq

Combining \eqref{eqn:conv2_1}-- \eqref{eqn:conv2_4}, we have
\beq
\begin{split}
&\max_{\substack{(\bsv,\bsu)\in[K]^m\times [K]^m\\ v_j\neq u_j,\forall j[m]}}\E\left[\frac{1}{m}\sum_{j=1}^m\mathbbm{1}((\hat{g}_j,\hat{h}_j)\neq (g_j,h_j))\Big|(\bsg,\bsh)=(\bsv,\bsu)\right]\\
&\geq \frac{1}{2}-\frac{1}{4}\sqrt{sn\overline{D}}.
\end{split}
\eeq
Thus, if $s\leq \frac{1}{4n\overline{D}}$, then the above inequality is lower bounded by $3/8$.
This completes the proof.

\section{Useful Inequalities}

In this section, we summarize the useful inequalities used in the proof of the main results.

The following inequality, which appeared in \cite{bandeira2016sharp} provides a non-asymptotic spectral norm bound for random matrices with independent random entries. 
\begin{thm}[Spectral norm bound of a random matrice with independent entries]\label{lem:bd_spnorm}
Consider a random matrix $\bsX\in\mathbb{R}^{n\times m}$, whose entries are independently generated and obey
\beq
\E[X_{i,j}]=0,\quad\text{and}\quad |X_{i,j}|\leq B, \quad1\leq i\leq n,\;\;1\leq j\leq m.
\eeq
Define
\beq
\nu:=\max\left\{\max_i \sum_j\E[X_{i,j}^2],\;\;{\max_j \sum_i\E[X_{i,j}^2]}\right\}.
\eeq
Then there exists some universal constant $c>0$ such that for any $t>0$,
\beq\label{eqn:spec_norm_Ban}
\P\left\{\|\bsX\|\geq 4\sqrt{\nu}+t\right\}\leq (n+m)\exp\left(-\frac{t^2}{cB^2}\right).
\eeq
\end{thm}
We also present a useful corollary of Theorem \ref{lem:bd_spnorm}, which can be shown from \eqref{eqn:spec_norm_Ban} by setting $\tilde{c}=\sqrt{9c}$ and $t=B\sqrt{9c\log (n+m)}$.
\begin{cor}[Corollary of Theorem \ref{lem:bd_spnorm}]\label{cor:spect_norm}
If  $\E[X_{i,j}^2]\leq \sigma^2$ for all $i,j$ and satisfying conditions in Theorem \ref{lem:bd_spnorm}, then we have
\beq
\|\bsX\|\leq 4\sigma \sqrt{\max(m,n)}+\tilde{c}B\sqrt{\log (n+m)}
\eeq
with probability $1-(n+m)^{-8}$ for some constant $\tilde{c}>0$.
\end{cor}

We next summarize the eigenspace perturbation theory for asymmetric matrices with singular value composition (SVD).
Suppose $\bsX:=[\bsX_0,\bsX_1]$ and $\bsZ:=[\bsZ_0,\bsZ_1]$ are orthonormal matrices. 
When we define the distance between two subspaces $\bsX_0$ and $\bsZ_0$  by
\beq
\mathsf{dist}(\bsX_0,\bsZ_0):=\|\bsX_0\bsX_0^\top-\bsZ_0\bsZ_0^\top\|,
\eeq
then we have
\beq
\mathsf{dist}(\bsX_0,\bsZ_0)=\|\bsX_0^\top \bsZ_1\|=\|\bsZ_0^\top \bsX_1\|.
\eeq
Given $\|\bsX_0^\top\bsZ_0\|\leq 1$, we write SVD of $\bsX_0^\top \bsZ_0\in \mathbb{R}^{r\times r}$ as  $\bsX_0^\top \bsZ_0:=\bsU\cos \Theta\bsV^\top$ where $\cos\Theta=\text{diag}(\cos \theta_1,\dots, \cos\theta_r)$. We call $\{\theta_1,\dots,\theta_r\}$ principal angles between $\bsX_0$ and $\bsZ_0$. Then, we have
\beq
\|\bsX_0^\top \bsZ_1\|=\|\sin\Theta\|=\max\{|\sin\theta_1|,\cdots,|\sin\theta_r|\}.
\eeq

Let $\bsM^*$ and $\bsM=\bsM^*+\bsE$ be two matrices in $\mathbb{R}^{n\times m}$ with $n\leq m$, whose SVD are represented by $\bsM^*=\sum_{i=1}^n \sigma_i^* \bsu_i^*{\bsv_i^*}^\top$ and $\bsM=\sum_{i=1}^n \sigma_i \bsu_i{\bsv_i}^\top$, where $\sigma_1\geq \cdots \geq \sigma_n$ (resp.  $\sigma_1^*\geq \cdots \geq \sigma_n^*$). Let us define
\beq
\begin{split}
\bsU_0:=[\bsu_1,\cdots,\bsu_r]\in\mathbb{R}^{n\times r},\quad\bsV_0:=[\bsv_1,\cdots,\bsv_r]\in \mathbb{R}^{m\times r}.
\end{split}
\eeq
The matrices $\bsU_0^*$ and $\bsV_0^*$ are defined analogously.

\begin{thm}[Wedin $\sin\Theta$ Theorem]\label{thm:WedinSin} 
If $\|\bsE\|<\sigma_r^*-\sigma_{r+1}^*$, then one has
\beq
\max\{\|\mathsf{dist}(\bsU_0,\bsU_0^*)\|,\|\mathsf{dist}(\bsV_0,\bsV_0^*)\|\}\leq \frac{\sqrt{2}\|\bsE\|}{\sigma_r^*-\sigma_{r+1}^*-\|\bsE\|},
\eeq
where $\bsU_0^*$ ($\bsV_0^*$) and  $\bsU_0$ ($\bsV_0$)  are subspaces spanned by the largest $r$  left (right) singular vectors of $\bsM^*$ and $\bsM$, respecively. 
\end{thm}

Lastly, we also write down two useful concentration inequalities.
\begin{thm}[Hoeffding] Let $X_1,X_2,\dots,X_n$ be independent random variables such that $X_i\in[a_i,b_i]$ for $1\leq i\leq n$. Then, we have
\beq\label{eqn:Hoeffding}
\p\left[\left|\sum_{i=1}^n (X_i-\E[X_i])\right|>t\right]\leq 2 \exp\left(-\frac{2t^2}{\sum_{i=1}^n (b_i-a_i)^2}\right).
\eeq
\end{thm}

\begin{thm}[Bernstein]\label{thm:Bernstein} Let $X_1,X_2,\dots,X_n$ be independent random variables such that $X_i\in[a_i,b_i]$ for $1\leq i\leq n$. Let $C:=\max_{1\leq i\leq n}(b_i-a_i)$ and $\sigma^2=\sum_{i=1}^n \var(X_i)$.  Then we have
\beq\label{eqn:Bernstein}
\p\left[\left|\sum_{i=1}^n (X_i-\E[X_i])\right|>t\right]\leq 2 \exp\left(-\frac{t^2/2}{\sigma^2+C\cdot t/3}\right).
\eeq
\end{thm}

\end{document}